\documentclass[10pt,reqno]{article}
\usepackage[margin=1.00in]{geometry}
\usepackage{authblk}
\usepackage{comment}    
\usepackage{tikz}
\usepackage{lmodern}
\usepackage{bm}         
\usepackage{dsfont}     
\usepackage{braket}     
\usepackage{lipsum}     

\usepackage{amsmath,amssymb,amsthm}  
\usepackage{mathtools}               
\usepackage{esvect}                  
\usepackage{makecell}
\numberwithin{equation}{section}
\usepackage[ruled,vlined]{algorithm2e}
\DontPrintSemicolon
\usepackage{xcolor, graphicx}
\usepackage[caption=false]{subfig}
\usepackage{tikz}
\usetikzlibrary{shapes.geometric, arrows.meta, positioning, quantikz2}

\definecolor{darkgreen}{rgb}{0.0,0.72,0.06}
\definecolor{cerisepink}{rgb}{0.93,0.23,0.51}
\definecolor{red-brown}{rgb}{0.65,0.16,0.16}
\definecolor{blue}{rgb}{0,0,1}

\usepackage[colorlinks]{hyperref}
\usepackage[capitalize,noabbrev]{cleveref}
\hypersetup{
    colorlinks=true,
    citecolor=darkgreen,
    linkcolor=blue,
    filecolor=magenta,
    urlcolor=magenta
}

\usepackage[backend=biber,style=numeric-comp,sorting=none,url=false,maxbibnames=99,doi=false]{biblatex}
\theoremstyle{plain}
\newtheorem{thm}{Theorem}[section]
\newtheorem{lem}[thm]{Lemma}
\newtheorem{prop}[thm]{Proposition}
\newtheorem{cor}[thm]{Corollary}

\theoremstyle{definition}
\newtheorem{defn}[thm]{Definition}
\newtheorem{assumption}[thm]{Assumption}

\newtheorem{example}[thm]{Example}
\newtheorem{remark}[thm]{Remark}
\newtheorem{result}[thm]{Result}

\crefname{thm}{Theorem}{Theorems}
\crefname{lem}{Lemma}{Lemmas}
\crefname{prop}{Proposition}{Propositions}
\crefname{cor}{Corollary}{Corollaries}
\crefname{defn}{Definition}{Definitions}
\crefname{assumption}{Assumption}{Assumptions}
\crefname{notation}{Notation}{Notations}
\crefname{remark}{Remark}{Remarks}
\crefname{example}{Example}{Examples}
\crefname{fact}{Fact}{Facts}
\crefname{result}{Result}{Results}
\crefname{ques}{Question}{Questions}

\newcommand{\C}{\mathbb{C}}
\newcommand{\R}{\mathbb{R}}
\newcommand{\Z}{\mathbb{Z}}

\begin{document}

\title{Linear Combination of Hamiltonian Simulation with Commutator Scaling} 


\author[1]{Junaid Aftab\thanks{\href{mailto:junaida@umd.edu}{junaida@umd.edu}}} \author[2,3]{Dong An} \author[1]{Konstantina Trivisa} \affil[1]{Department of Mathematics, University of Maryland, College Park, USA} \affil[2]{Beijing International Center for Mathematical Research, Peking University, Beijing, China} \affil[3]{Joint Center for Quantum Information and Computer Science, University of Maryland, College Park, USA}

\date{\today}

\maketitle

\begin{abstract}
The Linear Combination of Hamiltonian Simulation (LCHS) framework simulates dissipative linear dynamics by representing time evolution as an integral over unitary operators, which is discretized by quadrature and implemented via Hamiltonian simulation. While existing analyses achieve near-optimal scaling in time and precision using norm-based quantities of the dissipative generator, we show that implementing the Hamiltonian simulation steps with Multi-Product Formulas (MPFs) yields commutator-sensitive error and complexity bounds. We demonstrate that the quadrature rule affects not only discretization error but also commutator structure and query complexity. This dependence is quantified through post-quadrature analysis for abstract MPF error profiles and for general time-independent and local Hamiltonians using known commutator-sensitive MPF error estimates. We compare uniform trapezoidal and free-scale sinh--sinh quadrature, showing improved quadrature-cardinality scaling for the latter, and illustrate the framework with applications to fractional diffusion, advection--diffusion, and open quantum systems.
\end{abstract}


\setcounter{tocdepth}{2} \tableofcontents

\section{Introduction}
The quantum simulation problem~\cite{feynman1982,feynman1986quantum} is central to quantum computing, with applications in quantum chemistry, materials science, and condensed matter physics (see, e.g.,~\cite{jordan2012quantum,bauer2020quantum,cao2019quantum,babbush2018low}). An important special case is the Hamiltonian simulation problem, which concerns the time evolution of a finite-dimensional quantum system governed by the Schr\"odinger equation
\begin{equation}\label{eqn:HamSim}
    i \frac{d}{dt} u(t) = H(t) u(t), \quad t \in [0,T],
\end{equation}
where \(u(t) \in \mathbb{C}^N\) and \(H(t) \in \mathbb{C}^{N\times N}\) is Hermitian for each \(t \in [0,T]\). Since \cref{eqn:HamSim} describes a finite-dimensional closed quantum system, its propagator is a unitary operator. We therefore refer to \cref{eqn:HamSim} as describing \emph{unitary dynamics}. In this setting, efficient quantum algorithms have been developed for both time-independent~\cite{lloyd1996quantumsimulation,berry2007efficient,BerryChilds2012,BerryChildsCleveEtAl2014,BerryChildsKothari2015,BerryChildsCleveEtAl2015,LowChuang2017,low2019wellconditioned,Childs2021theorytrotter,aftab2024mpf,mizuta2025commutatorscalinghamiltoniansimulation} and time-dependent~\cite{wiebe2010higher,low2019hamiltoniansimulationinteractionpicture,berry2020time,AnFangLin2021,anFang2022time,mizuta2024MPFtimedep} cases. However, these methods do not directly apply to general linear non-homogeneous differential equations of the form
\begin{equation}\label{non-unitary-ode}
    \frac{d}{dt} u(t) = -A(t)u(t) + b(t),
\end{equation}
where \(A(t) \in \mathbb{C}^{N \times N}\) may be non-anti-Hermitian and \(b(t) \neq 0\). Such equations arise in broader quantum simulation settings, notably in the simulation of finite-dimensional open quantum systems. We therefore refer to \cref{non-unitary-ode} as describing \emph{non-unitary dynamics}.

Several quantum algorithms have been developed for linear differential equations of the form~\cref{non-unitary-ode} (see, e.g.,~\cite{Berry2014,BerryChildsOstranderEtAl2017,childs2020quantum,Krovi2022,fang2023time,berry2024quantum,jin2024quantum,LowSu2026}. Most of the existing approaches are based on the Quantum Linear System Algorithm (QLSA)~\cite{harrow2009quantum,childs2017QLSA,Costa2023QLSA,Dalzell2024}, typically embedding \(A(t)\) into a larger dilated matrix and solving an associated enlarged linear system with additional implementation and complexity overhead. Recently, An et al.~\cite{an2023LCHS,an2023LCHS-Optimal} and Low and Somma~\cite{low2025optimallchs} introduced and analyzed the Linear Combination of Hamiltonian Simulation (LCHS) framework, which expresses the solution of~\cref{non-unitary-ode} as an integral over unitary evolutions. This approach simulates \emph{non-unitary dynamics} by combining Hamiltonian simulation algorithms developed for \emph{unitary dynamics}. The resulting algorithms achieve optimal state-preparation cost and optimal scaling in both the evolution time and the target precision. 

From a broader perspective, LCHS is a special case of the Linear Combination of Unitaries (LCU) framework of Childs and Wiebe~\cite{Childs2012lcu}, as it represents non-unitary dynamics through linear combinations of unitary simulation primitives.  A related LCU-based construction is the Multi-Product Formula (MPF)~\cite{low2019wellconditioned,aftab2024mpf,mizuta2024MPFtimedep,mizuta2025commutatorscalinghamiltoniansimulation}, combining Lie--Trotter--Suzuki product formulas~\cite{lloyd1996quantumsimulation} to obtain high-order Hamiltonian simulation algorithms.  Specifically, MPFs construct these high-order approximations by linearly combining sequences of low-order Trotter steps with carefully chosen weights.  This systematic cancellation of algorithmic errors circumvents the deep exponential nesting required by standard high-order product formulas, achieving near-optimal time and precision dependence while maintaining favorable commutator scaling from product formulas~\cite{Childs2021theorytrotter}.  However, for general non-unitary dynamics, existing approaches, including LCHS and other advanced algorithms, cannot currently achieve both of these desirable properties simultaneously.  This limitation naturally motivates the central question of our work: can we design an efficient quantum algorithm for general non-unitary systems with both near-optimal scaling (in time and precision) and commutator scaling?

In this work, we answer this question affirmatively in a quadrature-dependent sense. We develop an LCHS--MPF algorithm for non-unitary dynamics that applies to both time-independent and time-dependent settings, and illustrate its use in applications arising from linear differential equations and open quantum systems. In this algorithm, LCHS reduces non-unitary simulation to a family of Hamiltonian simulations, while MPFs implement these simulations with commutator-sensitive complexity. Our key observation is that the outer LCHS quadrature does not merely control discretization error. Rather, the quadrature nodes also determine the Hamiltonians passed to the inner simulation routine, and therefore affect the commutator structure, step-size restrictions, LCU normalization, weighted error accumulation, and final query complexity. 
We track this dependence explicitly and obtain a quantum algorithm for non-unitary simulation whose complexity reflects both the quadrature structure of LCHS and the commutator scaling of MPFs.

\subsection{Summary of Results}
We briefly summarize our approach and main results in the time-independent homogeneous case. The non-homogeneous and time-dependent cases are discussed in \cref{sec:extensions}. Let \(A=L+iH\), where \(L\succeq0\) and \(L,H\) are Hermitian. The LCHS representation expresses the dissipative evolution \(e^{-AT}\) as an integral over unitary evolutions
\begin{equation}\label{lchs-intro}
e^{-AT}\approx \frac{1}{\sqrt{2\pi}}\int_{\mathbb{R}} \hat f(k)e^{-iG_k T}\,dk.
\end{equation}
Here \(k\in\mathbb{R}\) is the frequency variable, $\hat f$ is the kernel function and \(G_k:=H+kL\) is the Hamiltonian simulated at frequency \(k\).  After applying a quadrature rule \(Q=\{(k_i,w_i):i\in\mathcal I_Q\}\) with nodes \(k_i\) and weights \(w_i\), the integral is replaced by a finite sum of the form
\begin{equation}
W_Q^{\operatorname{ideal}} = \sum_{i\in\mathcal I_Q}v_i e^{-iG_{k_i}T}, \qquad v_i = \frac{w_i\hat f(k_i)}{\sqrt{2\pi}}.
\end{equation}
Thus, the quadrature rule determines not only the discretization error but also the Hamiltonians that must be simulated. We replace each \(e^{-iG_{k_i}T}\) with an MPF approximation. The post-quadrature error and complexity bounds are then determined by the quadrature radius \(R_Q := \max_{i\in\mathcal I_Q}|k_i|\), which controls the range of Hamiltonians \(H+kL\) appearing in the simulation, and by the LCU normalization factor \(\alpha_Q := \sum_{i\in\mathcal I_Q}|v_i|\), which controls the success probability and amplitude-amplification overhead of the outer LCU procedure. Since LCHS is combined with Hamiltonian simulation via  MPF, the cost also depends on a commutator scale \(\mu_Q\), which quantifies the MPF simulation difficulty of the quadrature-indexed Hamiltonian family \(\{G_{k_i}\}_{i\in \mathcal I_Q}\). Our first result is a post-quadrature complexity estimate that makes this quadrature-dependent commutator scaling explicit.
The following result is stated in a simplified informal form. More precise complexity estimates are given in \cref{abstract-fixed-order-arbitrary-quadrature} and \cref{abstract-optimized-order-arbitrary-quadrature}. 

\begin{result}[Informal Version of \cref{abstract-optimized-order-arbitrary-quadrature}]  \label{informal-result-1} Let $\epsilon > 0$. Suppose the kernel, quadrature rule, and MPF order are chosen so that the final implemented operator \(W_Q^{\mathrm{MPF}}\) satisfies $\left\|W_Q^{\mathrm{MPF}}-e^{-AT}\right\|\leq \epsilon$. The quadrature-indexed Hamiltonian family \(\{G_{k_i}\}_{i\in \mathcal I_Q}\) admits a commutator-sensitive MPF bound with scale \(\mu_Q\). Choosing MPF order $m=\mathcal O \left(\log(T/\epsilon)\right)$ gives block-encoding complexity
\begin{equation}
\widetilde{\mathcal O}\left( \mu_Q T F \left(\log(T/\epsilon)\right) \log^2(T/\epsilon) \right)
\end{equation}
and normalized-state preparation complexity
\begin{equation}
\widetilde{\mathcal{O}} \left( \alpha_Q \frac{\|u(0)\|}{\|u(T)\|} \mu_Q T F \left(\log(T/\epsilon)\right) \log^2(T/\epsilon) \right).
\end{equation}
Here \(F(m)\) records the order-dependence in the MPF estimate. 
\end{result}

The role of \(\alpha_Q\) in \cref{informal-result-1} is important. In previous LCHS algorithms, the kernel and quadrature rule are typically fixed, so the corresponding LCU normalization is either bounded explicitly or absorbed into constants. In contrast, we keep \(\alpha_Q\) explicit because the kernel and quadrature choices also affect the commutator scale. Consequently, its effect must be evaluated jointly with \(\alpha_Q\) and \(\mu_Q\). We remark that \(\alpha_Q\) and \(\mu_Q\) may also scale as \(\mathcal{O}(\operatorname{polylog}(1/\epsilon))\) for certain quadrature rules. We will state these scalings explicitly in the later discussions.

In \cref{informal-result-1}, \(F(m)\) and \(\mu_Q\) are auxiliary functions that upper bound \(\mu_{m,Q}\), which is defined in \cref{abstract-mu-mQ-def}, by \(\mu_{m,Q} \leq F(m) \mu_Q\). Here, we only note that \(\mu_{m,Q}\) depends on the quantities
\begin{equation} 
\Lambda_{m,Q} := \sum_{i\in\mathcal I_Q}|v_i|\Phi_m(k_i,N) \qquad \Phi_{m,Q}^{\ast} := \max_{i\in\mathcal I_Q}\Phi_m(k_i,N),
\end{equation} 
where \(m\) denotes the MPF order and \(\Phi_m(k,N)\) is an abstract MPF error profile for simulating \(G_k\).  Hence, \cref{informal-result-1} is abstract because it is phrased in terms of a general MPF error profile. In \cref{spec-aftab,spec-mizuta}, we specialize this result to two concrete settings, which we summarize below:

\begin{enumerate} 

\item In \cite{aftab2024mpf}, Aftab et al. provided the first complete error analysis for the time-independent MPF. Informally, their analysis associates to each pair \((H,L)\) a commutator scale, denoted here by \(\chi_J(H,L)\), which measures how rapidly the relevant nested commutators can grow. For a quadrature rule \(Q\), this scale must be compared with the quadrature radius \(R_Q\). Therefore, we choose \(\rho_Q\) satisfying 
\begin{equation}
0<\rho_Q<\frac{\chi_J(H,L)}{(J+1)^{1/J}\max\{1,R_Q\}}.
\end{equation}
In this case, we have 
\begin{align} \Phi_{m,\nu}(H,L) :={}& \sum_{\substack{j\in 2\mathbb Z_+\\ j\ge 2m}} \sum_{l=1}^{m} \frac{\rho_Q^{\,j+l-(2m+1)}}{l!} \sum_{\substack{j_1,\dots,j_l\in 2\mathbb Z_+\\ j_1+\cdots+j_l=j}} \sum_{\substack{\ell_1,\dots,\ell_l\ge0\\ 0\le \ell_\kappa\le j_\kappa+1\\ \ell_1+\cdots+\ell_l=\nu}} \prod_{\kappa=1}^{l} 
\mathcal{C}_{j_\kappa+1,\ell_\kappa}(H,L)
\end{align} 
for \(\nu\ge0\). Here, \(\mathcal{C}_{j_\kappa+1,\ell_\kappa}(H,L)\) is determined by nested commutators of length \(j_\kappa+1\) in \(H\) and \(L\) that contain exactly \(\ell_\kappa\) occurrences of \(L\). If the discrete moments are $M_\nu^Q := \sum_{i\in\mathcal I_Q}|v_i|\,|k_i|^\nu$ then $\Lambda_{m,Q} := \sum_{\nu\ge0}\Phi_{m,\nu}(H,L)M_\nu^Q$. In this case, $F(m) \equiv 1$.

\item In \cite{mizuta2025commutatorscalinghamiltoniansimulation}, Mizuta proved an improved MPF estimate of Aftab et al. for local and extensive Hamiltonians. We use this estimate when it applies to \(\{G_{k_i}\}_{i\in\mathcal I_Q}\). For \(p_0\in\mathbb N\), define \begin{equation} \mu_{m,p_0}(G_{k_i}) := \sup_{\substack{j,l\in\mathbb N\\ j\ge 2m,\ 1\le l\le \lfloor j/2\rfloor}} \left( \sum_{\substack{2\le j_1,\ldots,j_l\le p_0-1\\ j_1+\cdots+j_l=j}} \prod_{\kappa=1}^{l} \alpha_{\operatorname{comm},j_\kappa+1}(G_{k_i}) \right)^{\frac{1}{j+l}}, \end{equation} where \(\alpha_{\operatorname{comm},j}(G_{k_i})\) denotes the \(j\)-th order nested-commutator.  The locality-based estimate considers the quantities $\Lambda_{m,Q}(p_0) := \sum_{i\in\mathcal I_Q} |v_i|\, \mu_{m,p_0}(G_{k_i})^{2m+1}$ and  $\Phi_{m,Q}^{\ast}(p_0) := \max_{i\in\mathcal I_Q} \mu_{m,p_0}(G_{k_i})^{2m+1}$.  In this case, $F(m) = \mathcal{O}(\log^2 m)$. \end{enumerate}

\cref{tab:comparison} compares representative algorithms in the time-independent homogeneous setting. Compared with optimal LCHS, our algorithm replaces the norm-based Hamiltonian simulation cost with a commutator-sensitive MPF cost after quadrature. Thus, the method is useful when the family \(H+kL\) has favorable commutator structure. The complexity also depends on quadrature-induced quantities such as \(R_Q\) and \(\alpha_Q\). Hence, the method is not a black-box improvement over LCHS. Rather, it is a structure-sensitive refinement. The tradeoff is that the dependence on \(T\) and \(1/\epsilon\) is not linear and logarithmic, respectively, as in optimal LCHS. Therefore, our result should be viewed as near-optimal in regimes where the commutator scaling also enters the complexity analysis favorably. 

Our second result concerns the choice of quadrature rule. Because the quadrature rule determines the nodes \(k_i\) and weights \(v_i\), it affects both the discretization error and the MPF cost through the quantities \(R_Q\), \(\alpha_Q\), the discrete moments \(M_\nu^Q\), and the resulting weighted MPF profiles. We therefore compare quadrature rules in terms of these quantities, together with the quadrature error \(E_{\operatorname{quad}}(Q)\) and the node count \(|\mathcal I_Q|\).

\begin{result}[Informal Version of \cref{trap-budget} and \cref{sinh-budget}]
For the uniform trapezoidal rule, if the mesh and tail-error budgets are balanced as \(\epsilon_{\operatorname{mesh}} \asymp \epsilon_{\operatorname{tail}} \asymp \epsilon\), then the rule can be chosen so that \(E_{\operatorname{quad}}\!\left(Q^{\operatorname{trap}}\right)=\mathcal O(\epsilon)\) with
\begin{equation}
\left|\mathcal I_{Q^{\operatorname{trap}}}\right|
= \mathcal O\!\left(\log^{2}\!\left(\frac1{\epsilon}\right)\right),
\qquad
R_{Q^{\operatorname{trap}}}
= \mathcal O\!\left(\log\!\left(\frac1{\epsilon}\right)\right).
\end{equation}
For the sinh--sinh rule, the rule can be chosen so that \(E_{\operatorname{quad}}\!\left(Q^{\operatorname{sinh}}\right)=\mathcal O(\epsilon)\) with
\begin{equation}
\left|\mathcal I_{Q^{\operatorname{sinh}}}\right|
= \mathcal O\!\left(
\log\!\left(\frac1{\epsilon}\right)
\log\log\!\left(\frac1{\epsilon}\right)
\right)
\end{equation}
Both rules have the same leading quadrature-radius scaling and the same leading discrete-moment scaling. In particular, for the four-parameter \((a,b,c,d)\) kernel family used in
\cref{low-somma-kernel}, the LCU normalization satisfies
\begin{equation}
\alpha_Q
=
\begin{cases}
\mathcal O(1), & a>1,\\[1mm]
\mathcal O\!\left(\log\log(1/\epsilon)\right), & a=1,
\end{cases}
\end{equation}
provided that \(a,b,d=\mathcal O(1)\) and $c= \Omega  (\sqrt{\log(1/\epsilon)} )$. 
\end{result}

\begin{table}[t]
{
    \renewcommand{\arraystretch}{1.50}
    \centering
    \begin{tabular}{c|c|c}\hline\hline
        \textbf{Method} 
        & \textbf{Query Complexity} 
        & \textbf{Queries to $\ket{u(0)}$} 
        \\\hline 

        \makecell{Spectral method~\cite{childs2020quantum}} 
        & 
        $\widetilde{\mathcal{O}}\left( 
        \frac{\|u(0)\|}{\|u(T)\|} 
        \kappa_V \alpha_A T\,
        \operatorname{poly}\left(\log\left(\frac{1}{\epsilon}\right)\right)
        \right)$  
        & 
        $\widetilde{\mathcal{O}}\left( 
        \frac{\|u(0)\|}{\|u(T)\|} 
        \kappa_V \alpha_A T\,
        \operatorname{poly}\left(\log\left(\frac{1}{\epsilon}\right)\right)
        \right)$ 
        \\\hline

        \makecell{Truncated Dyson~\cite{berry2024quantum}}  
        & 
        $\widetilde{\mathcal{O}}\left( 
        \frac{\|u(0)\|}{\|u(T)\|} 
        \alpha_A T 
        \left(\log\left(\frac{1}{\epsilon}\right)\right)^2
        \right)$ 
        & 
        $\mathcal{O}\left( 
        \frac{\|u(0)\|}{\|u(T)\|} 
        \alpha_A T 
        \log\left(\frac{1}{\epsilon}\right) 
        \right)$ 
        \\\hline


        \makecell{Original LCHS~\cite{an2023LCHS}} 
        & 
        $\widetilde{\mathcal{O}}\left( 
        \left(\frac{\|u(0)\|}{\|u(T)\|}\right)^2 
        \alpha_A T /\epsilon 
        \right)$ 
        & 
        $\mathcal{O}\left( 
        \frac{\|u(0)\|}{\|u(T)\|} 
        \right)$ 
        \\\hline

        \makecell{Improved LCHS~\cite{an2023LCHS-Optimal}}
        & 
        $\widetilde{\mathcal{O}}\left( 
        \frac{\|u(0)\|}{\|u(T)\|} 
        \alpha_A T 
        \left(\log\left(\frac{1}{\epsilon}\right)\right)^{1+1/\beta} 
        \right)$  
        & 
        $\mathcal{O}\left( 
        \frac{\|u(0)\|}{\|u(T)\|} 
        \right)$  
        \\\hline

        \makecell{Optimal LCHS~\cite{low2025optimallchs}}
        & 
        $\mathcal{O}\left( 
        \frac{\|u(0)\|}{\|u(T)\|} 
        \alpha_A T 
        \log\left(\frac{1}{\epsilon}\right) 
        \right)$  
        & 
        $\mathcal{O}\left( 
        \frac{\|u(0)\|}{\|u(T)\|} 
        \right)$  
        \\\hline



        \makecell{LCHS--MPF} 
        & 
        $\widetilde{\mathcal{O}}\left( 
        \alpha_Q\frac{\|u(0)\|}{\|u(T)\|} 
        \mu_Q T 
        ~\text{poly}\log\left(\frac{1}{\epsilon}\right)
        \right)$  
        & 
        $\mathcal{O}\left( 
        \alpha_Q \frac{ \|u(0)\|}{\|u(T)\|} 
        \right)$  
        \\\hline\hline
    \end{tabular}
}
    \caption{Comparison of representative quantum algorithms for time-independent homogeneous linear differential equations. Here, \(\alpha_A\) denotes a norm or block-encoding normalization of \(A\), and \(\mu_Q\) is the quadrature-dependent commutator scale.
    }
    \label{tab:comparison}
\end{table}

Thus, the sinh--sinh quadrature rule improves the quadrature-cardinality scaling due to its double-exponential decay rate. This reduction in the number of quadrature nodes directly improves the LCU implementation cost of the algorithm. 

We further demonstrate the utility of the LCHS-MPF framework through three concrete classes of applications.

\begin{enumerate}
    \item Fractional diffusion equations with imaginary potentials, including fractional Bloch--Torrey-type dynamics, model anomalous diffusion in complex media and arise in transport and magnetic-resonance-type applications~\cite{Torrey1956,Magin2008,Yu2013,BuenoOrovioBurrage2017,Moutal2020}. After discretization, $A=L+iH$, where $L$ comes from the fractional diffusion operator and $H$ from the potential, and the LCHS--MPF analysis for $G_k=H+kL$ gives
    \begin{equation}
        \mu_Q=\mathcal{O} (A_V^{(N)} (1+R_QN^{\theta_s/d} ) ).
    \end{equation}
    Here $N$ is the number of grid points, $d$ is the spatial dimension, $0<s<1$ is the fractional order, $\theta_s=\max\{0,2s-1\}$, $R_Q$ is the quadrature radius, and $A_V^{(N)}$ measures the relevant discrete Fourier regularity of the potential. Hence, the complexity reflects the diffusion--potential commutator structure rather than only a worst-case norm of $A$.

    \item Advection--diffusion equations model heat transfer, mass transport, pollutant dispersion, and related transport phenomena~\cite{ParhiziEtAlCDRS,AgudAlbesaEtAlCDR,VillotaCadenaEtAlCDRSolv,an2026fast}. In the discretized system, diffusion gives the dissipative part and advection gives the Hamiltonian part, and for $G_k=H+kL$ the commutator scale is
    \begin{equation}
        \mu_Q=\mathcal{O} (|b|N^{1/d}+aR_QN^{2/d} ).
    \end{equation}
    Here $a>0$ is the diffusion coefficient, $b$ is the advection strength, $N$ is the number of grid points and $d$ is the dimension. This separates the effects of diffusion, advection, discretization, and quadrature, giving a sharper estimate than standard norm-based LCHS bounds.

    \item No-jump evolution in dissipative transverse-field Ising systems with local spontaneous emission arises in quantum-trajectory descriptions of open many-body systems~\cite{DalibardCastinMolmer1992,PlenioKnight1998,LeeChan2013,RobertsClerk2023}. The effective non-Hermitian Hamiltonian yields $\dot{\psi}=-(L+iH)\psi$, the quadrature Hamiltonians $G_k=H+kL$ remain $2$-local, and in the large-system regime the leading commutator scale is
    \begin{equation}
\mu_Q=\mathcal{O} \left( \alpha_Q^{3/2}|V|^{1/2}\left(|h|+D|J|+\frac{R_Q\gamma}{2}\right)^{3/2}\right).
    \end{equation}
    Here $|V|$ is the number of spins, $h$ is the transverse-field strength, $J$ is the Ising coupling, $D$ is the maximum graph degree and  $\gamma$ is the local emission rate. Thus, the complexity is controlled by Pauli commutator structure rather than only by a generic norm or block-encoding normalization, exposing local many-body structure via commutator scaling.
\end{enumerate}

\subsection{Related Works}
Several recent works study when non-unitary quantum simulation can inherit structural advantages from unitary Hamiltonian simulation. Wang et al.~\cite{WangZhouWangZhengZhangLi2026Lindbladian} recently derived commutator-based Trotter error bounds for Lindbladian simulation. Their work is close in motivation to ours. Both analyses ask whether the favorable commutator scaling of Hamiltonian product formulas persists in dissipative dynamics. The technical settings, however, are distinct. Their bounds apply directly to product-formula approximations of Lindblad generators and use Richardson extrapolation for observable estimation, whereas our approach first uses the LCHS representation to reduce dissipative linear evolution to a quadrature-indexed family of Hamiltonian simulations. Thus, the commutators appearing in our analysis are those of the Hamiltonians generated by the LCHS quadrature, and the main issue is to track how this quadrature dependence enters the MPF error and complexity estimates.

A second related direction concerns quantum algorithms for non-unitary dynamics based on matrix-function and eigenvalue-transformation representations. An et al.~\cite{AnChildsLinYing2024LaplaceLCHS} extended the LCHS approach from matrix exponentials to broader matrix eigenvalue transformations represented through Laplace transforms. Low and Su~\cite{LowSu2024QEP} introduced quantum eigenvalue processing, an eigenvalue transformation framework for non-normal matrices that addresses related non-Hermitian phenomena from a different algorithmic perspective. These works complement the present paper by broadening the class of matrix functions that can be implemented on quantum computers. By contrast, the focus of the present work is not to enlarge the class of admissible matrix functions, but to analyze how the post-quadrature LCHS representation interacts with a commutator-sensitive Hamiltonian simulation primitive.

Contour-integral methods provide another point of comparison. Takahira et al.~\cite{TakahiraOhashiSogabeUsuda2021Contour} used Cauchy's integral formula and block encoding to implement matrix functions as linear combinations of shifted matrix inverses. More recently, Jiang and An~\cite{JiangAn2026ContourQET} analyzed the complexity of contour-integral-based quantum eigenvalue transformations, with applications to Hamiltonian simulation, matrix polynomials, and linear differential equations. Wang et al.~\cite{WangLiuXueZhuangDouChenGuo2025CBMD} developed a contour-based matrix-decomposition framework for non-unitary dynamics using a finite decomposition derived from Cauchy's residue theorem. These approaches are related to LCHS through analytic matrix-function representations, However, their elementary operations are typically shifted resolvents or Hermitian transformations, rather than Hamiltonian evolutions of the form \(e^{-i(H+kL)t}\) whose commutator structure is tracked in the present work.

Other transform-based approaches have also been developed in the literature. \emph{Schr\"odingerization}~\cite{jin2023quantum,jin2024quantum,JinLiuMaYu2025Schrodingerization,jin2025quantumalgorithmsstochasticdifferential} converts non-unitary linear dynamics into unitary dynamics in one higher dimension, achieving optimal or near-optimal matrix-query dependence after suitable smooth initialization. Jin et al.~\cite{JinMaZuazua2026Transmutation} proposed a transmutation-based method for dissipative diffusion generated by positive semi-definite operators, expressing the diffusion semigroup via the Kannai transform as a Gaussian-weighted superposition of unitary wave propagators. A recent Poisson-summation framework by Wang et al.~\cite{WangZhuangDouChenGuo2026PSF} further emphasizes the role of spectral aliasing and discretization in non-unitary matrix transformations. These works illustrate that non-unitary simulation can often be reduced to superpositions of unitary or resolvent-type primitives.

\subsection{Discussion and Open Questions}
Our main result shows that the query complexity is polylogarithmic in time and inverse precision, and at the same time achieves commutator scaling. This matches our previous analysis of commutator scaling for MPF. It remains an open question whether one can obtain optimal scaling in both time and inverse precision while preserving commutator scaling.

Another natural question is whether the quadrature dependence can be further optimized beyond the quadrature rules analyzed here. Although the present analysis accommodates fairly general quadrature rules, the resulting complexity bounds still depend on quantities such as \(\Lambda_{m,Q}\), \(R_Q\), \(\alpha_Q\), and the discrete moments \(M^Q_\nu\). It would be interesting to determine whether there exist quadrature rules that simultaneously improve discretization accuracy, LCU normalization, moment growth, and the post-quadrature MPF cost. A related question is how to optimize the LCHS kernel profile jointly with the quadrature rule. The kernel parameters influence the approximation error, truncation error, quadrature radius, LCU normalization, and the discrete moment bounds entering the MPF error profiles. A more systematic kernel--quadrature optimization may therefore improve the overall complexity of the algorithm.

The LCU implementation introduces an overhead governed by the normalization of the linear combination and by the norm of the desired output state. Although amplitude amplification can boost the success probability, the resulting state-preparation cost still contains normalization-dependent factors. It would be interesting to determine whether alternative block-encoding constructions, quantum signal processing techniques, or dissipative embeddings can reduce this overhead while preserving the commutator-sensitive advantages of the LCHS--MPF framework.

Several extensions of the time-dependent framework remain open. The analysis in \cref{sec:time-dependent-extension} assumes sufficient regularity of the time-dependent generator on each subinterval. It would be interesting to determine whether similar complexity guarantees can be obtained under weaker assumptions, such as piecewise-smooth, discontinuous, or rapidly varying dissipative generators within the LCHS-compatible setting.

Finally, it would be useful to better understand the practical performance of MPF-based dissipative simulation. The present work focuses primarily on asymptotic complexity bounds, but concrete implementations may exhibit additional structure that is not visible at the worst-case level. Numerical studies could clarify which parameter regimes are most favorable for near-term or fault-tolerant quantum architectures.

\subsection{Organization}
The remainder of the paper is organized as follows.
\cref{sec:prelims} reviews the notation and algorithmic preliminaries. \cref{sec:alg-err-analysis} analyzes the pre-quadrature approximation, truncation, and inner-simulation errors for the time-independent homogeneous case. \cref{complexity-analysis} develops the post-quadrature LCHS--MPF error and complexity analysis, including the abstract theorem, commutator-scaling specializations, and comparison of quadrature rules.  \cref{applications} discusses various applications.

\subsection*{Acknowledgments} 
Junaid Aftab acknowledges the support by the National Science Foundation under the grant DMS-2231533.
Dong An acknowledges funding from Quantum Science and Technology - National Science and Technology Major Project via Project 2024ZD0301900, and the support by The Fundamental Research Funds for the Central Universities, Peking University.
Konstantina Trivisa acknowledges support from the National Science Foundation under grants DMS-2008568 and DMS-2231533.

\section{Preliminaries}\label{sec:prelims}
This section collects the preliminary details used throughout the paper. \cref{sec:notation} introduces the notation, \cref{sec:product} reviews product formulas, \cref{sec:lcu} recalls the Linear Combination of Unitaries (LCU) framework, \cref{sec:mpf} describes Multi-Product Formulas (MPFs), and \cref{sec:lchs} introduces the Linear Combination of Hamiltonian Simulation (LCHS) framework.

\subsection{Notation}\label{sec:notation}
We briefly discuss the main notation used throughout this work.

\subsubsection{Standard Notation} 
Let $\mathbb{R}$, $\mathbb{C}$, and $\mathbb{N}$ denote the sets of real, complex and natural numbers, respectively. We write \(\mathbb{Z}\) for the integers and \(\mathbb{Z}_{+}\) for the positive integers. For \(M \in \mathbb{N}\), let \(\Z^d_M := \{0,\ldots,M-1\}^d\) be the $d$-fold Cartesian product of the cyclic group, $\Z_M$. All logarithms are natural logarithms unless otherwise stated. If $f, g : \mathbb{N} \to \mathbb{R}^+$ are non-negative functions, we use the following standard notation from complexity theory:
\begin{enumerate}
\item $f(n) = \mathcal{O}(g(n))$ if and only if there exists a constant $C > 0$ and an integer $M \in \mathbb{N}$ such that $f(n) \leq C g(n)$ for all $n \geq M$. The notation $\widetilde{\mathcal O}$ hides polylogarithmic factors.
\item $f(n) = \Omega(g(n))$ if and only if there exists a constant $C > 0$ and an integer $M \in \mathbb{N}$ such that $f(n) \geq C g(n)$ for all $n \geq M$.
\end{enumerate}
The symbol \(\overleftarrow{\prod_{\gamma}}\) denotes a product in which the elements are arranged with increasing indices from right to left. Similarly, the notation \(\overrightarrow{\prod_{\gamma}}\) is defined to denote a product with elements ordered with increasing indices from left to right. Specifically, we have
\begin{equation}
\overset{\longleftarrow}{\prod_{i=1,\cdots,k}} A_i = A_k \cdots A_1, \quad \overset{\longrightarrow}{\prod_{i=1,\cdots,k}} A_i = A_1 \cdots A_k.
\end{equation}
$\| f \|_{L^1[a,b]}$ denotes the $L^1$-norm, defined as $\|f\|_{L^1[a,b]} = \int_a^b |f(s)| \, ds$, on $[a,b]$. We use the convention $\hat{f}(\xi) = \frac{1}{\sqrt{2\pi}} \int_{\mathbb{R}} f(x) e^{-i \xi x} dx$ and $f(x) = \frac{1}{\sqrt{2\pi}} \int_{\mathbb{R}} \hat{f}(\xi) e^{i \xi x} d\xi$ for the Fourier transform and the inverse Fourier transform, respectively.

\subsubsection{Linear Algebra Notation}
All vector spaces considered are finite-dimensional and defined over \(\mathbb{C}\). Vectors are denoted either by lowercase Roman letters  with a \(\vec{\cdot}\) symbol on top or by ket notation, depending on context. Linear operators—represented as matrices—are denoted by uppercase Roman letters. The linear algebra notation used throughout this paper is summarized below:
\begin{enumerate}
\item \(A^\dagger\) denotes the conjugate transpose of \(A\) and \(I\) denotes the identity matrix of appropriate dimension.
\item $\| \vec a \|_1$ denotes the $1$-norm of a vector~$\vec a$, defined as $\| \vec a \|_1 = \sum_{i=1}^n |a_i|$.
\item $\| A \|$ denotes the spectral norm of a matrix~$A$, defined as its largest singular value.
\item $\| A \|_{L^1[a,b]}$ denotes the $L^1$-norm of a time dependent matrix $A(t)$ on $[a,b]$, defined as $\int_a^b \|A(s)\| \, ds$.
\item The notation $A \succeq 0$ denotes that $A$ is a positive semi-definite matrix.
\item For two matrices $A$ and $B$, the commutator is defined by $[A,B] = AB - BA$. More generally, for matrices $A_1, A_2, \ldots, A_m$, we define the nested commutator
\begin{equation}
[A_1, A_2, \ldots, A_m] := [A_1, [A_2, \ldots, [A_{m-1}, A_m] \cdots ]].
\end{equation}
In the special case where $A_1 = \cdots = A_{m-1} = A$ and $A_m = B$, this nested commutator is denoted by $\operatorname{ad}_A^{\,m-1}(B)$.
\end{enumerate}
The formal solution to the matrix-valued linear differential equation $\frac{dU(t)}{dt} = -i H(t) U(t)$ for $T_0 \leq t \leq T$ is written as $U(T_0,T) = \exp_{\mathcal{T}} ( -i \int_{T_0}^T H(s) ds )$,  where $\exp_{\mathcal{T}}$ denotes the time-ordered exponential operator. We also denote the time-ordered exponential operator as $\mathcal{T}e$. The reader is referred to~\cite{dollard1984product} for more details. In what follows, we set $T_0 = 0$ and write $U(T) := U(0,T)$ when convenient.

\subsection{Product Formulas}\label{sec:product}
Consider a time-dependent Hamiltonian \(H(t)\) for \(t \in [0,T]\), decomposed as \(H(t)=\sum_{\gamma=1}^{\Gamma} H_{\gamma}(t)\). For each \(\gamma\), let \(U_{\gamma}(T)\) denote the time-ordered exponential generated by \(H_{\gamma}(t)\) on \([0,T]\). Product-formula methods assume that each \(U_{\gamma}(T)\) can be implemented efficiently. Under this assumption, a generic product formula approximates the exact time-evolution operator \(U(T)\) as
\begin{equation}\label{generic-prod}
U_p(T) = \prod_{\xi=1}^{\Xi} \prod_{\gamma=1}^{\Gamma} U_{\pi_\xi(\gamma)}(\beta_{\xi \gamma} T,(\beta_{\xi \gamma} + \alpha_{\xi \gamma})T),
\end{equation}
where $\pi_\xi$ is a permutation in $S_{\Gamma}$ and $\alpha_{\xi,\gamma},\beta_{\xi,\gamma}\in[0,1]$ such that $\beta_{\xi \gamma} + \alpha_{\xi \gamma} \in [0,1]$. We say that $U_p$ is a $p$th-order product formula if its single-step approximation error satisfies
\begin{equation}
\|U_p(T)-U(T)\| = \mathcal{O}(T^{p+1}).
\end{equation}
In practice, the interval \([0,T]\) is partitioned into \(r\) segments, and the evolution on each segment is approximated by a product formula. The parameter \(r\) is referred to as the Trotter number. The resulting approximation is \((U_p(T/r))^r\).

\begin{example}
Product formulas can be constructed using the Suzuki recursion~\cite{suzuki1985decomposition,suzuki1991general}, which yields, in particular, the following first-order and second-order product formulas
\begin{equation}
U_1(0,T) = \overset{\longleftarrow}{\prod_{\gamma=1,\cdots,\Gamma}} U_\gamma(0,T), \quad U_2(0,T) = \overset{\longrightarrow}{\prod_{\gamma=1,\cdots,\Gamma}} U_\gamma(T/2,T) \overset{\longleftarrow}{\prod_{\gamma=1,\cdots,\Gamma}} U_\gamma(0,T/2).
\end{equation}
We do not discuss higher-order product formulas in any detail, as these are not used in this work.
\end{example}

When \(H(t) \equiv H\) is time-independent, the error scaling of product formulas is well understood. In particular, \cite{Childs2021theorytrotter} showed that, for a general \(p\)th-order product formula, ensuring the global error bound \(\lVert U_p(T/r)^r - U(T) \rVert \le \epsilon\) requires choosing \(r\) such that
\begin{equation}\label{product-error}
r = \mathcal{O} \left( \frac{\alpha_{\operatorname{comm},p+1}(H)^{1/p} T^{1+1/p}}{\epsilon^{1/p}} \right), \quad \alpha_{\operatorname{comm},p+1}(H) = \sum_{j_1,\dots,j_{p+1}=1}^{\Gamma} \|[H_{j_1},\dots,H_{j_{p+1}}]\|,
\end{equation}
where $\epsilon$ is the target precision. The dependence on \(\alpha_{\operatorname{comm},p+1}(H)\) is commonly referred to as \emph{commutator scaling}. This property enables product-formula Hamiltonian simulation algorithms to achieve favorable system-size scaling for physically local Hamiltonians (see, e.g.,~\cite{babbush2015chemical,wecker2015solving,childs2018toward}). By contrast, the error scaling of product formulas for time-dependent Hamiltonians is less completely understood. Early work of Wiebe et al.~\cite{wiebe2010higher} derived bounds for smoothly varying Hamiltonians based on \(1\)-norm scaling. Subsequent studies~\cite{berry2020time,an2021time,ikeda2023minimum} analyzed low-order time-dependent product formulas, some of which exhibit commutator scaling. A general treatment was later given by Childs et al.~\cite{childs2019optimalproduct}, who established commutator scaling for smoothly varying Hamiltonians with finite-range interactions. More recently, Mizuta et al.~\cite{mizuta2024MPFtimedep} analyzed generic smoothly varying time-dependent product formulas and proved commutator-scaling bounds. In this work, we use only their results on time-dependent multi-product formulas, which are recalled in \cref{sec:mpf}.

\begin{remark}\label{prod-2-conv}
Throughout the remainder of the paper, we restrict attention to second-order product formulas, i.e., \(p=2\).\footnote{Consequently, the symbol \(p\) will be used freely for other purposes.}
\end{remark}

\subsection{Linear Combination of Unitaries}\label{sec:lcu}
The Linear Combination of Unitaries (LCU) algorithm~\cite{Childs2012lcu} provides a method for implementing linear combinations of unitary operators on a quantum computer. Let \(M \in \mathbb{N}\), let \(a_1,\ldots,a_M \in \mathbb{C}\), and let \(U_1,\ldots,U_M\) be unitary operators. The LCU algorithm aims to implement the operator \(U = \sum_{j=1}^{M} a_j U_j\). The algorithm assumes that two operators can be efficiently implemented. First, it requires access to state preparation (\(\operatorname{PREP}\)) oracles:
\begin{align}\label{F-LCU}
\operatorname{PREP}_{\operatorname{R}} \ket{0}  = \frac{1}{\sqrt{\| \vec a \|_1 }} \sum_{j=1}^{M}  \sqrt{a_{j}}  \ket{j}, \quad \operatorname{PREP}_{\operatorname{L}} \ket{0}  = \frac{1}{\sqrt{\| \vec a \|_1 }} \sum_{j=1}^{M}  \overline{\sqrt{a_{j}}} \ket{j}
\end{align}
where \(\sqrt{\cdot}\) denotes the principal branch of the square root. Second, it assumes access to a multi-qubit controlled operation, commonly called a select oracle, of the form
\begin{equation}\label{controlled-LCU}
    \operatorname{SEL}=\sum_{j=1}^{M} U_j \otimes \ket{j}\bra{j}.
\end{equation} 
Assuming access to the required operators, the LCU algorithm implements the unitary operator $U_{\text{LCU}} = (I \otimes \operatorname{PREP}^\dagger_{\operatorname{L}}) \cdot \operatorname{SEL} \cdot (I \otimes \operatorname{PREP}_{\operatorname{R}})$. If \(\ket{\psi}\) denotes an input state, we have
\begin{equation} U_{\text{LCU}} \ket{\psi} \ket{0}^{\log M} = \frac{1}{\| \vec a \|_1} \left(\sum_{j=1}^{M} a_j U_j\right) \ket{\psi} \ket{0}^{\log M} + \ket{\perp},
\end{equation}
where \(\ket{\perp}\) denotes a potentially non-normalized state satisfying \((I \otimes \ket{0}^{\log M}\bra{0}^{\log M}) \ket{\perp} = 0\). Constructing \(U_{\operatorname{LCU}}\) requires a single query to each of the oracles \(\operatorname{PREP}^\dagger_{\operatorname{L}}\), \(\operatorname{PREP}_{\operatorname{R}}\), and \(\operatorname{SEL}\), with the implementation of \(\operatorname{SEL}\) typically dominating the computational cost. The ancilla register, initially prepared in the state \(\ket{0}^{\log M}\), must be measured to obtain \(\sum_{j=1}^{M} a_j U_j\ket{\psi}\). To achieve a constant success probability, \(\mathcal{O}(\|\vec{a}\|_1)\) rounds of amplitude amplification~\cite{childs2015lcutaylor} are required. 

\subsection{Multi-Product Formula}\label{sec:mpf}
The Multi-Product Formula (MPF) for time-independent Hamiltonian simulation constructs a linear combination of time-independent product formulas
\begin{equation}\label{multiproduct}
    U_{\operatorname{MP}}(T) = \sum_{j=1}^{M} a_j \,  U^{b_j}_{2}(T/b_j),
    \quad a_j\in \mathbb{R}, b_j \in \mathbb N.
\end{equation}
Low, Kliuchnikov, and Wiebe~\cite{low2019wellconditioned} showed that, for \(M = m\), the coefficients \(a_j\) and exponents \(b_j\) can be chosen such that $\|\vec{a}\|_1 = \mathcal{O}(\log m)$ and $\|\vec{b}\|_1 = \mathcal{O}(m^2 \log m)$, ensuring that  $\|U_{\operatorname{MP}}(T) - e^{-i H T} \| = \mathcal{O}(T^{2m+1})$. Consequently, although an MPF is constructed from product formulas of fixed order, it can achieve arbitrarily high orders of convergence. It can be implemented using the LCU algorithm as follows:
\begin{enumerate}
\item Consider implementing the MPF over a short time-step, $\Delta$. Define the state-preparation oracles as in \cref{F-LCU}, and let the select oracle
\begin{equation}
\operatorname{SEL} = \sum_{j=0}^{m-1} U_2\left(\Delta/b_{j+1}\right)^{b_{j+1}} \otimes \ket{j}\bra{j} 
\end{equation}
encode the MPF. Then $(I \otimes \operatorname{PREP}^\dagger_{\operatorname{L}}) \cdot \operatorname{SEL} \cdot (I \otimes \operatorname{PREP}_{\operatorname{R}})$ implements the MPF via LCU, with the ancilla state $\ket{0}$ encoding the target operation and initial success probability $\mathcal{O}(1/\|\vec{a}\|_1^2)$.
\item Consider simulating evolution over total time $T$. Partition $[0,T]$ into $r$ equal segments. On each segment, apply the MPF for $\Delta=T/r$ using LCU. Robust oblivious amplitude amplification is applied at each step to maintain a constant success probability.
\end{enumerate}

The cost of the implementation above is dominated by the oracle \(\operatorname{SEL}\), which uses \(\mathcal{O}(\|\vec{b}\|_1)\) queries to controlled-\(U_2\) operations. Amplitude amplification contributes an additional factor of \(\mathcal{O}(\|\vec{a}\|_1)\). Thus, each time step requires \(\mathcal{O}(\|\vec{a}\|_1\|\vec{b}\|_1)\) queries to \(U_2\), and \(r\) steps require \(\mathcal{O}(r\|\vec{a}\|_1\|\vec{b}\|_1)\) queries in total. Although \(\|\vec{a}\|_1\) and \(\|\vec{b}\|_1\) are specified in \cref{sec:mpf}, a rigorous choice of \(r\) was previously unavailable and was expected to exhibit commutator scaling because the construction uses \(U_2\). This gap was recently closed by Aftab et al.~\cite{aftab2024mpf} and Mizuta~\cite{mizuta2025commutatorscalinghamiltoniansimulation}. We recall the relevant short-time simulation results below, with time-step size denoted by \(\Delta\).

\begin{prop}[Theorem 8 in \cite{aftab2024mpf}]
\label{aftab-time-indep-short}
Let $H=\sum_{\gamma=1}^{\Gamma} H_\gamma$ be a time-independent Hamiltonian. Suppose there exists a  $J \geq 1$ such that $\inf_{j \geq J} \alpha_{\operatorname{comm},j}^{-1/j}(H) > 0$. If $\Delta > 0$ is chosen such that $\Delta \leq  \inf_{j \geq J} \alpha_{\operatorname{comm},j}^{-1/j}(H)$, then we have
        \begin{equation}
            \|U_{\operatorname{MP}} (\Delta) - U(\Delta)\| \leq \|\vec{a}\|_1 \sum_{\substack{j \in 2\mathbb{Z}_+ \\ j \geq 2m}} \sum_{l=1}^{m}  \frac{\Delta^{j+l}}{l!} \left(\sum_{\substack{j_1,\cdots,j_l \in 2\mathbb{Z}_+, \\ j_1+\cdots+j_l=j}} \left(\prod_{\kappa=1}^{l}  \alpha_{\operatorname{comm},j_{\kappa}+1}(H)\right) \right). 
        \end{equation}
\end{prop}

While \cref{aftab-time-indep-short} applies to general time-independent Hamiltonians, it has the limitation that the infimum is taken over commutators of arbitrarily large depth. If \(\alpha_{\operatorname{comm},j}=\Omega(N^{j^\gamma})\) for some \(\gamma>1\), then this infimum vanishes, and no positive admissible \(\Delta\) exists. Mizuta~\cite{mizuta2025commutatorscalinghamiltoniansimulation} addressed this issue for \(q\)-local, \(g\)-extensive Hamiltonians by showing that only commutators up to finite depth need to be considered in the short-time MPF bound. If $H=\sum_{\gamma=1}^{\Gamma} H_\gamma$, let $\gamma$, let $X_\gamma$ denote the support of $H_\gamma$. We say that $H$ is $q$-local if $|X_\gamma|\le q$ for all $\gamma$ and $[H_\gamma,H_{\gamma'}]=0$ whenever $X_\gamma\cap X_{\gamma'}=\varnothing$. Moreover, we say that $H$ is $g$-extensive if $ \sum_{\gamma:\, j\in X_\gamma}\|H_\gamma\| \le g$ for all $j$. 

\begin{prop}[Theorem 4 in \cite{mizuta2025commutatorscalinghamiltoniansimulation}]
\label{mizuta-time-indep-short}
Let $H=\sum_{\gamma=1}^{\Gamma} H_\gamma$ be a $q$-local, $g$-extensive Hamiltonian. Fix $m \ge 1$.  For a truncation parameter $p_0 \in \mathbb{N}$, define
\begin{equation}\label{mu-m-p0}
\mu_{m,p_0} := \sup_{\substack{j,l \in \mathbb{N} \\ j \ge 2m,\; 1 \le l \le \left\lfloor j/2 \right\rfloor}} \left( \sum_{\substack{2 \le j_1,\ldots,j_l \le p_0-1 \\ j_1+\cdots+j_l = j}} \prod_{\kappa=1}^{l} \alpha_{\mathrm{comm},\,j_\kappa+1}(H) \right)^{\frac{1}{j+l}}.
\end{equation}
Let $\epsilon \in (0,1)$, $N \geq 1$ and let $p_0 := p_0(N,\epsilon) = \lceil \log(3N/\epsilon) \rceil $. Choose a $\Delta > 0$ such that
\begin{equation}
0 < \Delta \le \min\left\{ \frac{1}{16 e^3 p_0 q g}, \;\frac{1}{4 \mu_{m,p_0}} \right\}.
\label{eq:mizuta_tau_condition}
\end{equation}
Then we have
\begin{equation}
\left\| U_{\operatorname{MP}} (\Delta) - U(\Delta) \right\| \le 2 e^{1/2} \|\vec{a} \|_1 ( \mu_{m,p_0} \Delta )^{2m+1} + \|\vec a\|_1 \|\vec b\|_1 \epsilon.
\label{eq:mizuta_error}
\end{equation}
\end{prop}

The Multi-Product Formula for time-dependent Hamiltonian simulation, constructed from the second-order time-dependent product formula, is given by
\begin{equation}\label{time-dep-2MPF}
U_{\operatorname{MP}}(T) = \sum_{j=1}^{m} a_j \overset{\longleftarrow}{\prod_{b=1,\cdots,b_j}} U_{2}\left( \frac{(b - 1)T}{b_j}, \frac{bT}{b_j} \right).
\end{equation}
An error analysis was recently given by Mizuta et al.~\cite{mizuta2024MPFtimedep}, who showed that time-dependent MPFs also exhibit commutator scaling. In particular, they proved the following bound for MPFs constructed from second-order time-dependent Trotter--Suzuki formulas.

\begin{prop}[Specialization of Theorem 17 in \cite{mizuta2024MPFtimedep}]
\label{prop:mizuta-time-dep-short}
Let $H(t)=\sum_{\gamma=1}^{\Gamma} H_\gamma(t)$ be a smooth time-dependent $q$-local, $g$-extensive Hamiltonian\footnote{In the time-dependent case, we say that $H$ is $g$-extensive if $\sup_{\tau\in[0,\Delta]}\sum_{\gamma:\,j\in X_\gamma}\|H_\gamma(\tau)\|\le g$ for all $j$. }.
Moreover, assume there exists an $f>0$ such that
\begin{equation}
\sup_{\tau\in[0,\Delta]} \sum_{\gamma:\,j\in X_\gamma} \left\| \frac{d^n}{d\tau^n}H_\gamma(\tau) \right\| \le f^n g
\end{equation}
for all $j$ and all $n\ge 1$. Fix $m\ge 1$. For $1\le \gamma\le \Gamma+1$, define
\begin{equation}
D_\gamma(\tau) := \begin{cases} \operatorname{ad}_{H_\gamma(\tau)}+ i\dfrac{d}{d\tau}, & 1\le \gamma\le \Gamma,\\ \Gamma\dfrac{d}{d\tau}, & \gamma=\Gamma+1. \end{cases}
\end{equation}
For $p\ge 1$, define $\alpha_{\operatorname{comm},p}(\tau) := \sum_{\gamma_0=1}^{\Gamma} \sum_{\gamma_1,\ldots,\gamma_p=1}^{\Gamma+1} \left\| \left( D_{\gamma_1}(\tau)\cdots D_{\gamma_p}(\tau)H_{\gamma_0}(\tau) \right) H_{\gamma_0}(\tau) \right\|$. 
For a truncation parameter $p_0\in\mathbb{N}$, define
$\mu_{p_0} := \sup_{\tau\in[0,\Delta]} \max_{\substack{p\in\mathbb{N}\\2\le p\le p_0}} \left( \alpha_{\operatorname{comm},p}(\tau) \right)^{\frac{1}{p+1}}$. 
Let $\epsilon\in(0,1)$, $N\ge 1$, and let $p_0:=p_0(N,\epsilon)=\lceil\log(2N/\epsilon)\rceil$.
Assume $p_0\ge 2$ and choose a $\Delta>0$ such that
\begin{equation}
0<\Delta \le \min\left\{ \frac{1}{8e^3p_0(2qg+2\Gamma f)},\;\frac{1}{8\mu_{p_0}} \right\}.
\label{eq:td-mizuta-step-condition-p2}
\end{equation}
Then we have
\begin{equation}
\left\| U_{\operatorname{MP}}(\Delta)-U(\Delta) \right\| \;\le\; \|\vec a\|_1 (4\mu_{p_0}\Delta)^{2m+1} + \|\vec a\|_1\|\vec b\|_1\epsilon.
\label{eq:td-mizuta-mpf-error-p2}
\end{equation}
\end{prop}

\subsection{Linear Combination of Hamiltonian Simulation}\label{sec:lchs}
The Linear Combination of Hamiltonian Simulation (LCHS) algorithm simulates \emph{non-unitary dynamics} by expressing them as linear combinations of \emph{unitary dynamics}. 
LCHS formulation was first introduced in~\cite{an2023LCHS} and improved in~\cite{an2023LCHS-Optimal}. 
This approach generalizes the Fourier representation of the exponential function on the positive real axis,
\begin{equation}\label{fourier-one}
e^{-x}=\int_{\mathbb{R}} \hat f(k) e^{-ikx} \, dk,  \quad x \geq 0. 
\end{equation}
Specifically, for a possibly time-dependent operator $A(t) = L(t) + i H(t)$, with Hermitian $L(t) \succeq 0$ and $H(t)$, they show that 
\begin{equation}
\mathcal{T}e^{-\int_0^t A(s)\,ds} =\int_{\mathbb{R}} \hat f(k) U_k(t)dk, \quad t \in [0,T]
\end{equation}
solves \cref{non-unitary-ode} in the case $b(t) \equiv 0$. 
The case $b(t) \neq 0$ is then derived via an application of Duhamel's principle. 
The LCHS algorithm then proceeds by truncating the integral to $[-K,K]$ for some $K > 0$ and implementing it via the LCU algorithm, where each
\begin{equation}\label{Uk}
U_k(t) := \mathcal{T}  e^{-i \int_0^t (H(s) + k L(s)) \, ds}
\end{equation}
is implemented via a time-dependent Hamiltonian simulation algorithm.  The function $\hat f(k)$ is called the LCHS kernel function and should be the inverse Fourier transform of the function $f(x)$ which takes the value of $e^{-x}$ for $x \geq 0$. 
Previous works have suggested the choices of $\hat f(k)$ to be $\frac{1}{\pi(1+k^2)}$~\cite{an2023LCHS} or $\frac{1}{2 \pi e^{-2^\beta} e^{(1+iz)^\beta}}, \beta \in (0,1)$~\cite{an2023LCHS-Optimal}, where the LCHS algorithm with the latter kernel function can solve \cref{non-unitary-ode} in near-optimal time and precision complexity. Recently, Low and Somma~\cite{low2025optimallchs} introduced an approximate LCHS framework that further improves upon the previous formalism. 
Their key observation is that the requirement that \cref{fourier-one} hold exactly can be relaxed so that it need only hold approximately. That is, it suffices to find a function $f$ such that, for $\epsilon > 0$, we have $|e^{-x} - f(x)| \le \epsilon$ for $x \ge 0$, and such that \cref{fourier-one} holds only approximately. 
This weaker condition admits a broader class of kernels and leads to a flexible and more efficient LCHS algorithm with optimal complexity. 
Their main result is summarized below.

\begin{prop}[Theorem 1 in~\cite{low2025optimallchs}]
\label{somma-lchs}
For $T \geq 0$ and $n \geq 1$, let $A(t) = L(t) + iH(t) \in \C^{2^n \times 2^n}$ be a time-dependent matrix for $t \in [0,T]$, where $L(t)$ and $H(t)$ are Hermitian operators such that $\|L\|_{L^1[0,T]}$ and $\|H\|_{L^1[0,T]}$ are finite. Consider  the uniform strip
\begin{equation}
S_{[-y_0,0]} := \{ z \in \mathbb{C} : \operatorname{Im}(z) \in [-y_0,0] \}
\end{equation}
for some $y_0 > 1$. On $S_{[-y_0,0]}$, assume that the complex-valued kernel function $\hat{f}(z)$ is such that  $\lim_{|z| \to \infty} \hat{f}(z) = 0$, $(z+i)\hat{f}(z)$ is analytic and the residue at $z=-i$ is $\operatorname{Res}(\hat{f}, -i) = i/\sqrt{2\pi}$. For $t \in [0,T]$, we have
\begin{equation}\label{low-somma-thm-1}
\left \| \frac{1}{\sqrt{2 \pi}} \int_{-R}^R \hat f(k) U_k(t) dk - \mathcal{T}e^{-\int_0^t A(s)\,ds} \right \| \leq \frac{1}{\sqrt{2 \pi}} \int_{\R \setminus [-R,R]} | \hat f(k) | dk + \frac{1}{\sqrt{2 \pi}} \int_{\R} | \hat f(k-iy_0) | dk.
\end{equation}
\end{prop}

In what follows, we use the improved approximate formulation of Low and Somma~\cite{low2025optimallchs}.

\section{Pre-Quadrature Error Analysis}\label{sec:alg-err-analysis}
In this section, we give a pre-quadrature error analysis of the algorithm in the time-independent homogeneous case. Throughout, let \(A=L+iH\), where \(L\) and \(H\) are Hermitian and \(L\succeq 0\). In the pre-quadrature regime, the algorithm has three main sources of error:

\begin{enumerate}
\item (Approximation) The approximation error, \(E_{\mathrm{approx}}(y_0)\), arising by replacing \(e^{-AT}\) with the LCHS representation. That is, $E_{\mathrm{approx}}(y_0) :=  \| \frac{1}{\sqrt{2\pi}}\int_{\mathbb R}\hat f(k)U_k(T)\,dk - e^{-AT}  \|$. By \cref{low-somma-thm-1}, after taking the limit \(R\to\infty\), it is bounded by \(\frac{1}{\sqrt{2\pi}}\int_{\mathbb{R}}|\hat f(k-iy_0)|\,dk\).

\item (Truncation) The truncation error, \(E_{\mathrm{trunc}}(R)\), arising from restricting \(\frac{1}{\sqrt{2\pi}}\int_{\mathbb{R}}\hat f(k)\,U_k(t)\,dk\) to \([-R,R]\). That is, $E_{\mathrm{trunc}}(R) :=\| \frac{1}{\sqrt{2\pi}}\int_{\mathbb R}\hat f(k)U_k(T)\,dk - \frac{1}{\sqrt{2\pi}}\int_{-R}^{R}\hat f(k)U_k(T)\,dk \|$. By \cref{low-somma-thm-1}, since $U_k(T)$ is unitary, it is bounded by \(\frac{1}{\sqrt{2\pi}}\int_{\mathbb{R}\setminus[-R,R]}|\hat f(k)|\,dk\).

\item (Inner Simulation) The simulation error \(E_{\mathrm{MPF}}(T)\), arising from replacing each \(U_k(T)\) by its MPF approximation \(\widetilde U_k(T)\).
\end{enumerate}

In what follows, we use the four-parameter kernel family introduced by Low and Somma~\cite{low2025optimallchs}:
\begin{equation}\label{low-somma-kernel}
\hat{f}_{a,b}(k;c,d) = \frac{(b+1)^{a-1}}{\sqrt{2\pi}} \frac{e^{d(1-ik)} e^{-\frac{k^2+1}{4c^2}}}{(1-ik)(b+ik)^{a-1}}, \quad a \in \mathbb{N},\; b,c>0,\; d\in\mathbb{R}.
\end{equation}
We call \(\vec{\theta}=(a,b,c,d)\) the kernel profile and write \(\hat f(k)\) for \(\hat f_{a,b}(k;c,d)\) when clear. 
\cref{approx-error-time-indept} discusses the approximation error, \cref{truncation-error} the truncation error, \cref{inner-sim-subsub} the inner simulation error, and \cref{coupling-subsub} discusses the coupling between the truncation and inner simulation errors.

\subsection{Approximation Error}\label{approx-error-time-indept}
We first analyze the approximation error \(E_{\mathrm{approx}}(y_0)\). Low and Somma analyzed \(E_{\mathrm{approx}}(y_0)\) for the subclass \(\hat{f}_{2,1}(k;c,d)\). Here, we analyze \(E_{\mathrm{approx}}(y_0)\) for the full family \(\hat{f}_{a,b}(k;c,d)\). Most proofs are relegated to \cref{app-alg-err}.

\begin{lem}\label{approx-simplify-integ}
Fix a kernel profile \(\vec{\theta} = (a,b,c,d)\) and \(y_0 > 1\). We have
\begin{align}
E_{\mathrm{approx}}(y_0) &\leq \frac{(b+1)^{a-1}}{2\pi}e^{\,d(1-y_0)+\frac{y_0^2-1}{4c^2}}\int_{\mathbb{R}}\frac{e^{-k^2/(4c^2)}\,dk}{\sqrt{(y_0-1)^2+k^2}\,((b+y_0)^2+k^2)^{(a-1)/2}}.
\label{simplified-eq}
\end{align}
\end{lem}

The proof of \cref{approx-simplify-integ} can be found in \cref{proof-of-approx-simplify-integ}. The integral in \cref{approx-simplify-integ} has an exact representation in terms of the confluent hypergeometric function, but that form is less transparent. We instead use the bound \(((b+y_0)^2+k^2)^{(a-1)/2}\geq (b+y_0)^{a-1}\) to obtain the simpler upper bound
\begin{equation}\label{bessel-integ}
E_{\mathrm{approx}}(y_0) \le \frac{(b+1)^{a-1}}{(b+y_0)^{a-1}}\frac{e^{\,d(1-y_0)+\frac{y_0^2-1}{4c^2}}}{2\pi}\int_{\mathbb{R}}\frac{e^{-k^2/(4c^2)}}{\sqrt{(y_0-1)^2+k^2}}\,dk.
\end{equation}
This bound is explicit and retains the essential dependence on \(y_0\) and \(a,b,c,d\). \cref{bessel-lemma}, which is proved in \cref{proof-of-bessel-lemma}, evaluates this integral form.

\begin{lem}\label{bessel-lemma}
For \(s,\sigma>0\), we have
\begin{equation}\label{bessel-eq}
\int_{\mathbb{R}}\frac{e^{-k^2/(4\sigma^2)}}{\sqrt{k^2+s^2}}\,dk = e^{s^2/(8\sigma^2)}K_0\left(\frac{s^2}{8\sigma^2}\right),
\end{equation}
where \(K_0\) denotes the modified Bessel function of the second kind.
\end{lem}

Using \cref{bessel-lemma} with \(s=y_0-1\), \(\sigma=c\) and \(K_0(x)\leq \sqrt{\frac{\pi}{2x}}e^{-x}\) for \(x>0\), we obtain
\begin{equation}\label{surrogate}
E_{\mathrm{approx}}(y_0) \le \frac{c}{\sqrt{\pi}}\frac{(b+1)^{a-1}}{(y_0-1)(b+y_0)^{a-1}}e^{d(1-y_0)+\frac{y_0^2-1}{4c^2}}.
\end{equation}
\cref{surrogate} motivates minimizing the right-hand side and identifying kernel profiles for which this bound is at most \(\epsilon_{\operatorname{approx}}\). In \cref{minimize-surrogate-approx}, which is proved in \cref{proof-of-minimize-surrogate-approx} , we first study existence and uniqueness of minimizers with respect to \(y_0\).

\begin{lem}\label{minimize-surrogate-approx}
Fix a kernel profile \(\vec{\theta} = (a,b,c,d)\). Define the surrogate cost functional
\begin{equation}
B(y_0) := \frac{(b+1)^{a-1}}{(y_0-1)(b+y_0)^{a-1}}e^{d(1-y_0)+\frac{y_0^2-1}{4c^2}}, \quad y_0 > 1.
\end{equation}
Then $\frac{d}{dy_0}\log B(y_0) = -d+\frac{y_0}{2c^2}-\frac{1}{y_0-1}-\frac{a-1}{b+y_0}$. 
Consequently, any stationary point \(y_0^*\) of \(\log B\) satisfies
\begin{equation}\label{surrogate-stat-eq}
\frac{y_0^*}{2c^2}-d = \frac{1}{y_0^*-1}+\frac{a-1}{b+y_0^*}.
\end{equation}
Moreover, we have $\frac{d^2}{dy_0^2}\log B(y_0) = \frac{1}{2c^2}+\frac{1}{(y_0-1)^2}+\frac{a-1}{(b+y_0)^2}>0$. 
Therefore, \(\log B\) is strictly convex on \((1,\infty)\) and has at most one stationary point. Any such stationary point is the unique global minimizer of \(B\). Moreover, \(B\) admits a global minimizer.
\end{lem}

\cref{minimize-surrogate-approx} implies that, for a fixed kernel profile \(\vec{\theta}=(a,b,c,d)\), the surrogate cost \(B(y_0)\) admits a unique minimizer \(y_0^\ast\). If \(\epsilon_{\operatorname{approx}}>0\) denotes the approximation error budget, then
\begin{equation}\label{suff-from-surr}
\text{there exists } y_0 \in (1,\infty) \text{ such that } \frac{c}{\sqrt{\pi}}B(y_0)\le \epsilon_{\operatorname{approx}} \iff \frac{c}{\sqrt{\pi}}B(y_0^\ast)\le \epsilon_{\operatorname{approx}}.
\end{equation}
\cref{suff-from-surr} gives a sufficient condition for \(E_{\operatorname{approx}}(y_0)\leq \epsilon_{\operatorname{approx}}\). It remains to determine when \(y_0^\ast\) satisfies \(\frac{c}{\sqrt{\pi}}B(y_0^\ast)\leq \epsilon_{\operatorname{approx}}\). \cref{feasible-kernel-approx} gives a condition in terms of the parameter \(d\).

\begin{prop}\label{feasible-kernel-approx}
Let \(\epsilon_{\operatorname{approx}} > 0\). For \(s > 1\), define
\begin{align}
\Phi(s) &:= \frac{s}{2c^2}-\frac{1}{s-1}-\frac{a-1}{b+s}, \\
\Psi(s) &:= \frac{c}{\sqrt{\pi}(s-1)}\left(\frac{b+1}{b+s}\right)^{a-1}\exp\left(1-\frac{(s-1)^2}{4c^2}+\frac{(a-1)(s-1)}{b+s}\right).
\end{align}
\(\Phi\) is strictly increasing on \((1,\infty)\), \(\Psi\) is strictly decreasing on \((1,\infty)\), and there exists a unique \(s_\epsilon\in(1,\infty)\) such that \(\Psi(s_\epsilon)=\epsilon_{\operatorname{approx}}\). Moreover, for a fixed kernel profile \(\vec{\theta}=(a,b,c,d)\), let \(y_0^\ast\) be the unique minimizer of \(B\) from \cref{minimize-surrogate-approx}. Then
\begin{equation}
\frac{c}{\sqrt{\pi}}B(y_0^\ast)\le \epsilon_{\operatorname{approx}} \iff d \ge \Phi(s_\epsilon).
\end{equation}
\end{prop}

\begin{proof}
By \cref{minimize-surrogate-approx}, \(B\) admits a unique minimizer \(y_0^\ast\in(1,\infty)\) satisfying \(d=\Phi(y_0^\ast)\). From \cref{inc-phi}, \(\Phi'(s)>0\), so \(\Phi\) is strictly increasing on \((1,\infty)\). Moreover,
\begin{equation}
\Phi(s)\to -\infty \quad \text{as } s\to 1^+, \quad \Phi(s)\to +\infty \quad \text{as } s\to \infty.
\end{equation}
Hence, \(\Phi\) is a bijection from \((1,\infty)\) onto \(\mathbb{R}\). Using \(d=\Phi(y_0^\ast)\) in \(\frac{c}{\sqrt{\pi}}B(y_0^\ast)\) gives \(\frac{c}{\sqrt{\pi}}B(y_0^\ast)=\Psi(y_0^\ast)\). We now show that \(\Psi\) is strictly decreasing. Differentiating \(\log\Psi\) gives
\begin{equation}
\frac{d}{ds}\log \Psi(s) = -\frac{1}{s-1}-\frac{s-1}{2c^2}-\frac{(a-1)(s-1)}{(b+s)^2}<0.
\end{equation}
Thus \(\Psi\) is strictly decreasing on \((1,\infty)\). Moreover, we have
\begin{equation}
\Psi(s)\to +\infty \quad \text{as } s\to 1^+, \quad \Psi(s)\to 0 \quad \text{as } s\to \infty.
\end{equation}
Hence, there exists a unique \(s_\epsilon\in(1,\infty)\) such that \(\Psi(s_\epsilon)=\epsilon_{\operatorname{approx}}\). Finally, since \(\Psi\) is strictly decreasing and \(d=\Phi(y_0^\ast)\), with \(\Phi\) strictly increasing, we obtain
\begin{align}
\frac{c}{\sqrt{\pi}}B(y_0^\ast)\le \epsilon_{\operatorname{approx}} &\iff \Psi(y_0^\ast)\le \epsilon_{\operatorname{approx}} \iff y_0^\ast \ge s_\epsilon \iff d=\Phi(y_0^\ast)\ge \Phi(s_\epsilon).
\end{align}
This completes the proof.
\end{proof}

\cref{feasible-kernel-approx} shows that taking \(d\) above an explicit lower bound ensures \(E_{\mathrm{approx}}(y_0)\leq \epsilon_{\mathrm{approx}}\). The threshold is determined by the unique solution of \(\Psi(s)=\epsilon_{\mathrm{approx}}\), which can be computed numerically. \cref{feasible-kernel-profile-explicit}, which is proved in \cref{proof-of-feasible-kernel-profile-explicit}, gives explicit sufficient conditions on $c$ and \(d\) and determines a corresponding choice of \(y_0\).

\begin{cor}\label{feasible-kernel-profile-explicit}
Let \(\epsilon_{\operatorname{approx}} \in (0,1)\) and let \(\vec{\theta}=(a,b,c,d)\) be a kernel profile. Then the following statements hold.

\begin{enumerate}
\item   Define \(\omega_\epsilon := W(e^{2a}/(2\pi\epsilon_{\operatorname{approx}}^2))\) and \(\widetilde x_0 := 1+c\sqrt{2\omega_\epsilon}\), where \(W\) denotes the Lambert \(W\)-function. If
\begin{equation}\label{eq:lambertW-d-threshold}
d \ge \frac{\widetilde x_0}{2c^2}-\frac{1}{\widetilde x_0-1}-\frac{a-1}{b+\widetilde x_0},
\end{equation}
then there exists \(y_0>1\) such that \(\frac{c}{\sqrt{\pi}}B(y_0)\le \epsilon_{\operatorname{approx}}\). Consequently, \(E_{\operatorname{approx}}(y_0)\le \epsilon_{\operatorname{approx}}\).

\item   Define \(\overline x_0 := 1+2c\sqrt{a+\log(1/\epsilon_{\operatorname{approx}})}\). If
\begin{equation}\label{eq:elementary-d-threshold}
d \ge \frac{\overline x_0}{2c^2}-\frac{1}{\overline x_0-1}-\frac{a-1}{b+\overline x_0},
\end{equation}
then there exists \(y_0>1\) such that \(\frac{c}{\sqrt{\pi}}B(y_0)\le \epsilon_{\operatorname{approx}}\). Consequently, \(E_{\operatorname{approx}}(y_0)\le \epsilon_{\operatorname{approx}}\).

\item The sufficient lower bound on \(d\) in \cref{eq:elementary-d-threshold} is obtained by evaluating \(\Phi(s)\) at \(\overline x_0\). In particular,
\begin{equation}
\Phi(\overline x_0) = \frac{1}{2c^2}+\frac{1}{c}\sqrt{a+\log(1/\epsilon_{\operatorname{approx}})}-\frac{1}{\overline x_0-1}-\frac{a-1}{b+\overline x_0}.
\end{equation}
For fixed \(b\) and \(c\), as either \(a\to\infty\) or \(\epsilon_{\operatorname{approx}}\to 0\), or both, we have
\begin{equation}
\Phi(\overline x_0) = \frac{1}{2c^2}+\mathcal O\left(\frac{1}{c}\sqrt{a+\log(1/\epsilon_{\operatorname{approx}})}\right).
\end{equation}
Therefore, a sufficient lower bound on \(d\) grows at most on the order of \(\frac{1}{2c^2}+\frac{1}{c}\sqrt{a+\log(1/\epsilon_{\operatorname{approx}})}\).
\end{enumerate}
\end{cor}

\subsection{Truncation Error}\label{truncation-error}
We now compute the truncation error in \cref{trunc-error}.

\begin{lem}\label{trunc-error}
Fix a kernel profile \(\vec\theta=(a,b,c,d)\) and \(R>0\). We have
\begin{align}
E_{\mathrm{trunc}}(R) \leq \frac{c}{\sqrt{\pi}}\left(\frac{b+1}{b}\right)^{a-1}e^{d-\frac{1}{4c^2}}\operatorname{erfc}\left(\frac{R}{2c}\right).
\label{trunc-comp}
\end{align}
\end{lem}

The proof of \cref{trunc-error} can be founnd in \cref{proof-of-trunc-error}. The bound in \cref{trunc-error} shows that \(E_{\mathrm{trunc}}(R)\to 0\) as \(R\to\infty\). This motivates considering the set of radii for which the truncation error is below a fixed tolerance \(\epsilon'>0\)
\begin{equation}\label{g-def}
\mathcal H_{\epsilon'}:=\{R>0:E_{\mathrm{trunc}}(R)<\epsilon'\}.
\end{equation}
\cref{radii-lower-bd} gives a computable lower bound ensuring \(R\in\mathcal H_{\epsilon'}\).

\begin{cor}\label{radii-lower-bd}
Let \(\epsilon'>0\). The set \(\mathcal H_{\epsilon'}\) is non-empty for every \(\epsilon'>0\). For a kernel profile \(\vec{\theta}=(a,b,c,d)\), define
\begin{equation}
\vartheta(\vec\theta) := \frac{\sqrt{\pi}}{c}\left(\frac{b}{b+1}\right)^{a-1}\frac{\epsilon'}{e^{d-\frac{1}{4c^2}}}.
\end{equation}
If \(\vartheta(\vec\theta)<1\), then every \(R>2c\,\operatorname{erfc}^{-1}(\vartheta(\vec\theta))\) is in \(\mathcal H_{\epsilon'}\). If \(\vartheta(\vec\theta)\ge 1\), then every \(R>0\) is in \(\mathcal H_{\epsilon'}\).
\end{cor}

\begin{proof}
\(\mathcal H_{\epsilon'}\neq\emptyset\) since \(\lim_{R\to\infty}E_{\mathrm{trunc}}(R)=0\) by \cref{trunc-error}. If \(\vartheta(\vec{\theta})<1\), then \cref{trunc-error} shows that \(E_{\mathrm{trunc}}(R)<\epsilon'\) is guaranteed whenever \(\operatorname{erfc}(R/(2c))<\vartheta(\vec\theta)\). Since \(\operatorname{erfc}\) is strictly decreasing on \([0,\infty)\), this is equivalent to \(R>2c\,\operatorname{erfc}^{-1}(\vartheta(\vec\theta))\). If \(\vartheta(\vec{\theta})\ge 1\), then \(\operatorname{erfc}(R/(2c))<1\leq \vartheta(\vec\theta)\), and the same estimate gives \(E_{\mathrm{trunc}}(R)<\epsilon'\) for all \(R>0\).
\end{proof}

\subsection{Inner Simulation Error}\label{inner-sim-subsub}
We next compute the inner simulation error. For \(k\in\mathbb{R}\), define \(G_k:=H+kL\). Since \(G_k\) consists of \(H\) and \(L\), each nested commutator is formed by choosing either \(H\) or \(L\) at each position, motivating the following definition.

\begin{defn}\label{comm-decouple}
For integers \( j \ge 2 \) and \( 0 \le \ell \le j \), define
\begin{equation}
\mathcal
C_{j,\ell}(H,L) := \sum_{\substack{M \in \{H,L\}^{j}\\ N_L(M)=\ell}} \|[M_1,M_2,\dots,M_j] \|,
\label{eq:def_Cql}
\end{equation}
where \( N_L(M) \) denotes the number of occurrences of \( L \) in the word \( M=(M_1,\dots,M_j) \in \{H,L\}^{j} \).
\end{defn}

\cref{lem:polynomial_dependence_k} formalizes the fact that the nested commutators of \(G_k\) depend polynomially on \(k\). 

\begin{lem}
\label{lem:polynomial_dependence_k}
Fix \( j \ge 2 \).  The \(j\)-fold nested commutator profile associated with the decomposition \(G_k=H+kL\) satisfies
\begin{equation}\label{eq:polynomial_dependence_bound}
\alpha_{\mathrm{comm},j}(G_k) := \sum_{M_1,\ldots,M_j \in \{H,kL\}} \|[M_1,\ldots,M_j]\| = \sum_{\ell=0}^{j} |k|^{\ell}\mathcal{C}_{j,\ell}(H,L).
\end{equation}
In particular, $\alpha_{\mathrm{comm},j}(G_k)$ is a polynomial in \(|k|\) of degree at most \(j\), with coefficients determined by nested commutators of \(H\) and \(L\).
\end{lem}

\begin{proof} By multilinearity of nested commutators, each word in the decomposition \(G_k=H+kL\) containing exactly \(\ell\) copies of \(L\) contributes a factor \(k^\ell\). Hence its norm contributes \(|k|^\ell\). Grouping the resulting commutator terms by the number \(\ell\) of occurrences of \(L\) gives \cref{eq:polynomial_dependence_bound}.  \end{proof}

We approximate $U_k(t)=e^{-iG_kt}$ by a $2m$-th order MPF built from the second-order product formula, denoted $U_k^{(2m)}$. We state the error bound abstractly, which is instantiated later using either \cref{aftab-time-indep-short} or \cref{mizuta-time-indep-short}.

\begin{assumption}
\label{abstract-time-indep-mpf-revised}
For \(R>0\), consider the time-independent family \(\{G_k\}_{|k|\leq R}\). Fix \(m,N\geq 1\) and \(\delta>0\), and let \(U_k^{(2m)}(t)\) denote the order-\(2m\) MPF approximation to \(U_k(t)\). Assume there exist:
\begin{enumerate}
    \item a non-negative function
    \begin{equation}
    \Phi_m:[-R,R]\times\mathbb{N}\to\mathbb{R}_+,
    \end{equation}
    based on nested commutators of \(G_k\), the system size $N$ and satisfying
    $\Phi_{m,R}^{\ast}(N):=\sup_{|k|\le R}\Phi_m(k,N)<\infty$,
    \item a step size \(\Delta(R,N,\delta)>0\),
    \item a remainder term \(R_m(\delta)\geq 0\),
\end{enumerate}
such that, for all \(|k|\leq R\) and \(0<\Delta\leq\Delta(R,N,\delta)\),
\begin{equation}
\|U_k^{(2m)}(\Delta)-U_k(\Delta)\| \leq C_m\Delta^{2m+1}\Phi_m(k,N)+R_m(\delta) =:\eta_k(\Delta,k,N,\delta,m),
\end{equation}
where \(C_m>0\) is constant. For simplicity, we write \(\eta_k(\Delta)\) for \(\eta_k(\Delta,k,N,\delta,m)\).
\end{assumption}

\begin{remark}
\cref{abstract-time-indep-mpf-revised} captures different MPF error scalings through $\Phi_m(k,N)$. Under \cref{aftab-time-indep-short}, expanding $G_k$ gives a convergent series in $|k|$. In contrast, \cref{mizuta-time-indep-short} involves commutator products followed by a $(j+l)$-th root. This yields non-polynomial dependence on $k$. Thus, \cref{abstract-time-indep-mpf-revised} covers both polynomial and root-type commutator scalings. Also, $R_m(\delta)=0$ for \cref{aftab-time-indep-short}, while $R_m(\delta)=\|\vec a\|_1\|\vec b\|_1\delta$ for \cref{mizuta-time-indep-short}.
\end{remark}

We now quantify how the abstract error bound in \cref{abstract-time-indep-mpf-revised} propagates through the LCHS representation. We first define the quantities needed for the short-time analysis in \cref{def:kernel-moments-and-profile}, prove a one-step MPF error bound for step size $\Delta$ in \cref{short-time-step}, and extend it to $[0,T]$ in \cref{long-time-step}. Since the argument does not use specific properties of \(\hat f_{a,b}(k;c,d)\), we write \(\hat f(k)\) unless the parameter dependence is needed.

\begin{defn}\label{def:kernel-moments-and-profile}
Let \(\hat f\) be a kernel (as in \cref{low-somma-kernel}). For \(R>0\), define the truncated kernel norm by \(\alpha_{\hat f,R}:=\|\hat f\|_{L^1[-R,R]}/\sqrt{2\pi}\). Under \cref{abstract-time-indep-mpf-revised}, define the truncated profile average
\begin{equation}\label{eq:Lambda-cont}
\Lambda_m(R,N) := \frac{1}{\sqrt{2\pi}} \int_{-R}^{R} | \hat f (k)| \, \Phi_m(k,N) dk.
\end{equation}
\end{defn}

We first establish the error bound for a short time step \(\Delta\) in \cref{short-time-step}. The proof can be found in \cref{proof-of-short-time-step}. 

\begin{lem}
\label{short-time-step}
Fix $R, \delta > 0, \; m,N \geq 1$.
For all $0<\Delta\le \Delta(R,N,\delta)$, we have
\begin{equation}\label{eq:one-step-avg}
E_{\mathrm{MPF}} (\Delta) =
\left \| \frac{1}{\sqrt{2 \pi}} \int_{-R}^{R} \hat f(k)\big(U^{(2m)}_k(\Delta)-U_k(\Delta)\big)dk\right \|
\leq 
 C_m\Delta^{2m+1}\Lambda_m(R,N) + 
\alpha_{\hat f, R} R_m(\delta).
\end{equation}
\end{lem}

We now extend this short-time error bound to the full simulation over $[0,T]$ in \cref{long-time-step}.

\begin{prop}
\label{long-time-step}
Fix $R, \delta > 0, \; m,N \geq 1$. Let $T>0$, $r\in\mathbb{N}$. Set $\Delta := T/r$. Assume that $0< \Delta \le \Delta(R,N,\delta)$. Then
\begin{equation}\label{eq:long-time-avg-general}
E_{\mathrm{MPF}} (T) =
\left\| \frac{1}{\sqrt{2 \pi}} \int_{-R}^{R} \hat f(k)\Big(\big(U_k^{(2m)}(\Delta)\big)^r-U_k(T)\Big)dk \right\| \le \frac{r}{\sqrt{2 \pi}} \int_{-R}^{R} |\hat f(k)| \eta_k(\Delta) (1+\eta_k(\Delta) )^{r-1}dk,
\end{equation}
where $\eta_k(\Delta):=C_m \Delta ^{2m+1}\Phi_m(k,N)+R_m(\delta)$. In particular, if
$\eta_{R}(\Delta) = \sup_{|k| \leq R} \eta_k(\Delta)$, 
then
\begin{equation}\label{eq:long-time-avg-explicit}
\left\| \frac{1}{\sqrt{2 \pi}} \int_{-R}^{R} \hat f(k)\left((U_k^{(2m)}(\Delta))^r-U_k(T)\right)dk \right\| \le \frac{ C_m \frac{T^{2m+1}}{r^{2m}}\Lambda_m(R,N) + r \alpha_{\hat f, R} R_m(\delta) }{ (1+\eta_{R}(T/r))^{1-r}}.
\end{equation}
\end{prop}

\begin{proof}
Fix \(k\in[-R,R]\) and write \(A_k:=U_k^{(2m)}(\Delta)\) and \(B_k:=U_k(\Delta)\). Since \(U_k(T)=U_k(\Delta)^r=B_k^r\), the telescoping identity gives
\begin{equation}
A_k^r-B_k^r = \sum_{j=0}^{r-1} A_k^{\,r-1-j}(A_k-B_k)B_k^j.
\end{equation}
Taking norms and using that \(B_k\) is unitary, so \(\|B_k^j\|=1\), we obtain
\begin{equation}
\|A_k^r-B_k^r\| \le \sum_{j=0}^{r-1}\|A_k\|^{\,r-1-j}\|A_k-B_k\|.
\end{equation}
By \cref{abstract-time-indep-mpf-revised}, \(\|A_k-B_k\|\leq\eta_k(\Delta)\). Hence, \(\|A_k\|\leq \|B_k\|+\|A_k-B_k\|\leq 1+\eta_k(\Delta)\), and therefore
\begin{equation}
\|A_k^r-B_k^r\| \le \sum_{j=0}^{r-1} (1+\eta_k(\Delta) )^{r-1-j}\eta_k(\Delta) \le r\eta_k(\Delta) (1+\eta_k(\Delta) )^{r-1}.
\end{equation}
Multiplying by $|\hat f(k)|/\sqrt{2 \pi}$ and integrating over $[-R,R]$ yields
\begin{equation}
\left\| \frac{1}{\sqrt{2\pi}} \int_{-R}^{R} \hat f(k)\Big(\big(U_k^{(2m)}(\Delta)\big)^r-U_k(T)\Big)dk \right\| \le \frac{r}{\sqrt{2\pi}}\int_{-R}^{R} |\hat f(k)|\,\eta_k(\Delta) (1+\eta_k(\Delta) )^{r-1}dk,
\end{equation}
which proves~\cref{eq:long-time-avg-general}. Noting that for $|k|\le R$, $\eta_k(\Delta)\le \eta_{R}(T/r)$ implies that
\begin{equation}
\frac{r}{\sqrt{2\pi}}\int_{-R}^{R} |\hat f(k)|\,\eta_k(\Delta) (1+\eta_k(\Delta) )^{r-1}dk \le \frac{r}{\sqrt{2\pi}} (1+\eta_{R}(T/r) )^{r-1} \int_{-R}^{R} |\hat f(k)|\,\eta_k(\Delta)dk.
\end{equation}
Using the definition of $\eta_k(\Delta)$ and the quantities $\Lambda_m(R,N)$ and $\alpha_{\hat f, R}$, we find
\begin{equation}
\frac{1}{\sqrt{2\pi}} \int_{-R}^{R} |\hat f(k)|\,\eta_k(\Delta)dk = C_m \Delta^{2m+1}\Lambda_m(R,N)+\alpha_{\hat f, R} R_m(\delta).
\end{equation}
Substituting $\Delta=T/r$ gives
\begin{equation}
r \left( C_m \Delta^{2m+1}\Lambda_m(R,N)+\alpha_{\hat f, R} R_m(\delta) \right) = C_m\frac{T^{2m+1}}{r^{2m}}\Lambda_m(R,N) + r\,\alpha_{\hat f, R} R_m(\delta).
\end{equation}
This completes the proof.
\end{proof}

\subsection{Coupling of Truncation and Inner Simulation Errors}\label{coupling-subsub}
The key point of \cref{inner-sim-subsub} is that, after integrating over \(k\in[-R,R]\), the inner simulation error depends on \(\Lambda_m(R,N)\). Hence, \(R\) and the inner simulation error are coupled. Increasing \(R\) reduces truncation error but, when \(\Phi_m(k,N)\) grows with \(|k|\), also increases \(\Lambda_m(R,N)\) and amplifies simulation error. We capture this trade-off through a feasibility problem. Define

\begin{equation}\label{varf-def}
\mathcal{F}_{\epsilon_{\operatorname{comb}}} := \left\{ (R,\delta,r)\in(0,\infty)^2\times\mathbb{N} : E_{\mathrm{trunc}}(R) + E_{\mathrm{MPF}}(T) \le \epsilon_{\operatorname{comb}}, \frac{T}{r} \le \Delta(R,N,\delta) \right\}.
\end{equation}
Here $\epsilon_{\operatorname{comb}}$ is the total error budget for the truncation and inner simulation errors. If \(E_{\mathrm{trunc}}(R) \geq \epsilon_{\mathrm{comb}}\), then \((R,\delta,r)\notin \mathcal F_{\epsilon_{\mathrm{comb}}}\) for every \(r\in\mathbb N\) and \(\delta>0\). Hence a positive error budget for the inner simulation is available only when \(R\in\mathcal H_{\epsilon_{\mathrm{comb}}}\). For such \(R\), define $\epsilon_R := \epsilon_{\mathrm{comb}}-E_{\mathrm{trunc}}(R)>0$.  After fixing $R \in \mathcal H_{\epsilon_{\mathrm{comb}}}$, the problem therefore reduces to determining pairs \((r,\delta)\in\mathbb N\times(0,\infty)\) such that
\begin{equation}
E_{\mathrm{MPF}}(T) \le \epsilon_R, \quad \frac{T}{r} \le \Delta(R,N,\delta)
\label{eq:feasibility-reduced}
\end{equation}
The choice of \(r\) depends on whether \cref{aftab-time-indep-short} or \cref{mizuta-time-indep-short} instantiates \cref{abstract-time-indep-mpf-revised}. We defer this discussion to \cref{complexity-analysis}.

\section{Post-Quadrature Complexity Analysis}\label{complexity-analysis}
We now analyze the overall query complexity of our algorithm. We fix an admissible kernel from \cref{low-somma-kernel}. \cref{quadrature-error} gives the post-quadrature error bound, \cref{implem} describes the implementation, \cref{abs-rel} derives complexity under the abstract MPF estimate in \cref{abstract-time-indep-mpf-revised}, and \cref{spec-aftab,spec-mizuta} specialize the analysis to \cref{aftab-time-indep-short} and \cref{mizuta-time-indep-short}, respectively. \cref{quad-rules} discusses various quadrature rules.

\subsection{Post-Quadrature Error Bound}\label{quadrature-error}
The analysis in \cref{coupling-subsub} was carried out before discretizing the LCHS integral. We now show that the same coupling persists after applying a quadrature rule. Let $Q=\{(k_i,w_i)\}_{i\in\mathcal I_Q}$ be a quadrature rule, with $k_i\in\mathbb R$ and $w_i\in\mathbb C$. Define its radius by $R_Q:=\max_{i\in\mathcal I_Q}|k_i|$.
The corresponding ideal quadrature operator is
\begin{equation}\label{eq:generic-quadrature-operator}
W_Q^{\mathrm{ideal}} := \sum_{i\in\mathcal I_Q}v_iU_{k_i}(T), \qquad v_i:= \frac{w_i}{\sqrt{2\pi}}\hat f_{a,b}(k_i;c,d)
\end{equation}
and the associated quadrature mass is $\alpha_Q := \sum_{i\in\mathcal I_Q}|v_i|$. Under \cref{abstract-time-indep-mpf-revised}, we also define the discrete weighted MPF profile by
\begin{equation}\label{eq:Lambda-Q}
\Lambda_{m,Q} := \sum_{i\in\mathcal I_Q}|v_i|\Phi_m(k_i,N).
\end{equation}
This is the post-quadrature analogue of \(\Lambda_m(R,N)\). Define \(\Phi_{m,Q}^\ast:=\max_{i\in\mathcal I_Q}\Phi_m(k_i,N)\). Let \(U_{k_i}^{(2m)}(T)\) denote the \(2m\)-th order MPF approximation to \(U_{k_i}(T)\) over \([0,T]\), and define
$W_Q^{\mathrm{MPF}} := \sum_{i\in\mathcal I_Q}v_i U^{(2m)}_{k_i}(T)$. 
Specific quadrature rules are discussed in \cref{quad-rules}. We have the following discrete analog of \cref{long-time-step}.

\begin{cor}\label{discrete-mpf-error}
Let $Q=\{(k_i,w_i)\}_{i\in\mathcal I_Q}$ be a quadrature rule and set $R_Q:=\max_{i\in\mathcal I_Q}|k_i|$. Fix $m,N\ge 1$, $T>0$, $r\in\mathbb N$, and $\delta >0$. Set $\Delta:=T/r$. Assume that $0<\Delta\le \Delta(R_Q,N,\delta)$. Then
\begin{align}
\|W_Q^{\mathrm{ideal}}-W_Q^{\mathrm{MPF}}\| &\le \sum_{i\in\mathcal I_Q}|v_i| \|U_{k_i}(T)-U^{(2m)}_{k_i}(T)\|  \le r\sum_{i\in\mathcal I_Q} |v_i|\eta_{k_i}(\Delta)(1+\eta_{k_i}(\Delta))^{r-1}.
\label{eq:discrete-mpf-error-exact}
\end{align}
In particular, if $\eta_Q(\Delta) := \max_{i\in\mathcal I_Q}\eta_{k_i}(\Delta)$, then
\begin{equation}\label{eq:discrete-mpf-error-bound}
\|W_Q^{\mathrm{ideal}}-W_Q^{\mathrm{MPF}}\| \le (1+\eta_Q(T/r))^{r-1} \left( C_m\frac{T^{2m+1}}{r^{2m}}\Lambda_{m,Q} + r\alpha_QR_m(\delta) \right).
\end{equation}
\end{cor}

\begin{proof}
The result follows by the same telescoping argument as in \cref{long-time-step}, followed by the definitions of $\alpha_Q$ and $\Lambda_{m,Q}$.
\end{proof}

\cref{discrete-mpf-error} shows that, after fixing the quadrature rule, the inner simulation error couples to the quadrature error through \(R_Q\), \(\alpha_Q\), and \(\Lambda_{m,Q}\). This is the post-quadrature analogue of \cref{coupling-subsub}. We end with a total error bound, which is the starting point for the complexity analysis in \cref{abs-rel}.

\begin{lem}\label{post-quadrature-error-decomposition-general}
Fix a kernel profile \(\vec{\theta}=(a,b,c,d)\) such that $d$ and $y_0$ are chosen as in \cref{feasible-kernel-approx} or \cref{feasible-kernel-profile-explicit}. Let $Q$ be a quadrature rule, $T, \delta>0$, $r\in\mathbb N$ and set $\Delta=T/r$. If $0<\Delta\le \Delta(R_Q,N,\delta)$, then
\begin{align}\label{post-quadrature-error-general}
\|e^{-AT}-W_Q^{\mathrm{MPF}}\| &\le E_{\mathrm{approx}}(y_0) + E_{\mathrm{quad}}(Q) \\ & + (1+\eta_Q(T/r))^{r-1} \left( C_m\frac{T^{2m+1}}{r^{2m}}\Lambda_{m,Q} + r\alpha_QR_m(\delta) \right).
\end{align}
\end{lem}

\begin{proof}
Let $W_\infty := \frac{1}{\sqrt{2\pi}} \int_{\mathbb R} \hat f_{a,b}(k;c,d) U_k(T) dk$, and define the quadrature error by $E_{\mathrm{quad}}(Q) := \|W_\infty - W_Q^{\mathrm{ideal}}\|$. By the triangle inequality, we have
\begin{align}
\|e^{-AT}-W_Q^{\mathrm{MPF}}\| &\le \|e^{-AT}-W_\infty\| + \|W_\infty-W_Q^{\mathrm{ideal}}\| + \|W_Q^{\mathrm{ideal}}-W_Q^{\mathrm{MPF}}\|.
\end{align}
The first term is bounded by \(E_{\mathrm{approx}}(y_0)\), the second term is \(E_{\mathrm{quad}}(Q)\), and the third term follows from \cref{discrete-mpf-error}. Since the nodes of \(Q\) lie in \([-R_Q,R_Q]\), the MPF step-size condition is imposed with \(R_Q\), yielding a bound involving the quadrature-induced quantities \(\alpha_Q\), \(\Lambda_{m,Q}\), and \(\eta_Q\).
\end{proof}

\begin{remark} The quantity \(E_{\mathrm{quad}}(Q)\) in \cref{post-quadrature-error-decomposition-general} measures the total outer discretization error from the infinite LCHS integral to the finite quadrature rule. For concrete quadrature rules, this error may itself be decomposed into truncation and quadrature components (see \cref{quad-rules}).
\end{remark}

\cref{post-quadrature-error-decomposition-general} shows that the quadrature rule affects the final error through \(E_{\mathrm{quad}}(Q)\) and through \(R_Q\), \(\alpha_Q\), and \(\Lambda_{m,Q}\). Thus, quadrature couples discretization error to quantum implementation error. This yields the post-quadrature problem. Choose an admissible rule \(Q\) with \(E_{\mathrm{quad}}(Q)<\epsilon_{\mathrm{comb}}\), leaving residual budget $\epsilon_Q := \epsilon_{\mathrm{comb}}-E_{\mathrm{quad}}(Q) > 0$. Then choose \((r,\delta)\in\mathbb N\times(0,\infty)\) such that
\begin{equation}
(1+\eta_Q(T/r))^{r-1} \left( C_m\frac{T^{2m+1}}{r^{2m}}\Lambda_{m,Q} + r\alpha_QR_m(\delta ) \right) \le \epsilon_Q, \quad \frac{T}{r}\le \Delta(R_Q,N,\delta ).
\end{equation}
This is the post-quadrature analogue of the discussion in \cref{coupling-subsub}. 
 
\subsection{Quantum Implementation}\label{implem}
The complexity estimates below address two tasks in a fully quantum implementation. The first task is to construct a block-encoding of the post-quadrature approximation
\begin{equation}
W_Q^{\operatorname{MPF}} = \sum_{i \in \mathcal{I}_Q} v_i U^{(2m)}_{k_i} (T) = \sum_{i \in \mathcal{I}_Q} v_i \left( \sum_{j=1}^m a_j \left[ U_{2,k_i}\left(\frac{T}{r b_j}\right) \right]^{b_j} \right)^r .
\end{equation}
Here \(v_i\) denotes the quadrature coefficients, while \(a_j,b_j\) denote the coefficients from \cref{sec:mpf}. Given \(Q\), define the quadrature preparation oracles
\begin{equation}
\operatorname{PREP}^Q_{\operatorname{R}}\ket{0} = \frac{1}{\sqrt{\alpha_Q}} \sum_{i\in\mathcal I_Q} \sqrt{v_i}\ket{i}, \quad \operatorname{PREP}^{Q}_{\operatorname{L}}\ket{0} = \frac{1}{\sqrt{\alpha_Q}} \sum_{i\in\mathcal I_Q} \overline{\sqrt{v_i}}\ket{i},
\end{equation}
where \(\alpha_Q:=\sum_{i\in\mathcal I_Q}|v_i|\).  For each \(i\in\mathcal I_Q\), let \(\mathcal U_Q\) be a unitary block-encoding of \(U_{k_i}^{(2m)}(T)\), normalized to one after amplitude amplification at each of the $r$ time steps. Hence, we have 
\begin{equation}
(\bra{0} \otimes I)\mathcal U_Q (\ket{0} \otimes I) = U_{k_i}^{(2m)}(T).
\end{equation}
Define the controlled unitary $\operatorname{SELECT}_Q := \sum_{i\in\mathcal I_Q} \ket{i}\bra{i}\otimes \mathcal U_Q$. By the outer LCU construction, we obtain a block-encoding of $W_Q^{\operatorname{MPF}}$
\begin{equation}
(\bra{0} \bra{0} \otimes I) (\operatorname{PREP}_{\operatorname L}^{Q\dagger}\otimes I) \operatorname{SELECT}_Q (\operatorname{PREP}_{\operatorname R}^{Q}\otimes I) (\ket{0} \ket{0} \otimes I) = \frac{W_Q^{\operatorname{MPF}}}{\alpha_Q}.
\end{equation}
Writing \(K_m:=\|\vec a\|_1\|\vec b\|_1\), one controlled implementation of \(U_{k_i}^{(2m)}(T)\) with \(r\) time steps costs \(\mathcal O(rK_m)\) controlled second-order product-formula queries. Hence, one use of \(\operatorname{SELECT}_Q\) has the same cost. One block-encoding of \(W_Q^{\operatorname{MPF}}/\alpha_Q\) uses \(\operatorname{SELECT}_Q\), the quadrature preparation oracles \(\operatorname{PREP}_{\operatorname L}^{Q}\), \(\operatorname{PREP}_{\operatorname R}^{Q}\), and their required inverses in the outer LCU construction.

\begin{remark}
The implementation of \(W_Q^{\operatorname{MPF}}\) is a nested LCU construction. The outer LCU is the quadrature sum over \(q\), with normalization \(\alpha_Q\), while the inner LCU implements each selected \(U_{k_i}^{(2m)}(T)\) by an MPF. We do not flatten this into one LCU over \((q,j_1,\ldots,j_r)\), since that would introduce normalization \(\alpha_Q\|\vec a\|_1^r\), which is not the implementation model used here.
\end{remark}

The second task is to prepare the normalized output state \(W_Q^{\operatorname{MPF}}u(0)/\|W_Q^{\operatorname{MPF}}u(0)\|\). Write \(\ket{\widetilde u(0)}:=u(0)/\|u(0)\|\). Applying the block-encoding of \(W_Q^{\operatorname{MPF}}/\alpha_Q\) to \(\ket{\widetilde u(0)}\) and post-selecting the ancillas on \(\ket{0}\) gives the un-normalized state
\(W_Q^{\operatorname{MPF}}\ket{\widetilde u(0)}/\alpha_Q\), whose norm is
\begin{equation}
\frac{\|W_Q^{\operatorname{MPF}}\ket{\widetilde u(0)}\|} {\alpha_Q} = \frac{\|W_Q^{\operatorname{MPF}}u(0)\|} {\alpha_Q\|u(0)\|}.
\end{equation}
Hence, amplitude amplification requires \(\chi_Q(u(0)) := \frac{\alpha_Q\|u(0)\|}{\|W_Q^{\operatorname{MPF}}u(0)\|}\) uses of the block-encoding, up to absolute constants. Therefore, normalized state preparation costs \(\mathcal O(\chi_Q(u(0))rK_m)\) controlled second-order product-formula queries, \(\mathcal O(\chi_Q(u(0)))\) uses of the quadrature preparation oracles and their inverses, and \(\mathcal O(\chi_Q(u(0)))\) uses of the state-preparation oracle for \(\ket{\widetilde u(0)}\) and its inverse.

\begin{remark}\label{state-norm-1}
Let \(u(T):=e^{-AT}u(0)\). If \(\|W_Q^{\operatorname{MPF}}-e^{-AT}\|\le \epsilon\) and \(\epsilon\|u(0)\|\le \|u(T)\|/2\), then \begin{equation} \|W_Q^{\operatorname{MPF}}u(0)\|\ge \|u(T)\| - \epsilon\|u(0)\|\ge \frac{\|u(T)\|}{2}. \end{equation} Consequently, \(\chi_Q(u(0))\le 2\alpha_Q\|u(0)\|/\|u(T)\|\), so the state query complexity is \(\mathcal O\!\left( \frac{\|u(0)\|}{\|u(T)\|}\alpha_Q rK_m \right)\), which becomes \(\mathcal O(\alpha_Q rK_m/\|u(T)\|)\) when \(\|u(0)\|=1\). This is the standard form in which normalized-state query complexity is presented in the literature.
\end{remark}

\subsection{Overall Query Complexity}\label{abs-rel}
Split the error tolerance as \(\epsilon_{\operatorname{approx}}+\epsilon_{\operatorname{comb}}\leq\epsilon\), where \(\epsilon_{\operatorname{approx}}\) controls \(E_{\operatorname{approx}}(y_0)\) and \(\epsilon_{\operatorname{comb}}\) controls the combined quadrature and inner-simulation errors.  \cref{feasible-kernel-approx,feasible-kernel-profile-explicit} gives sufficient conditions on $c,d$ to ensure the approximation error is bounded blow by $\epsilon_{\operatorname{approx}}$. Choose a  quadrature rule \(Q\) such that \(E_{\operatorname{quad}}(Q)<\epsilon_{\operatorname{comb}}\), and define \(\epsilon_Q:=\epsilon_{\operatorname{comb}}-E_{\operatorname{quad}}(Q)\). Specific quadrature rules are discussed in \cref{quad-rules}. We first prove an abstract fixed-\(2m\)-order overall complexity estimate, based on \cref{abstract-time-indep-mpf-revised}, for arbitrary \(Q\). For \(\delta>0\), define
\begin{equation}\label{abstract-rQ-delta-def}
r_Q(\delta) := \left\lceil \max\left\{ \frac{T}{\Delta(R_Q,N,\delta)}, \left(2C_mT^{2m+1}\Phi_{m,Q}^{\ast}\right)^{1/2m}, \left( \frac{2e\,C_mT^{2m+1}\Lambda_{m,Q}}{\epsilon_Q} \right)^{1/2m} \right\} \right\rceil .
\end{equation}

\begin{thm}
\label{abstract-fixed-order-arbitrary-quadrature}
Fix \(m,N\in\mathbb N\) and \(T>0\). 
Let \(\epsilon,\epsilon_{\operatorname{approx}},\epsilon_{\operatorname{comb}}>0\) satisfy \(\epsilon_{\operatorname{approx}}+\epsilon_{\operatorname{comb}}\leq\epsilon\). Let \(\vec{\theta}=(a,b,c,d)\) be an admissible kernel profile, and let \(Q\) be a quadrature rule such that \(E_{\operatorname{quad}}(Q)<\epsilon_{\operatorname{comb}}\). Set \(\epsilon_Q:=\epsilon_{\operatorname{comb}}-E_{\operatorname{quad}}(Q)>0\). Suppose there exists \(\delta>0\) such that
\begin{equation}\label{abstract-delta-admissible}
r_Q(\delta)R_m(\delta)\leq \frac12, \qquad r_Q(\delta)\alpha_QR_m(\delta) \leq \frac{\epsilon_Q}{2e}.
\end{equation}
Choose \(r=r_Q(\delta)\). Then \(\|e^{-AT}-W_Q^{\operatorname{MPF}}\|\leq\epsilon\). Moreover, a block-encoding of \(W_Q^{\operatorname{MPF}}/\alpha_Q\) can be implemented using
\begin{equation}\label{abstract-fixed-block-complexity}
\mathcal Q_{\operatorname{block}} = \mathcal O\left(r_Q(\delta)K_m\right)
\end{equation}
controlled second-order product-formula queries. If \(W_Q^{\operatorname{MPF}}u(0)\neq0\), the normalized output state proportional
to \(W_Q^{\operatorname{MPF}}\ket{u(0)}\) can be prepared using
\begin{equation}\label{abstract-fixed-state-complexity}
\mathcal Q_{\operatorname{state}} = \mathcal O\left(\chi_Q(u(0))r_Q(\delta)K_m\right)
\end{equation}
controlled second-order product-formula queries. Equivalently, this uses \(\mathcal O(\chi_Q(u(0)))\) applications of the block-encoding of \(W_Q^{\operatorname{MPF}}/\alpha_Q\), hence \(\mathcal O(\chi_Q(u(0)))\) calls to \(\operatorname{SELECT}_Q\), the quadrature preparation oracles, their inverses, and the state-preparation oracle for \(\ket{\widetilde u(0)}\).
\end{thm}

\begin{proof}
Let \(r=r_Q(\delta)\). By \cref{abstract-rQ-delta-def}, we have
\begin{equation}
\frac{T}{r}\le \Delta(R_Q,N,\delta), \quad C_m\frac{T^{2m+1}}{r^{2m}}\Phi_{m,Q}^{\ast} \le \frac12, \quad C_m\frac{T^{2m+1}}{r^{2m}}\Lambda_{m,Q} \le \frac{\epsilon_Q}{2e}.
\end{equation}
Together with \cref{abstract-delta-admissible}, this gives
\begin{equation}
r\eta_Q(T/r) = C_m\frac{T^{2m+1}}{r^{2m}}\Phi_{m,Q}^{\ast} + rR_m(\delta) \le 1.
\end{equation}
Hence, we have
\begin{equation}
(1+\eta_Q(T/r))^{r-1} \le (1+\eta_Q(T/r))^r \le e^{r\eta_Q(T/r)} \le e.
\end{equation}
Using this estimate, the third condition in \cref{abstract-rQ-delta-def}, and \cref{abstract-delta-admissible}, we obtain
\begin{align}(1+\eta_Q(T/r))^{r-1} \left( C_m\frac{T^{2m+1}}{r^{2m}}\Lambda_{m,Q} + r\alpha_QR_m(\delta) \right) \le e\left( \frac{\epsilon_Q}{2e}+\frac{\epsilon_Q}{2e} \right) = \epsilon_Q.
\end{align}
By the estimate in \cref{post-quadrature-error-decomposition-general} and the admissibility of the kernel and quadrature rule, we have
\begin{align}
\|e^{-AT}-W_Q^{\operatorname{MPF}}\| &\le \epsilon_{\operatorname{approx}} + E_{\operatorname{quad}}(Q) + \epsilon_Q = \epsilon_{\operatorname{approx}} +\epsilon_{\operatorname{comb}} \le \epsilon.
\end{align}
The block-encoding and normalized-state query estimates follow from the discussion \cref{implem}.
\end{proof}

In what follows, define the \(m\)-dependent post-quadrature effective commutator scale
\begin{equation}\label{abstract-mu-mQ-def}
\mu_{m,Q}(\delta) := \max\left\{ \frac{1}{\Delta(R_Q,N,\delta)}, (C_m\Phi_{m,Q}^{\ast})^{1/2m}, \left(C_m\Lambda_{m,Q}\right)^{1/2m} \right\}.
\end{equation}
It can then be checked that \cref{abstract-rQ-delta-def} implies that
\begin{equation}\label{abstract-rQ-mu-bound}
r_Q(\delta) = \mathcal O\left( 1+ \mu_{m,Q}(\delta) \max\left\{ T,\, T^{1+\frac1{2m}},\, T^{1+\frac1{2m}}\epsilon_Q^{-\frac1{2m}} \right\} \right),
\end{equation}
for $\epsilon_Q > 0$. Under the additional assumption \(0<\epsilon_Q<1\), this simplifies to
\begin{equation}\label{abstract-rQ-mu-bound-simple}
r_Q(\delta) = \mathcal O\left( 1+ \mu_{m,Q}(\delta) \max\left\{ T,\, T^{1+\frac1{2m}}\epsilon_Q^{-\frac1{2m}} \right\} \right).
\end{equation}
In what follows, assume \(0<\epsilon_Q<1\). In \cref{abstract-optimized-order-arbitrary-quadrature}, we allow mild MPF-order dependence in the effective commutator scale and optimize over the order to obtain complexity logarithmic in \(T/\epsilon_Q\), up to the growth envelope \(F\).

\begin{cor}
\label{abstract-optimized-order-arbitrary-quadrature}
Assume the error budget, kernel profile, and quadrature hypotheses of \cref{abstract-fixed-order-arbitrary-quadrature}. Let \(K_m:=\|\vec a\|_1\|\vec b\|_1\) and \(0<\epsilon_Q<1\). Suppose there exist \(m_0\in\mathbb N\), \(\mu_Q\geq1\), and a non-decreasing function \(F:[m_0,\infty)\to[1,\infty)\) such that one can choose \(\delta_m>0\) satisfying
\begin{equation}\label{F-cond}
\mu_{m,Q}(\delta_m) \leq \mu_QF(m).
\end{equation}
for $m \geq m_0$. Assume also that, for every \(m\geq m_0\), we have
\begin{equation}
r_Q(\delta_m)R_m(\delta_m)\le \frac12, \quad r_Q(\delta_m)\alpha_QR_m(\delta_m) \le \frac{\epsilon_Q}{2e}.
\label{abstract-optimized-admissibility}
\end{equation}
Choose $m=\max\left\{m_0,\left\lceil\log\left(e+\frac{T}{\epsilon_Q}\right)\right\rceil\right\}$ and set \(\delta=\delta_m\) and \(r=r_Q(\delta_m)\). Then \(\|e^{-AT}-W_Q^{\operatorname{MPF}}\|\leq \epsilon\). Moreover, we have
\begin{equation}
\mathcal Q_{\operatorname{block}} = \mathcal O\Bigg( \left(1+\mu_QF(m)\max\{1,T\}\right) \left(\log\left(e+\frac{T}{\epsilon_Q}\right)\right)^2 \left(\log\log\left(e^e+\frac{T}{\epsilon_Q}\right)\right)^2 \Bigg).
\label{abstract-optimized-block-complexity}
\end{equation}
If \(W_Q^{\operatorname{MPF}}u(0)\neq0\), we have
\begin{align}
\mathcal Q_{\operatorname{state}}(u(0)) = \mathcal O\Bigg( \chi_Q(u(0)) \left(1+\mu_QF(m)\max\{1,T\}\right) \left(\log\left(e+\frac{T}{\epsilon_Q}\right)\right)^2 \left(\log\log\left(e^e+\frac{T}{\epsilon_Q}\right)\right)^2 \Bigg).
\label{abstract-optimized-state-complexity}
\end{align}
If \(T\geq1\) and \(\mu_QTF(m)\geq1\), then \cref{abstract-optimized-block-complexity} and \cref{abstract-optimized-state-complexity} simplify to
\begin{equation}
\mathcal Q_{\operatorname{block}} = \mathcal O\Bigg( \mu_QTF(m) \left(\log\left(e+\frac{T}{\epsilon_Q}\right)\right)^2 \left(\log\log\left(e^e+\frac{T}{\epsilon_Q}\right)\right)^2 \Bigg),
\label{abstract-optimized-block-complexity-simplified}
\end{equation}
and
\begin{align}
\mathcal Q_{\operatorname{state}}(u(0)) = \mathcal O\Bigg( \chi_Q(u(0))\,\mu_QTF(m) \left(\log\left(e+\frac{T}{\epsilon_Q}\right)\right)^2 \left(\log\log\left(e^e+\frac{T}{\epsilon_Q}\right)\right)^2 \Bigg).
\label{abstract-optimized-state-complexity-simplified}
\end{align}
\end{cor}

\begin{proof}
The admissibility conditions in \cref{abstract-optimized-admissibility} allow us to apply \cref{abstract-fixed-order-arbitrary-quadrature} with \(\delta=\delta_m\), so the error is at most \(\epsilon\). By \cref{abstract-fixed-block-complexity}, we have $\mathcal Q_{\operatorname{block}} = \mathcal O(r_Q(\delta_m)K_m)$.  Using \cref{abstract-rQ-mu-bound-simple} and \cref{F-cond}, we obtain
\begin{equation}
r_Q(\delta_m) = \mathcal O\left( 1+ \mu_QF(m) \max\left\{ T,\, T^{1+\frac1{2m}}\epsilon_Q^{-\frac1{2m}} \right\} \right).
\label{abstract-rQ-corrected-bound-in-proof}
\end{equation}
Since $T^{1+\frac1{2m}}\epsilon_Q^{-\frac1{2m}} = T\left(\frac{T}{\epsilon_Q}\right)^{1/2m}$,  the choice $m = \max\left\{ m_0, \left\lceil \log\left(e+\frac{T}{\epsilon_Q}\right) \right\rceil \right\}$ implies $(T/\epsilon_Q)^{1/2m} = \mathcal O(1)$. Therefore, we have
\begin{equation}
r_Q(\delta_m) = \mathcal O\left( 1+\mu_QF(m)\max\{1,T\} \right).
\label{abstract-rQ-optimized-corrected}
\end{equation}
Consequently, we have $\mathcal Q_{\operatorname{block}} = \mathcal O\left( K_m\left(1+\mu_QF(m)\max\{1,T\}\right) \right)$. 
Since \(K_m=\mathcal O(m^2(\log m)^2)\) as in \cref{sec:mpf}, the choice of \(m\) gives
\begin{equation}
m^2(\log m)^2 = \mathcal O\left( \left(\log\left(e+\frac{T}{\epsilon_Q}\right)\right)^2 \left(\log\log\left(e^e+\frac{T}{\epsilon_Q}\right)\right)^2 \right),
\label{abstract-Km-optimized-bound}
\end{equation}
with constants depending at most on \(m_0\). Hence, we have
\begin{equation}
\mathcal Q_{\operatorname{block}} = \mathcal O\left( \left(1+\mu_QF(m)\max\{1,T\}\right) \left(\log\left(e+\frac{T}{\epsilon_Q}\right)\right)^2 \left(\log\log\left(e^e+\frac{T}{\epsilon_Q}\right)\right)^2 \right).
\label{abstract-Qblock-optimized-proof-final}
\end{equation}
This proves \cref{abstract-optimized-block-complexity}. Multiplying by \(\chi_Q(u(0))\) gives \cref{abstract-optimized-state-complexity}. If \(T\geq1\) and \(\mu_QTF(m)\geq1\), then
\begin{equation}
1+\mu_QF(m)\max\{1,T\} = 1+\mu_QTF(m) = \mathcal O(\mu_QTF(m)).
\label{abstract-simplified-F-factor}
\end{equation}
The derivation establishes \cref{abstract-optimized-block-complexity-simplified} and \cref{abstract-optimized-state-complexity-simplified}.
\end{proof}

In the applications considered in this paper, we have \(F(m)=m^u\bigl(\log(e+m)\bigr)^v\) for some \(u,v\geq0\). Then \cref{abstract-optimized-order-arbitrary-quadrature} implies the following explicit polylogarithmic form:
\begin{align}\label{abstract-optimized-block-complexity-polylog}
\mathcal Q_{\operatorname{block}} = \mathcal O\Bigg( &\left( 1+ \mu_Q\max\{1,T\} \left(\log\left(e+\frac{T}{\epsilon_Q}\right)\right)^u \left( \log\left(e+\log\left(e+\frac{T}{\epsilon_Q}\right)\right) \right)^v \right) \notag\\ &\hspace{2cm}\times \left(\log\left(e+\frac{T}{\epsilon_Q}\right)\right)^2 \left( \log\left(e+\log\left(e+\frac{T}{\epsilon_Q}\right)\right) \right)^2 \Bigg).
\end{align}
Indeed, the choice of \(m\) implies 
$m = \mathcal O\left( \log\left(e+\frac{T}{\epsilon_Q}\right) \right)$. Therefore, we have
\begin{equation}
F(m) = \mathcal O\left( \left(\log\left(e+\frac{T}{\epsilon_Q}\right)\right)^u \left( \log\left(e+\log\left(e+\frac{T}{\epsilon_Q}\right)\right) \right)^v \right).
\end{equation}
Substituting this into \cref{abstract-optimized-block-complexity} gives \cref{abstract-optimized-block-complexity-polylog}. The state complexity follows similarly.

\subsection{General Commutator Scaling Specialization}\label{spec-aftab} 
The analysis in \cref{abs-rel} allowed arbitrary quadrature under \cref{abstract-time-indep-mpf-revised}. We now specialize to \cref{aftab-time-indep-short}, still with arbitrary quadrature. We first discuss the applicability of \cref{aftab-time-indep-short}. For each \(G_{k_i}=H+k_iL\), we must compute \(\inf_{j\geq J}\alpha_{\operatorname{comm},j}(G_{k_i})^{-1/j}\). Define \(\mathcal C_j(H,L):=\max_{0\leq \ell\leq j}\mathcal C_{j,\ell}(H,L)\). We have
\begin{align}
\label{series-alpha-comm-bound}
\alpha_{\operatorname{comm},j}(G_{k_i}) \leq \sum_{\ell=0}^{j}|k_i|^\ell\mathcal C_{j,\ell}(H,L) \leq \mathcal C_j(H,L)(j+1)\max\{1,|k_i|\}^j .
\end{align}
Let $\chi_J(H,L):=\inf_{j\geq J}\mathcal C_j(H,L)^{-1/j}$. For \(J\geq 1\), we have
\begin{align}
\label{eq:comm-bound-root}
\inf_{j\geq J}\alpha_{\operatorname{comm},j}(G_{k_i})^{-1/j} &\geq \frac{\inf_{j\geq J} ( (j+1)^{-1/j} \mathcal C_j(H,L)^{-1/j})} {\max\{1,|k_i|\}}  \\ &\geq \frac{ (\inf_{j\geq J}(j+1)^{-1/j}) (\inf_{j\geq J}\mathcal C_j(H,L)^{-1/j})} {\max\{1,|k_i|\}} \\ &= \frac{\chi_J(H,L)} {(J+1)^{1/J}\max\{1,|k_i|\}} \geq \frac{\chi_J(H,L)} {(J+1)^{1/J}\max\{1,R_Q\}} := \kappa_Q.
\end{align}
The last inequality follows from \(|k_i|\leq R_Q\). 
Thus \cref{aftab-time-indep-short} applies uniformly to all \(G_{k_i}\) whenever \(\chi_J(H,L)>0\) and \(\Delta\leq\rho_Q\) for some \(0<\rho_Q<\kappa_Q\). The parameter \(\rho_Q\) is a uniform admissible step-size and enters the error profile because \cref{aftab-time-indep-short} gives an MPF error series in powers of \(\Delta^{j+\ell}\), while the abstract framework factors out \(\Delta^{2m+1}\). Set \(C_m:=\|\vec a\|_1\) and define
\begin{align}
\label{series-Dmnu-def}
\Phi_{m,\nu}(H,L) := \sum_{\substack{j\in 2\mathbb Z_+\\ j\geq 2m}} \sum_{l=1}^{m} \frac{\rho_Q^{j+l-(2m+1)}}{l!} \sum_{\substack{j_1,\dots,j_l\in 2\mathbb Z_+\\ j_1+\cdots+j_l=j}} \sum_{\substack{\ell_1,\dots,\ell_l\geq 0\\ 0\leq \ell_\kappa\leq j_\kappa+1\\ \ell_1+\cdots+\ell_l=\nu}} \prod_{\kappa=1}^{l} \mathcal C_{j_\kappa+1,\ell_\kappa}(H,L).
\end{align}
for $\nu\geq 0$. Define the discrete quadrature moment \(M_\alpha^Q:=\sum_{i\in\mathcal I_Q}|v_i|\,|k_i|^\alpha\). The relevant quantities appearing in the complexity estimates are \(\Lambda_{m,Q}:=\sum_{\nu\geq 0}\Phi_{m,\nu}(H,L)M_\nu^Q\) and $\Phi^*_{m,Q} := \max_{i\in\mathcal I_Q} \Phi_{m,\nu}(H,k_iL)$. In \cref{series-profile-convergence-lemma}, which is proved in \cref{proof-of-series-profile-convergence-lemma}, we record the convergence of the profiles before stating the complexity consequence. This is the only point where the strict inequality \(\rho_Q<\kappa_Q\) is used.

\begin{lem}
\label{series-profile-convergence-lemma}
Assume that there exists \(J\geq 1\) such that $\chi_J(H,L):=\inf_{j\geq J}\mathcal C_j(H,L)^{-1/j}>0$.  Let \(0<\rho_Q<\kappa_Q\). The series defining \(\Phi_{m,\nu}(H,L)\), \(\Lambda_{m,Q}\), and \(\Phi^*_{m,Q}\) are absolutely convergent.
\end{lem}

It remains to verify that the MPF error has the abstract form required in \cref{abstract-time-indep-mpf-revised}. By \cref{aftab-time-indep-short}, for each \(i\in\mathcal I_Q\) and \(\Delta\leq\rho_Q\), we have
\begin{align}
\label{series-MPF-step-bound-original}
\| U_{\operatorname{MPF},k_i}(\Delta) - e^{-iG_{k_i}\Delta}\| \leq& \|\vec a\|_1 \sum_{\substack{j\in 2\mathbb Z_+\\ j\geq 2m}} \sum_{l=1}^{m} \frac{\Delta^{j+l}}{l!} \sum_{\substack{j_1,\dots,j_l\in 2\mathbb Z_+\\ j_1+\cdots+j_l=j}} \prod_{\kappa=1}^{l} \left( \sum_{\ell=0}^{j_\kappa+1} |k_i|^\ell \mathcal C_{j_\kappa+1,\ell}(H,L) \right).
\end{align}
Since \(j\geq 2m\) and \(l\geq 1\), we have $j+l\geq 2m+1$. Thus, for \(\Delta\leq \rho_Q\), we have $\Delta^{j+l} \leq \Delta^{2m+1}\rho_Q^{j+l-(2m+1)}$. Let $S_n(r):=\sum_{\ell=0}^{n}r^\ell\mathcal C_{n,\ell}(H,L)$. Hence, \cref{series-MPF-step-bound-original} gives
\begin{equation}
\label{series-MPF-step-bound-abstract}
\| U_{\operatorname{MPF},k_i}(\Delta) - e^{-iG_{k_i}\Delta} \| \leq C_m \Delta^{2m+1} \sum_{\substack{j\in 2\mathbb Z_+\\ j\geq 2m}} \sum_{l=1}^{m} \frac{\rho_Q^{j+l-(2m+1)}}{l!} \sum_{\substack{j_1,\dots,j_l\in 2\mathbb Z_+\\ j_1+\cdots+j_l=j}} \prod_{\kappa=1}^{l} S_{j_\kappa+1}(|k_i|) ,
\end{equation}
Denote the nested sum above as $\Phi_m(|k_i|)$. Its maximum over \(i\in\mathcal I_Q\) is precisely \(\Phi_{m,Q}^\ast\). Summing against the quadrature weights gives
\begin{align}
\label{series-weighted-MPF-bound}
\sum_{i\in\mathcal I_Q} |v_i|  \| U_{\operatorname{MPF},k_i}(\Delta) - e^{-iG_{k_i}\Delta} \| &\leq C_m \Delta^{2m+1} \sum_{i\in\mathcal I_Q}|v_i|\Phi_m(|k_i|) = C_m \Delta^{2m+1} \Lambda_{m,Q}.
\end{align}

Since \(R_m(\delta)=0\) in \cref{aftab-time-indep-short}, \(r_Q\) and \(\mu_{m,Q}\) are defined as in \cref{abstract-rQ-delta-def,abstract-mu-mQ-def} by setting \(\Delta(R_Q,N):=\rho_Q\) and suppressing \(\delta\).  The fixed-order and optimized error-complexity bounds then follow from \cref{abstract-fixed-order-arbitrary-quadrature,abstract-optimized-order-arbitrary-quadrature}. We record in \cref{lem:simplified-root-growth-scale} an easier sufficient condition for computing commutator scaling in practice.

\begin{lem}\label{lem:simplified-root-growth-scale}
For \(j\geq2\), set $\alpha_{Q,j}:=\max_{i\in\mathcal I_Q}\alpha_{\operatorname{comm},j}(G_{k_i})$ and $\omega:=\sup_{j\geq2}\alpha_{Q,j}^{1/j}$. Assume \(0<\omega<\infty\). If \(\rho_Q\omega<1\), then \cref{aftab-time-indep-short} applies uniformly to all \(G_{k_i}\), \(i\in\mathcal I_Q\), for every \(\Delta\leq\rho_Q\).
Moreover, we have
\begin{equation}\label{eq:mu-mQ-root-growth-bound}
\mu_{m,Q} \leq \max\left\{ \rho_Q^{-1}, \max\{1,\alpha_Q^{1/(2m)}\} (\|\vec a\|_1S_m)^{1/(2m)} \omega^{1+1/(2m)} \right\}, \end{equation}
where $S_m:=\sum_{\substack{j\in2\mathbb Z_+\\ j\geq2m}} \sum_{l=1}^{m} \frac{1}{l!}\binom{j-1}{l-1}(\rho_Q \omega)^{j+l-(2m+1)}$.  
\end{lem}

\begin{proof}
By definition, \(\alpha_{\operatorname{comm},j}(G_{k_i})\leq\alpha_{Q,j}\leq\omega^j\). Hence,  $\inf_{i\in\mathcal I_Q}\inf_{j\geq2} \alpha_{\operatorname{comm},j}(G_{k_i})^{-1/j} \geq \omega^{-1}$. If \(\Delta\leq\rho_Q\) and \(\rho_Q\omega<1\), then $\Delta\leq\rho_Q<\omega^{-1} \leq \inf_{i\in\mathcal I_Q}\inf_{j\geq2} \alpha_{\operatorname{comm},j}(G_{k_i})^{-1/j}$, 
so the step-size condition in \cref{aftab-time-indep-short} holds uniformly. For each product in the error profile in \cref{aftab-time-indep-short}, we have 
$\prod_{\kappa=1}^{l}\alpha_{Q,j_\kappa+1} \leq \prod_{\kappa=1}^{l}\omega^{j_\kappa+1} = \omega^{j+l}$ since \(j_1+\cdots+j_l=j\). The number of positive integer solutions of \(j_1+\cdots+j_l=j\) is \(\binom{j-1}{l-1}\), and the even restriction can only reduce this count. Therefore, we have 
\begin{equation}
\Phi_{m,Q}^{\ast}\leq \omega^{2m+1}S_m.
\end{equation}
Since \(\Lambda_{m,Q}\leq\alpha_Q\Phi_{m,Q}^{\ast}\), we get
\begin{equation}
(\|\vec a\|_1\Lambda_{m,Q})^{1/(2m)} \leq \alpha_Q^{1/(2m)} (\|\vec a\|_1S_m)^{1/(2m)} \omega^{1+1/(2m)}.
\end{equation}
Together with the step-size restriction \(\rho_Q^{-1}\), this gives \eqref{eq:mu-mQ-root-growth-bound}. 
\end{proof}

The upper bound \(\omega\) in \cref{lem:simplified-root-growth-scale} is easier to compute in practice. The condition \(\rho_Q\omega<1\) ensures both uniform applicability and convergence of the simplified error profiles. Moreover, \(\mu_{m,Q}\) in \cref{eq:mu-mQ-root-growth-bound} is independent of \(m\) in most applications. Indeed, \(\alpha_Q\) is independent of \(m\), \(\omega\) is independent of \(m\) for applications considered in \cref{applications}, and \cref{sm-lemma}, proved in \cref{proof-of-sm-lemma}, gives \(\sup_{m\geq1} S_m<\infty\).

\begin{lem}\label{sm-lemma}
Consider the following expression
\begin{equation}\label{eq:simplified-profile-Sm-def}
S_m(x) := \sum_{\substack{j\in2\mathbb Z_+\\ j\geq 2m}} \sum_{l=1}^{m} \frac{1}{l!} \binom{j-1}{l-1} x^{j+l-(2m+1)}
\end{equation}
for $0 < x < 1$.  We have $\sup_{m\geq1}S_m(x)^{1/(2m)}<\infty$.   
\end{lem}

Since the MPF coefficients satisfy \(\|\vec a\|_1=\mathcal O(\log m)\), we also have $\sup_{m\geq1}\|\vec a\|_1^{1/(2m)}<\infty$. This shows that \(\sup_{m \geq 1} (\|\vec a\|_1 S_m)^{1/(2m)} < \infty\). Thus, the sufficient condition \(\rho_Q\omega<1\) gives not only fixed-\(m\) convergence of the simplified profile, but also uniform-in-\(m\) boundedness required to invoke the optimized-order estimate.

\subsection{Local Hamiltonian Specialization}\label{spec-mizuta}
We next use the locality-based MPF estimate from \cref{mizuta-time-indep-short}. Unlike \cref{aftab-time-indep-short}, it includes the remainder \(R_m(\delta)=K_m\delta\), where \(K_m=\|\vec a\|_1\|\vec b\|_1\). Since \cref{mizuta-time-indep-short} gives \(p_0=\mathcal O(\log(3N/\delta))\), set \(\delta_{p_0}:=3Ne^{-p_0}\). We choose \(p_0>\log(3N)\), so \(\delta_{p_0}\in(0,1)\). For \(p_0\in\mathbb N\), define
\begin{equation}\label{loc-mu-bar-def}
\bar\mu_{m,Q}(p_0) := \max_{i\in\mathcal I_Q} \mu_{m,p_0}(G_{k_i}), 
\end{equation}
where \(\mu_{m,p_0}\) is defined in \cref{mu-m-p0}.  Since \(H\) is \(q_H\)-local and \(g_H\)-extensive, and \(L\) is \(q_L\)-local and \(g_L\)-extensive, set \(q_Q := \max\{q_H,q_L\}\) and \(g_Q := g_H + R_Q g_L\). Then \(G_{k_i}=H+k_iL\) is uniformly \(q_Q\)-local and \(g_Q\)-extensive for all \(i\in\mathcal I_Q\). The corresponding step-size threshold is
\begin{equation}\label{loc-step-size}
\Delta_{Q,p_0}:=\min\left\{\frac{1}{16e^3p_0q_Qg_Q},\frac{1}{4\bar\mu_{m,Q}(p_0)}\right\},
\end{equation}
Set \(C_m=2e^{1/2}\|\vec a\|_1\). The quantities appearing in the complexity estimate are
\begin{align}\label{loc-Lambda-Phi-def}
\Lambda_{m,Q}(p_0):=\sum_{i\in\mathcal I_Q}
|v_i|\,\mu_{m,p_0}(G_{k_i})^{2m+1},
\quad 
\Phi^{\ast}_{m,Q}(p_0):= \bar\mu_{m,Q}(p_0)^{2m+1}.
\end{align}
For fixed \(m\ge1\) and \(p_0\in\mathbb N\), define \begin{equation}\label{loc-r-choice} r_{Q,p_0}:=\left\lceil\max\left\{\frac{T}{\Delta_{Q,p_0}}, \left( 2C_mT^{2m+1} \Phi^{\ast}_{m,Q}(p_0) \right)^{1/2m}, \left( \frac{ 2e C_m T^{2m+1} \Lambda_{m,Q}(p_0) }{\epsilon_Q} \right)^{1/2m} \right\} \right\rceil
\end{equation}
In \cref{loc-mu-growth-envelope}, proved in \cref{proof-of-loc-mu-growth-envelope}, we derive a growth condition on \(\mu_{m,p_0}\) needed to control \(r_{Q,p_0}\) before applying the results of \cref{abs-rel}.

\begin{lem}\label{loc-mu-growth-envelope}
Let \(H\) be a \(q\)-local, \(g\)-extensive Hamiltonian on \(N\) sites. For every \(p_0\geq2\), the quantity \(\mu_{m,p_0}(H)\) satisfies
\begin{equation}
\mu_{m,p_0}(H) \leq C_{q,g,N}(1+p_0)^{4/3},
\label{eq:local-growth-envelope-from-mizuta}
\end{equation}
where one may take $C_{q,g,N} := \max\{1,2qg\}\max\{1,Ng\}^{1/3}$. Consequently, for a quadrature family \(\{G_{k_i}:i\in \mathcal I_Q\}\) which is uniformly \(q_Q\)-local and \(g_Q\)-extensive, one has
\begin{equation}
\overline\mu_{m,Q}(p_0) := \max_{i\in \mathcal I_Q}\mu_{m,p_0}(G_{k_i}) \leq C_{\mu,Q,m}(1+p_0)^{4/3},
\label{eq:mubar-local-growth-envelope-from-mizuta}
\end{equation}
with $C_{\mu,Q,m} := \max\{1,2q_Q g_Q\}\max\{1,Ng_Q\}^{1/3}$. 
\end{lem}

By \cref{loc-mu-growth-envelope}, \(\bar\mu_{m,Q}(p_0)\), and hence \(\Delta_{Q,p_0}^{-1}\), grow at most polynomially in \(p_0\). By \cref{loc-Lambda-Phi-def}, we have $\Phi_{m,Q}^{\ast}(p_0)\leq \bar\mu_{m,Q}(p_0)^{2m+1}$ and $\Lambda_{m,Q}(p_0)\leq \alpha_Q\bar\mu_{m,Q}(p_0)^{2m+1}$. Thus, \(\Phi_{m,Q}^{\ast}(p_0)\), \(\Lambda_{m,Q}(p_0)\), and \(r_{Q,p_0}\) grow at most polynomially in \(p_0\). Since \(\delta_{p_0}=3Ne^{-p_0}\), we have $r_{Q,p_0}K_m\delta_{p_0}\to 0$ as \(p_0\to\infty\). Hence, there exists \(p_0>\log(3N)\) such that
\begin{equation}\label{loc-p0-admissible}
K_m\delta_{p_0} \le \min\left\{ \frac{1}{2r_{Q,p_0}}, \frac{\epsilon_Q}{2e\,\alpha_Qr_{Q,p_0}} \right\}.
\end{equation}
Under \cref{mizuta-time-indep-short}, \(R_m(\delta)=K_m\delta\). By \cref{loc-p0-admissible}, we have
\begin{equation}
r_mR_m(\delta_m)=r_mK_m\delta_m\leq \frac12, \qquad r_m\alpha_QR_m(\delta_m)=r_m\alpha_QK_m\delta_m\leq \frac{\epsilon_Q}{2e}.
\end{equation}
Thus, the accumulated locality-remainder conditions in \cref{abstract-fixed-order-arbitrary-quadrature} hold. By \cref{loc-r-choice},
\begin{equation}
\frac{T}{r_m}\leq \Delta_{Q,p_0}, \qquad \frac{C_mT^{2m+1}\Phi_{m,Q}^{\ast}(p_0)}{r_m^{2m}}\leq \frac12, \qquad \frac{C_mT^{2m+1}\Lambda_{m,Q}(p_0)}{r_m^{2m}}\leq \frac{\epsilon_Q}{2e}.
\end{equation}
Therefore, we have
\begin{equation}
r_m\eta_Q(T/r_m) = \frac{C_mT^{2m+1}\Phi_{m,Q}^{\ast}(p_0)}{r_m^{2m}} + r_mR_m(\delta_m) \leq 1.
\end{equation}
Combining the accumulated-remainder condition with the third inequality above gives the post-quadrature MPF error budget required in \cref{abstract-fixed-order-arbitrary-quadrature}. Hence, all hypotheses of \cref{abstract-fixed-order-arbitrary-quadrature} hold, and the fixed-order error and complexity estimates follow. 
Define
\begin{equation}\label{loc-mu-mQ-def}
\mu_{m,Q} := \max\left\{ \frac{1}{\Delta_{Q,p_0(m)}}, ( C_m\Phi_{m,Q}^{\ast}(p_0(m)) )^{1/2m}, \left( C_m \Lambda_{m,Q}(p_0(m)) \right)^{1/2m} \right\}.
\end{equation}
For the optimized statement, apply \cref{abstract-optimized-order-arbitrary-quadrature} with the admissible parameters \(\delta_m=\delta_{p_0(m)}\), \(r_Q(\delta_m)=r_m\), and \(\mu_{m,Q}(\delta_m)=\mu_{m,Q}\). The growth condition \(\mu_{m,Q}\leq \mu_QF(m)\) verifies the hypothesis in \cref{F-cond} for the chosen $m$.

\subsection{Quadrature Rules}\label{quad-rules} 
We now fix a quadrature rule \(Q\). For a family \(\{Q_\eta\}\) of quadrature rules, we choose the quadrature parameters \(\eta\) so that \(E_{\operatorname{quad}}(Q_\eta)<\epsilon_{\operatorname{comb}}\), leaving the residual budget $\epsilon_{Q_\eta}:=\epsilon_{\operatorname{comb}} - E_{\operatorname{quad}}(Q_\eta)>0$.  The final choice of quadrature parameters is then governed by the complexity estimates in \cref{abstract-fixed-order-arbitrary-quadrature,abstract-optimized-order-arbitrary-quadrature}.

\subsubsection{Uniform Trapezoidal Rule}
\label{trap-rule-subsub}
We first specialize to the uniform trapezoidal rule. Following the Low--Somma discretization strategy~\cite{low2025optimallchs}, we compare the infinite LCHS integral with the infinite trapezoidal sum using strip analyticity, and then truncate the lattice to obtain the quadrature rule used in the LCU implementation. We use the notation
\begin{equation}\label{eq:FT-def}
F_T(z) := \frac{1}{\sqrt{2\pi}}\hat f_{a,b}(z;c,d)U_z(T).
\end{equation}
Then \(W_\infty = \frac{1}{\sqrt{2\pi}}\int_{\mathbb{R}} \hat f_{a,b}(k;c,d)U_k(T)\,dk = \int_{\mathbb{R}} F_T(k)\,dk\). We first state \cref{general-kernel-strip-bound}, which is proved in \cref{proof-of-general-kernel-strip-bound},  a basic estimate used in the complexity analysis of the uniform trapezoidal quadrature rule.

\begin{lem}\label{general-kernel-strip-bound}
Fix a kernel profile \(\vec{\theta}=(a,b,c,d)\), and define
\begin{equation}\label{eq:rho-star-def}
\rho_\ast := \begin{cases} 1, & a=1,\\ \min\{1,b\}, & a>1. \end{cases}
\end{equation}
For every \(\rho\in(0,\rho_\ast)\), \(F_T\) is analytic on an open set containing the closed strip \(\overline S_\rho := \{z\in\mathbb C:\ |\operatorname{Im}z|\le \rho\}\). Moreover, \(F_T(x+i\beta)\) decays uniformly to zero on \(\overline S_\rho\) as \(|x|\to\infty\), and
\begin{equation}\label{eq:strip-L1-bound}
\sup_{|\beta|\le \rho}\int_{\mathbb R}\|F_T(x+i\beta)\|\,dx \le M_{a,b}(\rho),
\end{equation}
where $D_{a,b}(\rho) = (b-\rho)^{a-1}$ and
\begin{equation}\label{eq:Mab-sharp-def}
M_{a,b}(\rho) := \frac{c}{\sqrt{\pi}}\frac{(b+1)^{a-1}}{(1-\rho)D_{a,b}(\rho)}\exp\left(d+\frac{\rho^2-1}{4c^2}+\rho\max\{-d,d+T\|L\|\}\right).
\end{equation}
\end{lem}

We use the uniform trapezoidal rule because it converges exponentially for functions analytic in a strip around the real axis \cite{low2025optimallchs}. This is formalized below.

\begin{lem}\label{infinite-line-trapezoidal-estimate}
Let \(\rho>0\). Suppose that \(g\) is analytic on an open set containing \(\overline S_\rho\) and that \(g\) decays uniformly to zero on \(\overline S_\rho\) as \(|z|\to\infty\). If \(M>0\) satisfies $\int_{\mathbb R}\|g(x+i\beta)\|\,dx\le M$
for every \(\beta\in(-\rho,\rho)\), then, for every \(h>0\), we have
\begin{equation}\label{eq:trap-infinite-line}
\left\|\int_{\mathbb R}g(x)\,dx-h\sum_{q\in\mathbb Z}g(qh)\right\|\le \frac{2M}{e^{2\pi \rho/h}-1}.
\end{equation}
\end{lem}

\begin{proof}
This is the exponentially convergent trapezoidal estimate of Trefethen--Weideman~\cite[Theorem~5.1]{trefethen2014exponentially}, in the matrix-valued form used by Low--Somma~\cite[Lemma~10]{low2025optimallchs}.
\end{proof}

We now analyze the uniform trapezoidal rule in our setting. Fix \(R>0\) and \(h>0\), and set $n_h:= \left\lceil\frac{R}{h}\right\rceil$ and $R_h:=n_hh$.  Then \(R_h\ge R\) and \(R_h<R+h\). Define the finite uniform trapezoidal rule
\begin{equation}\label{eq:uniform-trap-Q-def}
Q_{R,h}^{\operatorname{trap}}
:=
\{(k_i,w_i): k_i=qh,\ w_i=h,\ q=-n_h,\ldots,n_h\}.
\end{equation}
For this rule, \(R_{Q_{R,h}^{\operatorname{trap}}}=R_h\). The LCHS quadrature coefficients are \(v_q=\frac{h}{\sqrt{2\pi}}\hat f_{a,b}(qh;c,d)\). The corresponding ideal quadrature operator is
\begin{equation}\label{eq:uniform-lattice-rule}
W_{R,h}^{\operatorname{trap},\operatorname{ideal}} := \frac{h}{\sqrt{2\pi}}\sum_{q=-n_h}^{n_h}\hat f_{a,b}(qh;c,d)U_{qh}(T) = h\sum_{q=-n_h}^{n_h}F_T(qh).
\end{equation}
The associated LCU normalization and discrete moments are
\begin{equation}\label{eq:trap-alpha-def}
\alpha_{R,h}^{\operatorname{trap}} := \frac{h}{\sqrt{2\pi}}\sum_{q=-n_h}^{n_h}|\hat f_{a,b}(qh;c,d)|, \quad M_{\nu,R,h}^{\operatorname{trap}} := \frac{h}{\sqrt{2\pi}}\sum_{q=-n_h}^{n_h}|\hat f_{a,b}(qh;c,d)|\,|qh|^\nu,
\end{equation}
for \(\nu\ge0\). Moreover, the omitted lattice tail is
\begin{equation}\label{eq:lattice-tail}
E_{\operatorname{tail},h}(R_h) := \frac{h}{\sqrt{2\pi}}\sum_{|q|>n_h}|\hat f_{a,b}(qh;c,d)|.
\end{equation}
\cref{uniform-lattice-quadrature} bounds the quadrature errors \(E_{\operatorname{tail},h}(R_h)\) and \(E_{\operatorname{trap}}(h,\rho)\), the latter arising from \cref{infinite-line-trapezoidal-estimate}.

\begin{prop}\label{uniform-lattice-quadrature}
Fix a kernel profile \(\vec{\theta}=(a,b,c,d)\), and let \(\rho\in(0,\rho_\ast)\), where \(\rho_\ast\) is defined in \cref{eq:rho-star-def}. Let \(R>0\), \(h>0\). Then
\begin{equation}\label{eq:uniform-lattice-error}
E_{\operatorname{quad}}(Q_{R,h}^{\operatorname{trap}}) = \|W_\infty-W_{R,h}^{\operatorname{trap},\operatorname{ideal}}\| \le E_{\operatorname{trap}}(h,\rho)+E_{\operatorname{tail},h}(R_h),
\end{equation}
where \(E_{\operatorname{trap}}(h,\rho):=\frac{2M_{a,b}(\rho)}{e^{2\pi \rho/h}-1}\). Moreover, we have $E_{\operatorname{tail},h}(R_h)\le E_{\operatorname{trunc}}(R_h)\le E_{\operatorname{trunc}}(R)$. 
Consequently, we have
\begin{equation}\label{eq:uniform-lattice-error-simple}
E_{\operatorname{quad}}(Q_{R,h}^{\operatorname{trap}})\le E_{\operatorname{trap}}(h,\rho)+E_{\operatorname{trunc}}(R).
\end{equation}
\end{prop}

\begin{proof}
Applying \cref{infinite-line-trapezoidal-estimate} to \(g=F_T\), and using \cref{general-kernel-strip-bound} with \(M=M_{a,b}(\rho)\), gives
\begin{equation}
\left\|\int_{\mathbb R}F_T(x)\,dx-h\sum_{q\in\mathbb Z}F_T(qh)\right\|\le E_{\operatorname{trap}}(h,\rho).
\end{equation}
Therefore,
\begin{align}
\left\|W_\infty-W_{R,h}^{\operatorname{trap},\operatorname{ideal}}\right\| &\le \left\|\int_{\mathbb R}F_T(x)\,dx-h\sum_{q\in\mathbb Z}F_T(qh)\right\|+h\sum_{|q|>n_h}\|F_T(qh)\| \\
&\le E_{\operatorname{trap}}(h,\rho)+E_{\operatorname{tail},h}(R_h).
\end{align}
It remains to compare the lattice tail with the continuous truncation error. For real \(k\), \(U_k(T)\) is unitary and
\begin{equation}\label{eq:real-kernel-modulus}
|\hat f_{a,b}(k;c,d)| = \frac{(b+1)^{a-1}e^{d-\frac{1}{4c^2}}}{\sqrt{2\pi}}\frac{e^{-k^2/(4c^2)}}{\sqrt{1+k^2}(b^2+k^2)^{(a-1)/2}}.
\end{equation}
This function is even. For \(k>0\),
\begin{equation}
\frac{d}{dk}\log|\hat f_{a,b}(k;c,d)| = -\frac{k}{2c^2}-\frac{k}{1+k^2}-(a-1)\frac{k}{b^2+k^2}<0.
\end{equation}
Hence \(k\mapsto|\hat f_{a,b}(k;c,d)|\) is decreasing on \([0,\infty)\). Since \(R_h=n_hh\), monotonicity gives
\begin{equation}
h\sum_{q=n_h+1}^{\infty}|\hat f_{a,b}(qh;c,d)| \le \int_{R_h}^{\infty}|\hat f_{a,b}(k;c,d)|\,dk.
\end{equation}
The same argument applies on the negative half-line. Multiplying by \(1/\sqrt{2\pi}\) gives \(E_{\operatorname{tail},h}(R_h)\le E_{\operatorname{trunc}}(R_h)\). Finally, since \(R_h\ge R\), monotonicity of \(E_{\operatorname{trunc}}\) gives \(E_{\operatorname{trunc}}(R_h)\le E_{\operatorname{trunc}}(R)\).
\end{proof}

\cref{uniform-lattice-quadrature} bounds the quadrature error by mesh and tail errors. Then \cref{trap-budget} gives sufficient conditions for controlling both. In what follows, we use

\begin{align}
C_{\operatorname{tail}} := \frac{c}{\sqrt{\pi}}\left(\frac{b+1}{b}\right)^{a-1}e^{d-\frac{1}{4c^2}}, \quad
C_{\operatorname{ker}} := \frac{1}{2\pi}\left(\frac{b+1}{b}\right)^{a-1}e^{d-\frac{1}{4c^2}}, \quad
\widetilde C_{\operatorname{ker}} := \frac{(b+1)^{a-1}}{2\pi}e^{d-\frac{1}{4c^2}}.
\end{align}
\cref{trap-budget} bounds the quantities that determine the post-quadrature parameters for the trapezoidal quadrature rule.

\begin{cor}\label{trap-budget}
Let \(\rho\in(0,\rho_\ast)\) and \(\epsilon_{\operatorname{mesh}},\epsilon_{\operatorname{tail}}>0\). Define
\begin{equation}\label{eq:tail-log-def}
\ell_{\operatorname{tail}} := \max\left\{1,\log\left(\frac{C_{\operatorname{tail}}}{\epsilon_{\operatorname{tail}}}\right)\right\}, \quad \ell_{\operatorname{mesh}} := \log\left(1+\frac{2M_{a,b}(\rho)}{\epsilon_{\operatorname{mesh}}}\right).
\end{equation}
Choose \(h=\min\left\{1,2c,\frac{2\pi\rho}{\ell_{\operatorname{mesh}}}\right\}\) and \(R=2c\sqrt{\ell_{\operatorname{tail}}}\), and set \(n_h=\lceil R/h\rceil\), \(R_h=n_hh\). Then $E_{\operatorname{quad}}(Q_{R,h}^{\operatorname{trap}})\le \epsilon_{\operatorname{mesh}}+\epsilon_{\operatorname{tail}}$. Moreover, we have
\begin{align}
|\mathcal I_{Q_{R,h}^{\operatorname{trap}}}| &= 2n_h+1 = \mathcal O\left(1+(1+c)\sqrt{\ell_{\operatorname{tail}}}+\frac{c}{\rho}\sqrt{\ell_{\operatorname{tail}}}\,\ell_{\operatorname{mesh}}\right), \label{eq:trap-comparison-cardinality-1}\\
R_{Q_{R,h}^{\operatorname{trap}}} &= R_h = \mathcal O\left(c\sqrt{\ell_{\operatorname{tail}}}+\min\left\{1,2c,\frac{\rho}{\ell_{\operatorname{mesh}}}\right\}\right). \label{eq:trap-comparison-cardinality-2}
\end{align}
For a fixed admissible kernel profile and balanced budgets \(\epsilon_{\operatorname{mesh}}\asymp\epsilon_{\operatorname{tail}}\asymp\epsilon_q\), we obtain
\begin{equation}\label{eq:trap-balanced-cardinality}
|\mathcal I_{Q_{R,h}^{\operatorname{trap}}}| = \mathcal O\left(\frac{c}{\rho}\log^{3/2}\left(\frac1{\epsilon_q}\right)\right), \quad R_{Q_{R,h}^{\operatorname{trap}}} = \mathcal O\left(c\sqrt{\log\left(\frac1{\epsilon_q}\right)}\right),
\end{equation}
for fixed \(\rho\) and \(c\). For \(\nu\ge0\), define
\begin{equation}\label{eq:trap-sharp-moment-scale}
\mathfrak M_{\nu,a,b}(c) := \begin{cases} 1, & \nu<a-1,\\ 1+\log(1+c), & \nu=a-1,\\ 1+c^{\nu-a+1}, & \nu>a-1. \end{cases}
\end{equation}
Then there exists a constant \(C_{\nu,a,b}^{\operatorname{trap}}>0\), depending only on \(\nu,a,b\), such that
\begin{equation}\label{eq:trap-sharp-moment-bound}
M_{\nu,R,h}^{\operatorname{trap}}\le \widetilde C_{\operatorname{ker}}C_{\nu,a,b}^{\operatorname{trap}}\mathfrak M_{\nu,a,b}(c).
\end{equation}
In particular, \(\alpha_{R,h}^{\operatorname{trap}}=M_{0,R,h}^{\operatorname{trap}}\le \widetilde C_{\operatorname{ker}}C_{0,a,b}^{\operatorname{trap}}\mathfrak M_{0,a,b}(c)\). Equivalently, the trapezoidal moments satisfy
\begin{equation}\label{eq:trap-sharp-moment-asymptotic}
M_{\nu,R,h}^{\operatorname{trap}} = \begin{cases} \mathcal O(\widetilde C_{\operatorname{ker}}), & \nu<a-1,\\ \mathcal O(\widetilde C_{\operatorname{ker}}\log(1+c)), & \nu=a-1,\\ \mathcal O(\widetilde C_{\operatorname{ker}}c^{\nu-a+1}), & \nu>a-1. \end{cases}
\end{equation}
\end{cor}

\begin{proof}
The choice of \(h\) controls the mesh error. By \cref{infinite-line-trapezoidal-estimate}, \(E_{\operatorname{trap}}(h,\rho)\le \epsilon_{\operatorname{mesh}}\). It remains to control the truncation error. By \cref{radii-lower-bd}, \(E_{\operatorname{trunc}}(R)\le C_{\operatorname{tail}}\operatorname{erfc}(R/(2c))\). Since \(\operatorname{erfc}(x)\le e^{-x^2}\), the condition \(R\ge 2c\sqrt{\ell_{\operatorname{tail}}}\) implies \(E_{\operatorname{trunc}}(R)\le \epsilon_{\operatorname{tail}}\). Combining these estimates with \cref{uniform-lattice-quadrature} gives
\begin{equation}
E_{\operatorname{quad}}(Q_{R,h}^{\operatorname{trap}})\le E_{\operatorname{trap}}(h,\rho)+E_{\operatorname{trunc}}(R)\le \epsilon_{\operatorname{mesh}}+\epsilon_{\operatorname{tail}}.
\end{equation}
Since \(n_h=\lceil R/h\rceil\), we have \(|\mathcal I_{Q_{R,h}^{\operatorname{trap}}}|=2n_h+1\le 2R/h+3\). Moreover,
\begin{equation}
\frac1h = \max\left\{1,\frac1{2c},\frac{\ell_{\operatorname{mesh}}}{2\pi\rho}\right\} = \mathcal O\left(1+\frac1c+\frac{\ell_{\operatorname{mesh}}}{\rho}\right).
\end{equation}
Substituting \(R=2c\sqrt{\ell_{\operatorname{tail}}}\) gives the bound on \(|\mathcal I_{Q_{R,h}^{\operatorname{trap}}}|\) in \cref{eq:trap-comparison-cardinality-1}. Moreover, \(R_h=n_hh\le R+h\), giving the bound on \(R_{Q_{R,h}^{\operatorname{trap}}}\) in \cref{eq:trap-comparison-cardinality-2}. The balanced-budget estimates follow from \(\ell_{\operatorname{mesh}}=\ell_{\operatorname{tail}}=\mathcal O(\log(1/\epsilon_q))\), for fixed \(\rho\) and \(c\). It remains to bound the LCU normalization and discrete moments. For real \(k\), \cref{eq:real-kernel-modulus} implies
\begin{equation}\label{eq:kernel-sharp-majorant}
\frac{1}{\sqrt{2\pi}}|\hat f_{a,b}(k;c,d)| = \widetilde C_{\operatorname{ker}}\frac{e^{-k^2/(4c^2)}}{\sqrt{1+k^2}(b^2+k^2)^{(a-1)/2}}.
\end{equation}
Therefore,
\begin{equation}\label{eq:trap-moment-sharp-start}
M_{\nu,R,h}^{\operatorname{trap}}\le \widetilde C_{\operatorname{ker}}h\sum_{q\in\mathbb Z}\frac{e^{-(qh)^2/(4c^2)}|qh|^\nu}{\sqrt{1+(qh)^2}(b^2+(qh)^2)^{(a-1)/2}}.
\end{equation}
For \(k>0\), we have
\begin{equation}\label{eq:trap-denom-min-bound}
\frac{1}{\sqrt{1+k^2}(b^2+k^2)^{(a-1)/2}}\le \min\left\{\frac{1}{b^{a-1}},\frac{1}{k^a}\right\}.
\end{equation}
Indeed, the first bound follows from \(\sqrt{1+k^2}\ge1\) and \((b^2+k^2)^{(a-1)/2}\ge b^{a-1}\), while the second follows from \(\sqrt{1+k^2}\ge k\) and \((b^2+k^2)^{(a-1)/2}\ge k^{a-1}\). Let
\begin{equation}\label{eq:trap-kappa-b-def}
\kappa_b := \begin{cases} 1, & a=1,\\ b^{(a-1)/a}, & a>1. \end{cases}
\end{equation}
We split the full-line sum in \cref{eq:trap-moment-sharp-start} into the regions \(|qh|\le \kappa_b\) and \(|qh|>\kappa_b\). Since \(h\le1\), the first contribution is bounded by a constant depending only on \(\nu,a,b\). More precisely,
\begin{equation}
h\sum_{|qh|\le\kappa_b}\frac{e^{-(qh)^2/(4c^2)}|qh|^\nu}{\sqrt{1+(qh)^2}(b^2+(qh)^2)^{(a-1)/2}}\le C_{\nu,a,b}.
\end{equation}
On the complementary region, \cref{eq:trap-denom-min-bound} gives
\begin{equation}
h\sum_{|qh|>\kappa_b}\frac{e^{-(qh)^2/(4c^2)}|qh|^\nu}{\sqrt{1+(qh)^2}(b^2+(qh)^2)^{(a-1)/2}}\le h\sum_{|qh|>\kappa_b}e^{-(qh)^2/(4c^2)}|qh|^{\nu-a}.
\end{equation}
The latter Riemann sum has three regimes. If \(\nu<a-1\), then \(|k|^{\nu-a}\) is integrable at infinity, and hence
\begin{equation}
h\sum_{|qh|>\kappa_b}e^{-(qh)^2/(4c^2)}|qh|^{\nu-a}\le C_{\nu,a,b}.
\end{equation}
If \(\nu=a-1\), then the summand behaves like \(1/|qh|\) until the Gaussian cutoff scale \(|qh|\sim c\), and therefore
\begin{equation}
h\sum_{|qh|>\kappa_b}e^{-(qh)^2/(4c^2)}|qh|^{-1}\le C_{a,b}\bigl(1+\log(1+c)\bigr).
\end{equation}
If \(\nu>a-1\), comparison with the corresponding Gaussian integral gives
\begin{align}
h\sum_{|qh|>\kappa_b}e^{-(qh)^2/(4c^2)}|qh|^{\nu-a} &\le C_{\nu,a,b}+C_\nu\int_0^\infty e^{-k^2/(4c^2)}k^{\nu-a}\,dk \le C_{\nu,a,b}+C_\nu c^{\nu-a+1}.
\end{align}
Combining these three cases with \cref{eq:trap-moment-sharp-start} proves \cref{eq:trap-sharp-moment-bound}. The normalization bound is the case \(\nu=0\), and the asymptotic statement follows directly from the definition of \(\mathfrak M_{\nu,a,b}(c)\).
\end{proof}

\cref{feasible-kernel-profile-explicit} shows that, when the kernel parameters satisfy $a,d = \mathcal{O}(1)$, one may take $c = \Omega(\sqrt{\log(1/\epsilon_{\operatorname{approx}})})$. Assuming that all error-tolerance parameters are of the same order as $\epsilon > 0$, \cref{trap-budget} implies
\begin{equation}
|\mathcal I_{Q_{R,h}^{\operatorname{trap}}}| = \mathcal O\left(\frac{1}{\rho}\log^{2}\left(\frac1{\epsilon_q}\right)\right), \quad R_{Q_{R,h}^{\operatorname{trap}}} = \mathcal O\left(\log\left(\frac1{\epsilon_q}\right)\right).
\end{equation}
Moreover, \cref{trap-budget} yields exponential mesh-error decay in $1/h$, Gaussian tail decay in $R$, and moment estimates that retain the rational decay of the kernel, improving the crude Gaussian scaling $\mathcal O((\log(1/\epsilon_q))^{(\nu+1)/2})$ to the $a$-dependent scaling in \cref{eq:trap-sharp-moment-asymptotic}.

\subsubsection{Sinh--Sinh Trapezoidal Rule}
\label{sinh-trap-rule-subsub} 
We next consider a non-uniform quadrature rule obtained by applying the trapezoidal rule after the change of variables \(k=\phi_\eta(x):=\eta\sinh x\), where \(0<\eta\le 2c\) is a free scale parameter. Set \(\lambda_\eta:=\eta/(2c)\). Under this change of variables, the Gaussian factor in \(\hat f_{a,b}\) becomes \(e^{-\phi_\eta(x)^2/(4c^2)}=e^{-\lambda_\eta^2\sinh^2 x}\), so the transformed integrand has double-exponential decay along the real axis. The choice \(\eta=2c\) gives the Gaussian-normalized case \(e^{-\lambda_\eta^2\sinh^2 x}=e^{-\sinh^2 x}\). See~\cite[Section~15]{trefethen2014exponentially} for more on this quadrature rule. Define
\begin{equation}\label{eq:sinh-transformed-integrand}
G_{T,\eta}(z) := F_T(\phi_\eta(z)) \phi_\eta'(z) = \eta\cosh z\, F_T(\eta\sinh z).
\end{equation}
Then, by the change of variables \(k=\phi_\eta(x)\), we have \(W_\infty=\int_{\mathbb R}F_T(k)\,dk=\int_{\mathbb R}G_{T,\eta}(x)\,dx\). The exponential convergence of the sinh--sinh rule follows by applying the infinite-line trapezoidal estimate to the transformed integrand \(G_{T,\eta}\). The relevant analyticity strip is now an \(x\)-strip rather than the original \(k\)-strip. Let
\begin{equation}\label{eq:sinh-beta-star-def}
\beta_{\ast,\eta} := \sup\left\{\beta>0: G_{T,\eta} \text{ is analytic on an open set containing } \{z\in\mathbb C:|\operatorname{Im}z|\le\beta\}\right\}.
\end{equation}
Equivalently, \(\beta_{\ast,\eta}\) is limited by the nearest preimage of the poles \(k=-i\) and \(k=ib\) under \(z\mapsto \eta\sinh z\), and by the requirement that the transformed Gaussian decay on horizontal translates. In particular, one may restrict to \(\beta<\pi/4\), since
\begin{equation}
\operatorname{Re}(\sinh^2(x+i\beta)) = \sinh^2x\cos(2\beta)-\sin^2\beta
\end{equation}
has positive quadratic growth in \(|x|\) whenever \(|\beta|<\pi/4\). 
For \(\beta\in(0,\min\{\beta_{\ast,\eta},\pi/4\})\), define
\begin{equation}\label{eq:sinh-denom-distances}
\delta_{1,\eta}(\beta) := \inf_{\substack{x\in\mathbb R\\ |\omega|\le \beta}}\left|1-i\phi_\eta(x+i\omega)\right|, \qquad \delta_{b,\eta}(\beta) := \inf_{\substack{x\in\mathbb R\\ |\omega|\le \beta}}\left|b+i\phi_\eta(x+i\omega)\right|.
\end{equation}
Here \(\delta_{b,\eta}(\beta)\) is only needed when \(a>1\). Since \(\beta<\beta_{\ast,\eta}\), these quantities are strictly positive. Set
\begin{equation}\label{eq:sinh-rho-A-def}
\rho_\beta:=\cos(2\beta)>0, \qquad A_{\eta,\beta} := \eta\sin\beta\,\bigl(|d|+T\|L\|\bigr).
\end{equation}
Define
\begin{equation}\label{eq:Msinh-explicit-def}
M_{\operatorname{sinh},\eta}(\beta) := \frac{\eta (b+1)^{a-1}}{2\pi\,\delta_{1,\eta}(\beta)\,\delta_{b,\eta}(\beta)^{a-1}}\exp\left(d-\frac{1}{4c^2}+\lambda_\eta^2\sin^2\beta+A_{\eta,\beta}+\frac{A_{\eta,\beta}^2}{2\lambda_\eta^2\rho_\beta}\right)\sqrt{\frac{2\pi}{\lambda_\eta^2\rho_\beta}},
\end{equation}
with the convention \(\delta_{b,\eta}(\beta)^{a-1}=1\) when \(a=1\).

\begin{lem}\label{sinh-strip-bound}
Let \(0<\eta\le 2c\) and \(\beta\in(0,\min\{\beta_{\ast,\eta},\pi/4\})\). Then \(G_{T,\eta}\) is analytic on an open set containing the closed strip \(\overline S_\beta := \{z\in\mathbb C:|\operatorname{Im}z|\le\beta\}\). Moreover, \(G_{T,\eta}(x+i\omega)\) decays uniformly to zero on \(\overline S_\beta\) as \(|x|\to\infty\), and
\begin{equation}\label{eq:sinh-strip-L1-bound}
\sup_{|\omega|\le\beta}\int_{\mathbb R}\|G_{T,\eta}(x+i\omega)\|\,dx \le M_{\operatorname{sinh},\eta}(\beta).
\end{equation}
Consequently, for every \(h>0\),
\begin{equation}\label{eq:sinh-infinite-trap-error}
\left\|\int_{\mathbb R}G_{T,\eta}(x)\,dx-h\sum_{q\in\mathbb Z}G_{T,\eta}(qh)\right\|\le \frac{2M_{\operatorname{sinh},\eta}(\beta)}{e^{2\pi\beta/h}-1}.
\end{equation}
\end{lem}

\begin{proof}
Let \(z=x+i\omega\), with \(|\omega|\le\beta\), and write \(\phi_\eta(z)=\eta\sinh z=u+iv\), where \(u=\eta\sinh x\cos\omega\) and \(v=\eta\cosh x\sin\omega\). Since \(\beta<\beta_{\ast,\eta}\), the transformed contour stays a positive distance from the pole \(k=-i\), and, when \(a>1\), from the pole \(k=ib\). Hence \(|1-i\phi_\eta(z)|\ge \delta_{1,\eta}(\beta)\) and, for \(a>1\), \(|b+i\phi_\eta(z)|\ge \delta_{b,\eta}(\beta)\). Moreover, \(|e^{d(1-i\phi_\eta(z))}|=e^{d(1+v)}\le e^d e^{|d||v|}\). Using the time-independent bound for \(U_k(T)\) as in \cref{general-kernel-strip-bound}, we have \(\|U_{\phi_\eta(z)}(T)\|\le e^{T\|L\|\,|v|}\). Since \(|v|\le \eta\sin\beta\,\cosh x\), the combined non-Gaussian exponential growth is bounded by \(\exp(A_{\eta,\beta}\cosh x)\). Next, we have
\begin{equation}
\operatorname{Re}(\sinh^2(x+i\omega)) = \sinh^2x\cos(2\omega)-\sin^2\omega \ge \rho_\beta\sinh^2x-\sin^2\beta.
\end{equation}
Therefore, the Gaussian factor satisfies
\begin{equation}
|e^{-(\phi_\eta(z)^2+1)/(4c^2)}| = e^{-\lambda_\eta^2\operatorname{Re}(\sinh^2 z)-1/(4c^2)} \le \exp\left(-\lambda_\eta^2\rho_\beta\sinh^2x+\lambda_\eta^2\sin^2\beta-\frac{1}{4c^2}\right).
\end{equation}
Also \(|\cosh(x+i\omega)|\le\cosh x\). Combining these estimates gives
\begin{equation}
\|G_{T,\eta}(x+i\omega)\| \le \frac{\eta (b+1)^{a-1}}{2\pi\,\delta_{1,\eta}(\beta)\,\delta_{b,\eta}(\beta)^{a-1}}\exp\left(d-\frac{1}{4c^2}+\lambda_\eta^2\sin^2\beta\right)\cosh x\exp\left(-\lambda_\eta^2\rho_\beta\sinh^2x+A_{\eta,\beta}\cosh x\right).
\end{equation}
Using \(\cosh x\le 1+|\sinh x|\) and \(A_{\eta,\beta}|\sinh x|\le \frac{\lambda_\eta^2\rho_\beta}{2}\sinh^2x+\frac{A_{\eta,\beta}^2}{2\lambda_\eta^2\rho_\beta}\), we obtain
\begin{equation}
-\lambda_\eta^2\rho_\beta\sinh^2x+A_{\eta,\beta}\cosh x \le -\frac{\lambda_\eta^2\rho_\beta}{2}\sinh^2x+A_{\eta,\beta}+\frac{A_{\eta,\beta}^2}{2\lambda_\eta^2\rho_\beta}.
\end{equation}
Hence,
\begin{align}
\int_{\mathbb R}\|G_{T,\eta}(x+i\omega)\|\,dx  \le \frac{\eta (b+1)^{a-1}}{2\pi\,\delta_{1,\eta}(\beta)\,\delta_{b,\eta}(\beta)^{a-1}} & \exp\left(d-\frac{1}{4c^2}+\lambda_\eta^2\sin^2\beta+A_{\eta,\beta}+\frac{A_{\eta,\beta}^2}{2\lambda_\eta^2\rho_\beta}\right) \\
& \times 
\int_{\mathbb R}\cosh x\,e^{-\frac{\lambda_\eta^2\rho_\beta}{2}\sinh^2x}\,dx.
\end{align}
With the substitution \(u=\sinh x\), \(du=\cosh x\,dx\), we get
\begin{equation}
\int_{\mathbb R}\cosh x\,e^{-\frac{\lambda_\eta^2\rho_\beta}{2}\sinh^2x}\,dx = \int_{\mathbb R}e^{-\frac{\lambda_\eta^2\rho_\beta}{2}u^2}\,du = \sqrt{\frac{2\pi}{\lambda_\eta^2\rho_\beta}}.
\end{equation}
This proves \cref{eq:sinh-strip-L1-bound} with \(M_{\operatorname{sinh},\eta}(\beta)\). The same bound also implies uniform decay as \(|x|\to\infty\). Finally, applying \cref{infinite-line-trapezoidal-estimate} to \(g=G_{T,\eta}\) gives \cref{eq:sinh-infinite-trap-error}.
\end{proof}

\cref{sinh-strip-bound} gives exponential mesh convergence in \(1/h\). The advantage of the sinh--sinh transformation is that the lattice truncation error decays double-exponentially in the \(x\)-truncation length. Fix \(Y>0\) and \(h>0\). As in \cref{trap-rule-subsub}, set $n_h := \left\lceil\frac{Y}{h}\right\rceil$ and $Y_h:=n_hh$.  Define the finite sinh--sinh quadrature rule
\begin{equation}\label{eq:sinh-Q-def}
Q_{Y,h}^{\operatorname{sinh},\eta} := \left\{(k_i,w_i): k_i=\eta\sinh(qh),\ w_i=\eta h\cosh(qh),\ q=-n_h,\ldots,n_h\right\}.
\end{equation}
For this rule, \(R_{Q_{Y,h}^{\operatorname{sinh},\eta}}=\eta\sinh(Y_h)\). 
The corresponding ideal quadrature operator is
\begin{equation}\label{eq:sinh-ideal-operator}
W_{Y,h}^{\operatorname{sinh},\eta,\operatorname{ideal}} := h\sum_{q=-n_h}^{n_h}G_{T,\eta}(qh) = \frac{\eta h}{\sqrt{2\pi}}\sum_{q=-n_h}^{n_h}\cosh(qh)\,\hat f_{a,b}(\eta\sinh(qh);c,d)U_{\eta\sinh(qh)}(T).
\end{equation}
The LCHS quadrature coefficients are \(v_q=\frac{\eta h\cosh(qh)}{\sqrt{2\pi}}\hat f_{a,b}(\eta\sinh(qh);c,d)\). The associated LCU normalization is \(\alpha_{Y,h}^{\operatorname{sinh},\eta}:=\sum_{q=-n_h}^{n_h}|v_q|\), and the discrete moments are \(M_{\nu,Y,h}^{\operatorname{sinh},\eta}:=\sum_{q=-n_h}^{n_h}|v_q|\,|\eta\sinh(qh)|^\nu\), for \(\nu\ge0\). Moreover, the omitted transformed-lattice tail is \(E_{\operatorname{tail},h}^{\operatorname{sinh},\eta}(Y_h):=h\sum_{|q|>n_h}\|G_{T,\eta}(qh)\|\).
\cref{sinh-lattice-quadrature} is the analogue of \cref{uniform-lattice-quadrature}. 

\begin{prop}\label{sinh-lattice-quadrature}
Let \(0<\eta\le 2c\) and \(\beta\in(0,\min\{\beta_{\ast,\eta},\pi/4\})\). Let \(Y>0\), \(h>0\). Then
\begin{equation}\label{eq:sinh-lattice-error}
E_{\operatorname{quad}}(Q_{Y,h}^{\operatorname{sinh},\eta}) = \left\|W_\infty-W_{Y,h}^{\operatorname{sinh},\eta,\operatorname{ideal}}\right\| \le E_{\operatorname{sinh},\eta}(h,\beta)+E_{\operatorname{tail},h}^{\operatorname{sinh},\eta}(Y_h),
\end{equation}
where \(E_{\operatorname{sinh},\eta}(h,\beta):=\frac{2M_{\operatorname{sinh},\eta}(\beta)}{e^{2\pi\beta/h}-1}\). Moreover, if \(\lambda_\eta\sinh Y_h\ge1\), then
\begin{equation}\label{eq:sinh-tail-bound}
E_{\operatorname{tail},h}^{\operatorname{sinh},\eta}(Y_h) \le 2c\sqrt{\pi}\,C_{\operatorname{ker}}\,\operatorname{erfc}(\lambda_\eta\sinh Y_h) \le 2c\sqrt{\pi}\,C_{\operatorname{ker}}\,e^{-\lambda_\eta^2\sinh^2 Y_h},
\end{equation}
where \(C_{\operatorname{ker}}:=\frac{1}{2\pi}\left(\frac{b+1}{b}\right)^{a-1}e^{d-\frac{1}{4c^2}}\).
\end{prop}

\begin{proof}
The mesh contribution follows from \cref{sinh-strip-bound}. For the tail, use the real-axis bound \(\|G_{T,\eta}(x)\|\le \eta C_{\operatorname{ker}}\cosh x\,e^{-\lambda_\eta^2\sinh^2x}\), which follows from the Gaussian majorant for \(|\hat f_{a,b}|\) and unitarity of \(U_k(T)\) for real \(k\). The function \(x\mapsto \cosh x\,e^{-\lambda_\eta^2\sinh^2x}\) is decreasing whenever \(\lambda_\eta\sinh x\ge1\). Hence, if \(\lambda_\eta\sinh Y_h\ge1\), then monotonicity gives
\begin{align}
E_{\operatorname{tail},h}^{\operatorname{sinh},\eta}(Y_h) &\le 2\int_{Y_h}^{\infty}\eta C_{\operatorname{ker}}\cosh x\,e^{-\lambda_\eta^2\sinh^2x}\,dx \\
&= 2\eta C_{\operatorname{ker}}\int_{\sinh Y_h}^{\infty}e^{-\lambda_\eta^2u^2}\,du = 2c\sqrt{\pi}\,C_{\operatorname{ker}}\,\operatorname{erfc}(\lambda_\eta\sinh Y_h).
\end{align}
The final inequality follows from \(\operatorname{erfc}(x)\le e^{-x^2}\).
\end{proof}

\cref{sinh-lattice-quadrature} shows that the sinh--sinh rule has exponential mesh convergence in \(1/h\) and double-exponential truncation decay in the transformed cutoff \(Y\). Finally, \cref{sinh-budget} bounds the quantities determining the implementation complexity of the sinh quadrature rule. Since the proof is similar to that of \cref{trap-budget}, it is given in \cref{proof-of-sinh-budget}.

\begin{cor}\label{sinh-budget}
Let \(0<\eta\le 2c\), \(\beta\in(0,\min\{\beta_{\ast,\eta},\pi/4\})\), and \(\epsilon_{\operatorname{mesh}},\epsilon_{\operatorname{tail}}>0\). Consider the following quantities:
\begin{equation}\label{eq:sinh-logs-def}
\ell_{\operatorname{mesh}}^{\operatorname{sinh},\eta} := \log\left(1+\frac{2M_{\operatorname{sinh},\eta}(\beta)}{\epsilon_{\operatorname{mesh}}}\right), \quad \ell_{\operatorname{tail}}^{\operatorname{sinh},\eta} := \max\left\{1,\log\left(\frac{2c\sqrt{\pi}C_{\operatorname{ker}}}{\epsilon_{\operatorname{tail}}}\right)\right\}.
\end{equation}
Choose
\[
h = \min\left\{1,\frac{2\pi\beta}{\ell_{\operatorname{mesh}}^{\operatorname{sinh},\eta}}\right\}, \qquad Y = \max\left\{1,\operatorname{arsinh}\left(\frac{\sqrt{\ell_{\operatorname{tail}}^{\operatorname{sinh},\eta}}}{\lambda_\eta}\right)\right\} = \max\left\{1,\operatorname{arsinh}\left(\frac{2c}{\eta}\sqrt{\ell_{\operatorname{tail}}^{\operatorname{sinh},\eta}}\right)\right\},
\]
and set \(n_h=\lceil Y/h\rceil\), \(Y_h=n_hh\). Then \(E_{\operatorname{quad}}(Q_{Y,h}^{\operatorname{sinh},\eta}) \le \epsilon_{\operatorname{mesh}}+\epsilon_{\operatorname{tail}}\). Moreover, we have
\begin{align}
|\mathcal I_{Q_{Y,h}^{\operatorname{sinh},\eta}}| &= 2n_h+1 = \mathcal O\left(1+\log\left(1+\frac{\sqrt{\ell_{\operatorname{tail}}^{\operatorname{sinh},\eta}}}{\lambda_\eta}\right)\left(1+\frac{\ell_{\operatorname{mesh}}^{\operatorname{sinh},\eta}}{\beta}\right)\right), \label{eq:sinh-cardinality-bound}\\
R_{Q_{Y,h}^{\operatorname{sinh},\eta}} &= \eta\sinh(Y_h) = \mathcal O\left(c\sqrt{\ell_{\operatorname{tail}}^{\operatorname{sinh},\eta}}+\eta\right). \label{eq:sinh-radius-bound}
\end{align}
For a fixed admissible kernel profile and balanced computational budgets \(\epsilon_{\operatorname{mesh}}\asymp \epsilon_{\operatorname{tail}}\asymp\epsilon_q\), we obtain
\begin{equation}\label{eq:sinh-balanced-cardinality-radius}
|\mathcal I_{Q_{Y,h}^{\operatorname{sinh},\eta}}| = \mathcal O\left(\frac{1}{\beta}\log\left(\frac1{\epsilon_q}\right)\log\left(1+\frac{2c}{\eta}\sqrt{\log\left(\frac1{\epsilon_q}\right)}\right)\right), \quad R_{Q_{Y,h}^{\operatorname{sinh},\eta}} = \mathcal O\left(c\sqrt{\log\left(\frac1{\epsilon_q}\right)}+\eta\right).
\end{equation}
Define \(\mathfrak M_{\nu,a,b}(c)\) as in \cref{eq:trap-sharp-moment-scale}. Then there exists a constant \(C_{\nu,a,b}^{\operatorname{sinh}}>0\), depending only on \(\nu,a,b\), such that
\begin{equation}\label{eq:sinh-sharp-moment-bound}
M_{\nu,Y,h}^{\operatorname{sinh},\eta}\le \widetilde C_{\operatorname{ker}}C_{\nu,a,b}^{\operatorname{sinh}}\mathfrak M_{\nu,a,b}(c).
\end{equation}
In particular, \(\alpha_{Y,h}^{\operatorname{sinh},\eta}=M_{0,Y,h}^{\operatorname{sinh},\eta}\le \widetilde C_{\operatorname{ker}}C_{0,a,b}^{\operatorname{sinh}}\mathfrak M_{0,a,b}(c)\). Equivalently, the sinh--sinh moments satisfy the same scaling as in \cref{eq:trap-sharp-moment-asymptotic}. 
\end{cor}

If one takes $a,d=\mathcal O(1)$, $c=\Theta(\sqrt{\log(1/\epsilon_{\operatorname{approx}})})$, and assumes that all error parameters are of order $\epsilon>0$, then \cref{sinh-budget} shows that, for a fixed free-scale choice $\eta=\Theta(\min(1,b))$, the transformed strip width satisfies $\beta^{-1}=\mathcal O(1)$. Therefore
\begin{equation}
|\mathcal I_{Q_{Y,h}^{\operatorname{sinh},\eta}}|
=
\mathcal O\left(
\log\left(\frac1{\epsilon_q}\right)
\log\log\left(\frac1{\epsilon_q}\right)
\right),
\end{equation}
which is asymptotically smaller than $\log^{3/2}(1/\epsilon_q)$.

\section{Applications}\label{applications}
In this section, we present applications of our algorithm and analyze its implementation complexity. \cref{fractional-diff} discusses the fractional diffusion equation with an imaginary potential, \cref{generalized-heat-dirichlet} discusses the advection--diffusion equation, and \cref{sec:no-jump-tfim} the no-jump dynamics regime in a dissipative Ising model.

\subsection{Fractional Diffusion Equation}\label{fractional-diff} Let \(0<s<1\). Consider the following partial differential equation:
\begin{equation}\label{eq:fractional-diffusion-pde}
\begin{cases} \partial_t u(t,x) = -(-\Delta)^s u(t,x) + (a(t,x)-iv(t,x))u(t,x) + b(t,x), & (t,x)\in (0,T]\times\mathbb T^d,\\[0.5em] u(0,x)=u_0(x), & x\in\mathbb T^d. \end{cases}
\end{equation}

Here \((-\Delta)^s\) is a fractional power of the positive semidefinite Laplacian on \(\mathbb T^d\), \(a,v:[0,T]\times\mathbb T^d\to\mathbb R\) are real-valued potentials, and \(b:[0,T]\times\mathbb T^d\to\mathbb C\) is a source term. If \(v\equiv0\), then \cref{eq:fractional-diffusion-pde} reduces to a linear fractional reaction--diffusion equation. The operator \((-\Delta)^s\) models non-local diffusion, \(a(t,x)\) gives a linear reaction term, and the imaginary potential \(-iv(t,x)\) produces local phase modulation, analogous to a Larmor-frequency-offset term in Bloch--Torrey-type models. Thus, \cref{eq:fractional-diffusion-pde} is a forced fractional Bloch--Torrey-type equation with an additional real potential. Fractional Bloch--Torrey models have been used for anomalous diffusion in complex media, especially in magnetic resonance applications~\cite{Torrey1956,Magin2008,Yu2013,BuenoOrovioBurrage2017,Moutal2020}. In what follows, we take \(a,b\equiv0\) and assume that \(v\) is time-independent.

\begin{remark}
See also \cite[Equation~10]{an2023LCHS}, which considers complex absorbing potentials in open-system dynamics through a similar Schrödinger-type equation with a complex-valued potential.
\end{remark}

Discretize \cref{eq:fractional-diffusion-pde} on an \(M^d\) periodic grid with mesh size \(h=1/M\) and \(N=M^d\). Let \(F_M\) be the one-dimensional discrete Fourier transform, \(F_N:=F_M^{\otimes d}\), and let \(-\Delta_{h,1}\) be the one-dimensional periodic discrete Laplacian, with eigenvalues \(4M^2\sin^2(\pi m/M)\), \(m\in\mathbb Z_M\). The \(d\)-dimensional periodic discrete Laplacian is given by the Kronecker sum
\begin{equation}\label{eq:periodic-laplacian-kronecker-sum}
-\Delta_{h,d} = \sum_{j=1}^d I^{\otimes(j-1)} \otimes (-\Delta_{h,1}) \otimes I^{\otimes(d-j)}, \quad -\Delta_{h,1} = \frac{1}{h^2} \begin{pmatrix} 2 & -1 &        &        &        & -1\\ -1 & 2 & -1     &        &        &   \\ & -1 & 2     & \ddots &        &   \\ &    & \ddots& \ddots & -1     &   \\ &    &       & -1     & 2      & -1\\ -1 &    &       &        & -1     & 2 \end{pmatrix}.
\end{equation}
Consequently, \(-\Delta_{h,d}\) is diagonalized by \(F_N=F_M^{\otimes d}\), with eigenvalues $\lambda_m = 4M^2 \sum_{\alpha=1}^d \sin^2\left(\frac{\pi m_\alpha}{M}\right)$ for \(m=(m_1,\dots,m_d)\in \mathbb Z_M^d\). Therefore, the discretized fractional Laplacian is given by
\begin{equation}\label{eq:discrete-fractional-laplacian}
L_s := (-\Delta_{h,d})^s = F_N^\dagger \operatorname{diag}\left[ \left( 4M^2 \sum_{\alpha=1}^d \sin^2\left(\frac{\pi m_\alpha}{M}\right) \right)^s \right]_{m\in\{0,\dots,M-1\}^d} F_N .
\end{equation}
Discretize the potential as \(H:=\operatorname{diag}(v(x_i/M))_{i=1}^N\), with \(x_i\in\mathbb Z_M^d\). If \(\vec u(t)\) approximates \(u(t,x)\) on the grid and \(\vec u_0\) is the initial condition, then the semi-discrete equation is
\begin{equation}\label{eq:fractional-diffusion-discrete}
\begin{cases} \dfrac{d\vec{u}}{dt} = -(L_s+iH)\vec{u}(t), \\ \vec{u}(0)=\vec{u}_0 . \end{cases}
\end{equation}

The estimate in \cref{mizuta-time-indep-short} does not apply to \cref{eq:fractional-diffusion-discrete}, since \(L_s\) is non-local for \(0<s<1\). On \(\mathbb R^d\), the fractional Laplacian has the singular integral representation
\begin{equation}\label{eq:frac-lap-singular-integral}
(-\Delta)^s u(x) = \frac{2^{2s}\Gamma (s+\frac{d}{2})} {\pi^{d/2}|\Gamma (-s)|} \operatorname{p.v.} \int_{\mathbb R^d} \frac{u(x)-u(y)}{|x-y|^{d+2s}}\,dy \qquad 0<s<1.
\end{equation}
On \(\mathbb T^d\), the kernel is obtained by periodization. Thus, \((-\Delta)^s u(x)\) depends on all spatial values \(u(y)\), with algebraically decaying interactions. This non-locality is inherited by \(L_s\), so the standard qubit encoding of \(L_s\) is non-local. On the other hand, \cref{aftab-time-indep-short} does apply. Using \(\|[A,B]\|\leq 2\|A\|\|B\|\), we obtain
\begin{equation}\label{eq:fractional-crude-Cjl-bound}
C_{j,\ell}(H,L_s) \leq \binom{j}{\ell} 2^{j-1} \|H\|^{j-\ell} \|L_s\|^\ell .
\end{equation}
With \(C_j(H,L_s):=\max_{0\leq \ell\leq j} C_{j,\ell}(H,L_s)\), we have \(C_j(H,L_s) \leq 2^{j-1}\bigl(\|H\|+\|L_s\|\bigr)^j\). It follows that
\begin{equation}
\chi_J(H,L_s) := \inf_{j\geq J} C_j(H,L_s)^{-1/j} \geq \frac{1}{2(\|H\|+\|L_s\|)} > 0 .
\end{equation}

Since \(L_s\succeq0\) and \(H=H^\dagger\), our algorithm applies. We compute the commutator scaling used in the complexity analysis. Since \(L_s\) is diagonalized by \(F_N\), we have $(F_NL_sF_N^\dagger)_{p,q}=\lambda_p^s\delta_{p,q}$, where \(\lambda_p\) is an eigenvalue of \(-\Delta_{h,d}\). The Fourier-basis matrix of $H$ is
\begin{equation}\label{eq:fractional-FHF}
(F_NHF_N^\dagger)_{p,q}
=
\frac{1}{N}\sum_{x\in\mathbb Z_M^d}v(x/M)e^{-2\pi i(p-q)\cdot x/M}
=: \widehat V_{p-q}.
\end{equation}
\cref{lem:fractional-symbol-difference-simple}, which is proved in \cref{proof-of-fractional-symbol-difference-simple}, estimates differences of fractional-Laplacian eigenvalues. Define \(\beta_s:=\min\{2s,1\}\), \(\theta_s:=2s-\beta_s=\max\{0,2s-1\}\), and let \(|\cdot|_{\operatorname{per}}\) denote distance on \(\mathbb Z_M^d\).

\begin{lem}\label{lem:fractional-symbol-difference-simple}
Let $0 < s < 1$. There exists a constant \(C_{s,d}>0\), independent of \(N\), such that
\begin{equation}\label{eq:fractional-symbol-difference-simple}
|\lambda_p^s-\lambda_q^s| \leq C_{s,d} N^{\theta_s/d} |p-q|_{\operatorname{per}}^{\beta_s}
\end{equation}
for all \(p,q\in \mathbb Z_M^d\).
\end{lem}

Using \cref{lem:fractional-symbol-difference-simple}, \cref{lem:canonical-fractional-commutator-simple} estimates the scaling of \(\operatorname{ad}_{L_s}^j(H)\) for \(j\geq1\). This serves as a prototype for bounding general commutator words in \cref{fractional-general-word-simple}.

\begin{lem}\label{lem:canonical-fractional-commutator-simple}
For every \(\ell\geq1\), we have $(F_N\operatorname{ad}_{L_s}^\ell(H)F_N^\dagger)_{p,q} = (\lambda_p^s-\lambda_q^s)^\ell \widehat V_{p-q}$.  Consequently, we have
\begin{equation}\label{eq:canonical-adLsH-bound}
\|\operatorname{ad}_{L_s}^\ell(H)\| \leq C_{\ell,V,s,d}^{(N)} N^{\ell\theta_s/d},
\end{equation}
where $C^{(N)}_{\ell,V,s,d}  :=    C_{s,d}^\ell   \sum_{r \in \mathbb Z^d_M} |r|_{\operatorname{per}}^{\ell\beta_s}  |\widehat V_r|$. In particular, if $\sup_{M\geq 1} \sum_{r\in\mathbb Z_M^d} |r|_{\operatorname{per}}^{\ell\beta_s}|\widehat V_r|<\infty$ then \(C_{\ell,V,s,d}^{(N)}\) is uniformly bounded in \(N\). In this case we may write $\|\operatorname{ad}_{L_s}^\ell(H)\| \leq C_{\ell,V,s,d} N^{\ell\theta_s/d}$, 
where \(C_{\ell,V,s,d}\) independent of $N$.
\end{lem}

\begin{proof}
Since \(F_NL_sF_N^\dagger\) is diagonal, we have
\begin{equation}\label{eq:comm-with-Ls-fourier}
(F_N[L_s,A]F_N^\dagger)_{p,q} = ([F_NL_sF_N^\dagger,F_NAF_N^\dagger])_{p,q} =(\lambda_p^s-\lambda_q^s) ( F_NAF_N^\dagger )_{p,q}.
\end{equation}
for any matrix \(A\). With \(A=H\), \cref{eq:fractional-FHF} gives $(F_N[L_s,H]F_N^\dagger)_{p,q} = (\lambda_p^s-\lambda_q^s)\widehat V_{p-q}$. 
Iterating \cref{eq:comm-with-Ls-fourier} gives the matrix elements of
$F_N\operatorname{ad}_{L_s}^\ell(H)F_N^\dagger$. By \cref{lem:fractional-symbol-difference-simple}, we have
\begin{equation}\label{entrywise-est}
|( F_N\operatorname{ad}_{L_s}^\ell(H)F_N^\dagger)_{p,q}| \leq C_{s,d}^\ell N^{\ell\theta_s/d} |p-q|_{\operatorname{per}}^{\ell\beta_s} |\widehat V_{p-q}|.
\end{equation}
Using \cref{entrywise-est}, we have
\begin{align}
\sum_{q \in \mathbb Z^d_M} |(F_N\operatorname{ad}_{L_s}^\ell(H)F_N^\dagger)_{p,q}| \leq C_{s,d}^\ell N^{\ell\theta_s/d} \sum_{q \in \mathbb Z^d_M} |p-q|_{\operatorname{per}}^{\ell\beta_s} |\widehat V_{p-q}|= C_{s,d}^\ell N^{\ell\theta_s/d} \sum_{r \in \mathbb Z^d_M} |r|_{\operatorname{per}}^{\ell\beta_s} |\widehat V_r|.
\end{align}
Hence, we have
\begin{equation}\label{eq:Bk-infty-bound}
\|F_N\operatorname{ad}_{L_s}^\ell(H)F_N^\dagger\|_\infty \leq C_{s,d}^\ell N^{\ell\theta_s/d} \sum_{r \in \mathbb Z^d_M} |r|_{\operatorname{per}}^{\ell\beta_s} |\widehat V_r|.
\end{equation}
Similarly, we have $\|F_N\operatorname{ad}_{L_s}^\ell(H)F_N^\dagger\|_1 \leq C_{s,d}^\ell N^{\ell\theta_s/d} \sum_{r \in \mathbb Z^d_M} |r|_{\operatorname{per}}^{\ell\beta_s} |\widehat V_r|$.  Using the estimate \(\|A\|\leq \sqrt{\|A\|_1\|A\|_\infty}\) which valid for any matrix \(A\), we obtain
\begin{equation}
\|\operatorname{ad}_{L_s}^\ell(H)\| = \|F_N\operatorname{ad}_{L_s}^\ell(H)F_N^\dagger\| \leq C_{s,d}^\ell N^{\ell\theta_s/d} \sum_{r \in \mathbb Z^d_M}|r|_{\operatorname{per}}^{\ell\beta_s} |\widehat V_r|.
\end{equation}
This proves \cref{eq:canonical-adLsH-bound}.
\end{proof}

We now extend \cref{lem:canonical-fractional-commutator-simple} to arbitrary nested commutator words in \cref{fractional-general-word-simple}, which is proved in \cref{proof-of-fractional-general-word-simple}. In \cref{fractional-general-word-simple}, we use the weaker consequence of \cref{lem:fractional-symbol-difference-simple}
\begin{equation}\label{eq:fractional-symbol-difference-weak}
|\lambda_p^s-\lambda_q^s|
\leq
C_{s,d}N^{\theta_s/d}\bigl(1+|p-q|_{\operatorname{per}}\bigr)^{\beta_s}.
\end{equation}
The constant in \cref{fractional-general-word-simple} is not sharp since the proof uses the crude bound \(\|\operatorname{ad}_H(A)\|_{a,\mathcal F}\leq2\mathcal V_a^{(N)}\|A\|_{a,\mathcal F}\) in \cref{eq:H-weighted-bound-current}. For special words, direct Fourier-basis formulas give sharper constants. For example, \cref{lem:canonical-fractional-commutator-simple} uses \((F_N\operatorname{ad}_{L_s}^{\ell}(H)F_N^\dagger)_{p,q}=(\lambda_p^s-\lambda_q^s)^\ell\widehat V_{p-q}\) together with fractional-Laplacian eigenvalue differences. More general nested commutators finite differences of \(p\mapsto\lambda_p^s\), which may yield sharper constants. We do not pursue these refinements here.

\begin{prop}\label{fractional-general-word-simple}
Let \(W_{j,\ell}(H,L_s)\) be any nested commutator word of length \(j\geq2\) containing exactly \(\ell\) occurrences of \(L_s\) and \(j-\ell\) occurrences of \(H\), with \(1\leq \ell\leq j-1\). Define
$\mathcal V_a^{(N)} := \sum_{r\in\mathbb Z_M^d} \left(1+|r|_{\operatorname{per}}\right)^a |\widehat V_r|$. 
Then
\begin{equation}
\|W_{j,\ell}(H,L_s)\| \leq 2^{j-\ell} C_{s,d}^{\ell} (\mathcal V_{j\beta_s}^{(N)} )^{j-\ell} N^{\ell\theta_s/d}.
\end{equation}
In particular, if \(\mathcal V_{j\beta_s}^{(N)}\) is bounded independently of \(N\), then
\begin{equation}
\|W_{j,\ell}(H,L_s)\| \leq C_{j,V,s,d}N^{\ell\theta_s/d},
\end{equation}
where \(C_{j,V,s,d} > 0\) is independent of \(N\).
\end{prop}

We now analyze the implications of \cref{fractional-general-word-simple} for the complexity of our algorithm. Since \(\mathcal C_{j,\ell}(H,L_s)\) collects all commutator words with exactly \(\ell\) occurrences of \(L_s\), \cref{fractional-general-word-simple} implies
\begin{equation}\label{eq:fractional-Cjell-application-bound}
\mathcal C_{j,\ell}(H,L_s) \leq 2^j 2^{j-\ell} C_{s,d}^{\ell} (\mathcal V_{j\beta_s}^{(N)})^{j-\ell} N^{\ell\theta_s/d}
\end{equation}
for \(1\leq \ell\leq j-1\). 
Define
\begin{equation}\label{eq:fractional-Aj-def}
A_{V,j}^{(N)} := \max_{1\leq \ell\leq j-1} \left( 2^j 2^{j-\ell} C_{s,d}^{\ell}  (\mathcal V_{j\beta_s}^{(N)})^{j-\ell} \right)^{1/j}.
\end{equation}
\cref{eq:fractional-Cjell-application-bound} gives $\mathcal C_{j,\ell}(H,L_s) \leq (A_{V,j}^{(N)})^j N^{\ell\theta_s/d}$. Using the definition of \(\alpha_{Q,j}\) from \cref{lem:simplified-root-growth-scale}, we obtain
\begin{align}
\alpha_{Q,j}\le \sum_{\ell=0}^{j} R_Q^\ell \mathcal C_{j,\ell}(H,L_s)\le (A_{V,j}^{(N)})^j\sum_{\ell=1}^{j-1}(R_QN^{\theta_s/d})^\ell\le \bigl(2A_{V,j}^{(N)}(1+R_QN^{\theta_s/d})\bigr)^j,
\end{align}
where the last inequality follows from \(\sum_{\ell=1}^{j-1}x^\ell\leq (2(1+x))^j\) for \(x\geq0\). Hence, if \(A_V^{(N)} := \sup_{j\ge 2} A_{V,j}^{(N)} < \infty\), we have $\alpha_{Q,j} \leq ( 2A_V^{(N)} (1+R_QN^{\theta_s/d} ) )^j$ Therefore, the root-growth scale \(\omega\) from \cref{lem:simplified-root-growth-scale} satisfies
\begin{equation}\label{eq:fractional-omega-bound}
\omega = \mathcal O ( A_V^{(N)} (1+R_QN^{\theta_s/d}) ).
\end{equation}
Choose \(\rho_Q=c/\omega\) for some fixed \(0<c<1\). Then \(\rho_Q\omega=c<1\), so \cref{lem:simplified-root-growth-scale} applies, and \(\rho_Q^{-1}=\mathcal O(\omega)\). Combining \cref{eq:fractional-omega-bound} with \cref{eq:mu-mQ-root-growth-bound} gives
\begin{equation}\label{eq:fractional-muQ-scale}
\mu_Q = \mathcal {O} ( A_V^{(N)} (1+R_QN^{\theta_s/d})),
\end{equation}
where the \(\mathcal O\)-notation absorbs the fixed finite factor \(S\) from \cref{lem:simplified-root-growth-scale}. More explicitly, the optimized MPF overhead contributes the polylogarithmic factor
\begin{align}\label{eq:fractional-polylog-factor}
\operatorname{polylog}_Q(T,\epsilon_Q,N) &:= \mathcal{O} \left( \log\left(\frac{ A_V^{(N)} (1+R_QN^{\theta_s/d}) T}{\epsilon_Q}\right)^2 \left( \log\log\left(\frac{ A_V^{(N)} (1+R_QN^{\theta_s/d}) T}{\epsilon_Q}\right) \right)^2\right), \\ &= \mathcal{O} \left( \log\left(\frac{ A_V^{(N)} (1+R_QN^{\theta_s/d}) T}{\epsilon_Q}\right)^4\right),
\end{align}
where we have assumed that $A_V^{(N)} (1+R_QN^{\theta_s/d}) T/\epsilon_Q \geq e^e$.  Substituting \cref{eq:fractional-muQ-scale} into \cref{abstract-optimized-block-complexity-simplified} with \(F(m)\equiv1\), we obtain
\begin{equation}
\mathcal Q_{\operatorname{block}} = \mathcal O ( A_V^{(N)} (1+R_QN^{\theta_s/d} )T \operatorname{polylog}_Q(T,\epsilon_Q,N) ).
\end{equation}
Similarly, provided \(W_Q^{\operatorname{MPF}}u_0\neq0\), substituting \cref{eq:fractional-muQ-scale} into \cref{abstract-optimized-state-complexity-simplified}, we obtain
\begin{equation}\label{eq:fractional-state-complexity-final}
\mathcal Q_{\operatorname{state}}(u_0) = \mathcal O ( \chi_Q(u_0) A_V^{(N)} (1+R_QN^{\theta_s/d} )T \operatorname{polylog}_Q(T,\epsilon_Q,N) ).
\end{equation}

\subsection{Advection-Diffusion Equation}\label{generalized-heat-dirichlet}
Let \(a>0\) and \(b\in\mathbb R\). Consider the following partial differential equation: 
\begin{equation}\label{eq:generalized-heat-pde}
\begin{cases} \partial_t u(t,x) = a\Delta u(t,x) + b \nabla\cdot(\vec 1 u(t,x)) - c(t,x)u(t,x) + f(t,x), & (t,x)\in (0,T]\times (0,1)^d,\\ u(0,x)=u_0(x), & x\in [0,1]^d,\\ u(t,x)=0, & (t,x)\in[0,T]\times\partial [0,1]^d. \end{cases}
\end{equation}
\cref{eq:generalized-heat-pde} is a \(d\)-dimensional advection--diffusion equation with homogeneous Dirichlet boundary conditions, where \(a\Delta u\) is diffusion, \(b\nabla\cdot(\vec 1u)\) is drift, \(c(t,x)\geq0\) is absorption, and \(f(t,x)\) is the source. Such equations model heat and mass transfer, reactive systems, drug delivery, and related transport phenomena~\cite{ParhiziEtAlCDRS,AgudAlbesaEtAlCDR,VillotaCadenaEtAlCDRSolv,an2026fast}. We take \(c,f\equiv0\) for simplicity. Discretize using \(M+1\) equidistant grid points per direction, mesh size \(h=1/M\), and \(N=(M+1)^d\). The \(d\)-dimensional Dirichlet finite-difference Laplacian is the Kronecker sum
\begin{equation}
-\Delta_{h,d} := \sum_{j=1}^d I^{\otimes(j-1)} \otimes (-\Delta_{h,1}) \otimes I^{\otimes(d-j)}, \quad -\Delta_{h,1} := \frac{1}{h^2} \begin{pmatrix} 2 & -1 &        &        &   \\ -1 & 2 & -1     &        &   \\ & -1 & 2     & \ddots &   \\ &    & \ddots & \ddots & -1 \\ &    &        & -1     & 2 \end{pmatrix}.
\end{equation}
By abuse of notation, we use the same notation for the Dirichlet finite-difference Laplacian and the periodic finite-difference Laplacian from \cref{fractional-diff}. Similarly, the \(d\)-dimensional Dirichlet finite-difference drift matrix is the Kronecker sum
\begin{equation}
D_{h,d} := \sum_{j=1}^d I^{\otimes(j-1)} \otimes D_{h,1} \otimes I^{\otimes(d-j)}, \quad D_{h,1} = \frac{1}{2h} \begin{pmatrix} 0  & 1  &        &        &   \\ -1 & 0  & 1      &        &   \\ & -1 & 0      & \ddots &   \\ &    & \ddots & \ddots & 1 \\ &    &        & -1     & 0 \end{pmatrix}.
\end{equation}
Here \(\Delta_{h,d}\) discretizes \(\Delta=\sum_{i=1}^d\partial_{x_i}^2\), while \(D_{h,d}\) discretizes \(\nabla\cdot(\vec 1u)=\sum_{i=1}^d\partial_{x_i}u\). 
Note that  Since \(\Delta_{h,d}=\Delta_{h,d}^\dagger\preceq0\)  \(D_{h,d}^\dagger=-D_{h,d}\). 
 Define $L := -a\Delta_{h,d}$ and $H := ibD_{h,d}$.  If \(\vec u(t)\) approximates \(u(t,x)\) on the grid and \(\vec u_0\) is the initial condition, then the semi-discrete equation is in dissipative-plus-Hamiltonian form is
\begin{equation}\label{eq:generalized-heat-lchs-form}
\begin{cases} \dfrac{d\vec u}{dt} = -\left(L+iH\right) \vec u(t),\\ \vec u(0)=\vec u_0 . \end{cases}
\end{equation}

\begin{remark}
The locality-extensive estimate in \cref{mizuta-time-indep-short} does not apply directly. Although \(\Delta_{h,d}\) and \(D_{h,d}\) are sparse in the grid basis, this does not imply \(q\)-locality under a many-body tensor-product encoding. For example, in binary encoding, \(|m+1\rangle\langle m|\) may involve non-local carries. Nevertheless, \cref{aftab-time-indep-short} applies as in \cref{fractional-diff}.
\end{remark}

We now derive commutator bounds. For \(d=1\), define \(A_1:=-h^2\Delta_{h,1}\) and \(B_1:=2hD_{h,1}\).

\begin{lem}\label{lem:dirichlet-first-commutator-1d}
Let \(\{e_0,\dots,e_M\}\) denote the standard basis of \(\mathbb C^{M+1}\). Define the matrices 
\begin{equation}
S := \sum_{j=0}^{M-1} e_j e_{j+1}^T, \quad 
P_0:=e_0e_0^T  
\quad P_M:=e_Me_M^T.
\end{equation}
We have $A_1:=2I-S-S^T$ and $B_1:=S-S^T$ and $[S,S^T]=P_0-P_M$. Consequently, we have
\begin{equation}\label{eq:dirichlet-first-commutator-1d}
[A_1,B_1] = 2(P_0-P_M), \quad [-\Delta_{h,1},D_{h,1}] =\frac{1}{h^3}(P_0-P_M).
\end{equation}
\end{lem}

\begin{proof}
The formulas for \(A_1\) and \(A_2\) are immediate from their definitions. By definition, we have $S=\sum_{j=0}^{M-1}e_je_{j+1}^T$ and $S^T=\sum_{j=0}^{M-1}e_{j+1}e_j^T$. Therefore, we have
\begin{equation}
SS^T=\sum_{j=0}^{M-1}e_je_j^T=I-P_M, \quad S^TS=\sum_{j=0}^{M-1}e_{j+1}e_{j+1}^T=I-P_0.
\end{equation}
Hence, we have $ [S,S^T] = SS^T-S^T S = (I-P_M)-(I-P_0) = P_0-P_M$. Therefore, we have
\begin{equation}
[A_1,B_1]=[2I-S-S^T,S-S^T]=[-S-S^T,S-S^T] = 2[S,S^T] = 2(P_0-P_M).
\end{equation}
Hence, we have
\begin{equation}
[-\Delta_{h,1},D_{h,1}]  = \left[\frac{A_1}{h^2},\frac{B_1}{2h}\right] =  \frac{1}{2h^3}[A_1,B_1]  =  \frac{1}{h^3}(P_0-P_M).
\end{equation}
This completes the proof.
\end{proof}

\cref{lem:dirichlet-first-commutator-1d} shows that the non-commutativity of \(A_1\) and \(B_1\) is localized at the boundary. Consequently, repeated commutators remain boundary-supported. This is made precise in \cref{dirichlet-nested-commutator-bound-1d}, which is proved in \cref{proof-of-dirichlet-nested-commutator-bound-1d}.

\begin{lem}\label{dirichlet-nested-commutator-bound-1d}
Let \(S,P_0\), and \(P_M\) be defined as in \cref{lem:dirichlet-first-commutator-1d}. For \(r\geq0\), define
\begin{equation}\label{eq:boundary-layer-space}
\mathcal B_r:=\operatorname{span}\left\{(S^T)^pP_0S^q,\;S^pP_M(S^T)^q:\;0\leq p,q\leq r\right\}.
\end{equation}
For \(X\in\mathcal B_r\), define the coefficient seminorm
\begin{equation}\label{eq:boundary-layer-coefficient-seminorm}
\|X\|_{\mathcal B}:=\inf\left\{\sum_\eta |c_\eta|:\;X=\sum_\eta c_\eta E_\eta,\;E_\eta\in\left\{(S^T)^pP_0S^q,\;S^pP_M(S^T)^q:\;0\leq p,q\leq r \leq M \right\}\right\}.
\end{equation}
Then $\|X\|\leq \|X\|_{\mathcal B}$. Moreover, for every \(X\in\mathcal B_r\), we have $[A_1,X]\in\mathcal B_{r+1}$, $[B_1,X]\in\mathcal B_{r+1}$, $\|[A_1,X]\|_{\mathcal B}\leq4\|X\|_{\mathcal B}$ and $\|[B_1,X]\|_{\mathcal B}\leq4\|X\|_{\mathcal B}$. Consequently, every non-zero right-nested mixed commutator \(W_j(A_1,B_1)\) of length \(j\geq2\) satisfies $\|W_j(A_1,B_1)\|\leq4^{j-1}$. 
\end{lem}

We lift the $d=1$ estimate to the \(d\)-dimensional tensor-product discretization in \cref{prop:dirichlet-nested-commutator-bound-d}.

\begin{prop}\label{prop:dirichlet-nested-commutator-bound-d}
Let \(L:=-a\Delta_{h,d}\) and \(H:=ibD_{h,d}\). For \(1\leq i\leq d\), define
\begin{equation}\label{eq:dirichlet-coordinate-operators}
A_i := I^{\otimes(i-1)} \otimes A_1 \otimes I^{\otimes(d-i)}, \qquad B_i := I^{\otimes(i-1)} \otimes B_1 \otimes I^{\otimes(d-i)}.
\end{equation}
Then $-\Delta_{h,d} = h^{-2}\sum_{i=1}^d A_i$ and $D_{h,d}  = \frac{1}{2h}\sum_{i=1}^d B_i$. Let \(W_{j,\ell}(H,L)\) be a non-zero right-nested mixed commutator of length \(j\geq2\), containing exactly \(\ell\) occurrences of \(L\) and \(j-\ell\) occurrences of \(H\), where \(1\leq \ell\leq j-1\). Then
\begin{equation}\label{eq:dirichlet-word-bound}
\|W_{j,\ell}(H,L)\| \leq d 2^{j+\ell-2} a^\ell |b|^{j-\ell} h^{-(j+\ell)}.
\end{equation}
Consequently, if \(\mathcal C_{j,\ell}(H,L)\) denotes the sum over all right-nested mixed commutator words of length \(j\) with exactly \(\ell\) occurrences of \(L\), then
\begin{equation}\label{eq:dirichlet-Cjell-bound}
\mathcal C_{j,\ell}(H,L) \leq \binom{j}{\ell} d 2^{j+\ell-2} a^\ell |b|^{j-\ell} h^{-(j+\ell)} \leq d 2^{2j+\ell-2} a^\ell |b|^{j-\ell}  h^{-(j+\ell)}.
\end{equation}
\end{prop}

\begin{proof}
The formulas for \(-\Delta_{h,d}\) and \(D_{h,d}\) follow directly from their Kronecker-sum definitions. We expand \(W_{j,\ell}(H,L)\) multilinearly in \(A_i\) and \(B_i\). Clearly, $[A_i,A_k]=[B_i,B_k]=[A_i,B_k]=0$ for $i\neq k$. For each fixed \(i\), \cref{dirichlet-nested-commutator-bound-1d} gives
\begin{equation}\label{eq:coordinate-wise-dimensionless-bound}
\|W_j(A_i,B_i)\| \leq 4^{j-1}.
\end{equation}
Each occurrence of \(L\) contributes a factor \(ah^{-2}\), while each occurrence of \(H\) contributes a factor \(|b|/(2h)\). Summing over the possible coordinate directions contributes only a factor of \(d\). Therefore, we have
\begin{equation}\label{eq:dirichlet-word-bound-proof-first}
\|W_{j,\ell}(H,L)\| \leq d  4^{j-1} a^\ell h^{-2\ell} \left(\frac{|b|}{2h}\right)^{j-\ell} = d  2^{j+\ell-2} a^\ell |b|^{j-\ell} h^{-(j+\ell)}. 
\end{equation}
There are at most \(\binom{j}{\ell}\) right-nested mixed words of length \(j\) with exactly \(\ell\) occurrences of \(L\). Summing \cref{eq:dirichlet-word-bound} over these words gives \cref{eq:dirichlet-Cjell-bound}.
\end{proof}

Since \(N=(M+1)^d\) and \(h^{-1}=M\leq N^{1/d}\) \cref{prop:dirichlet-nested-commutator-bound-d} implies
\begin{equation}\label{eq:dirichlet-Cjell-N-scaling}
\mathcal C_{j,\ell}(H,L) \leq d 2^{2j+\ell-2} a^\ell |b|^{j-\ell} N^{(j+\ell)/d}.
\end{equation}
We now translate \cref{prop:dirichlet-nested-commutator-bound-d} into a complexity bound.  Let \(\alpha_{Q,j} \leq \sum_{\ell=1}^{j-1} R_Q^\ell \mathcal C_{j,\ell}(H,L)\) be as in \cref{lem:simplified-root-growth-scale}. Using \cref{eq:dirichlet-Cjell-bound}, for every \(j \geq 2\) we have

\begin{align}\label{eq:advection-diffusion-alphaQj-bound}
\alpha_{Q,j} &\leq d2^{2j-2} \sum_{\ell=1}^{j-1} (2aR_Qh^{-2} )^\ell (|b|h^{-1} )^{j-\ell} \leq d2^{2j-2} ( |b|h^{-1} + 2aR_Qh^{-2} )^j .
\end{align}
Since \(h^{-1}=M\leq N^{1/d}\), the root-growth scale \(\omega\) from \cref{lem:simplified-root-growth-scale} satisfies
\begin{equation}\label{eq:advection-diffusion-omega-N-scale}
\omega = \mathcal O ( |b|N^{1/d} + aR_QN^{2/d}).
\end{equation}
Choose \(\rho_Q=c/\omega\) for fixed \(0<c<1\). Then \(\rho_Q\omega=c<1\), so \cref{lem:simplified-root-growth-scale} applies and \(\rho_Q^{-1}=\mathcal O(\omega)\). Combining \cref{eq:advection-diffusion-omega-N-scale} with \cref{eq:mu-mQ-root-growth-bound} gives
\begin{equation}\label{eq:advection-diffusion-muQ-scale}
\mu_Q= \mathcal O ( |b|N^{1/d} + aR_QN^{2/d}),
\end{equation}
where the \(\mathcal O\)-notation absorbs the fixed finite factor \(S\) from \cref{lem:simplified-root-growth-scale}. For the optimized-order estimate, the MPF overhead contributes the polylogarithmic factor
\begin{align}\label{eq:advection-diffusion-polylog-factor}
\operatorname{polylog}_Q(T,\epsilon_Q,N) &:= \mathcal O\left( \log\left( \frac{ \left(|b|N^{1/d}+aR_QN^{2/d}\right)T }{\epsilon_Q} \right)^2 \left( \log\log\left( \frac{ \left(|b|N^{1/d}+aR_QN^{2/d}\right)T }{\epsilon_Q} \right) \right)^2 \right) \nonumber\\ &= \mathcal O\left( \log\left( \frac{ \left(|b|N^{1/d}+aR_QN^{2/d}\right)T }{\epsilon_Q} \right)^4 \right),
\end{align}
where in the second line we assume $  (|b|N^{1/d}+aR_QN^{2/d})T/\epsilon_Q \geq e^e$. Substituting \cref{eq:advection-diffusion-muQ-scale} into \cref{abstract-optimized-block-complexity-simplified}, we obtain
\begin{equation}\label{eq:advection-diffusion-block-complexity-final}
\mathcal Q_{\operatorname{block}} = \mathcal O ( ( |b|N^{1/d} + aR_QN^{2/d} ) T \operatorname{polylog}_Q(T,\epsilon_Q,N) ).
\end{equation}
Similarly, provided \(W_Q^{\operatorname{MPF}}\vec u_0\neq0\), substituting \cref{eq:advection-diffusion-muQ-scale} into \cref{abstract-optimized-state-complexity-simplified} gives
\begin{equation}\label{eq:advection-diffusion-state-complexity-final}
\mathcal Q_{\operatorname{state}}(\vec u_0) = \mathcal O ( \chi_Q(\vec u_0) ( |b|N^{1/d} + aR_QN^{2/d} ) T \operatorname{polylog}_Q(T,\epsilon_Q,N) ).
\end{equation}

\subsection{No-Jump Dynamics for a Dissipative Ising Model}
\label{sec:no-jump-tfim}
We next consider the no-jump regime of a dissipative transverse-field Ising model with local spontaneous emission, formulated as a Lindblad dynamics. The no-jump regime is obtained from the standard quantum-trajectory formalism~\cite{DalibardCastinMolmer1992,PlenioKnight1998,LeeChan2013,RobertsClerk2023}. Let \(\mathcal G=(V,E(V))\) be a finite graph. The closed-system transverse-field Ising Hamiltonian is 
\begin{equation}\label{eq:tfim-hamiltonian}
H = \sum_{j\in V} h_j X_j + \sum_{\langle j,\ell\rangle\in E(V)} J_{j\ell} Z_j Z_\ell,
\end{equation}
where \(h_j \in \mathbb{R}\) are transverse-field strengths and \(J_{j\ell} \in \mathbb{R}\) are Ising couplings. Local spontaneous emission is modeled by Lindblad jump operators \(L_j=\sqrt{2\gamma_j}\,\sigma_j^-\), where \(\gamma_j\geq0\) and \(\sigma_j^-=\ket{0}\!\bra{1}\). Thus \(L_j\) represents local decay at site \(j\) with rate \(\gamma_j\). The associated Lindblad equation for \(\rho\) is
\begin{equation}\label{eq:tfim-lindblad}
\frac{d\rho}{dt} = -i[H,\rho] + \sum_{j\in V} ( L_j\rho L_j^\dagger - \frac12\{L_j^\dagger L_j,\rho\} ).
\end{equation}

\cref{eq:tfim-lindblad} gives the unconditional evolution averaged over all emission histories. In the quantum-trajectory picture, \(L_j\rho L_j^\dagger\) represents a detected jump at site \(j\). Conditioning on records with no detected emissions gives the no-jump evolution, generated by the effective Hamiltonian
\begin{equation}\label{eq:tfim-effective-Hamiltonian}
H_{\operatorname{eff}} = H-\frac{i}{2}\sum_{j\in V}L_j^\dagger L_j = \sum_{j\in V} h_j X_j + \sum_{\langle j,\ell\rangle\in E(V)} J_{j\ell} Z_j Z_\ell - i\sum_{j\in V}\gamma_j n_j,
\end{equation}
where $n_j=\sigma_j^+\sigma_j^-= |1\rangle \langle 1|=(I-Z_j)/2$.
The final term reflects the no-jump conditioning, where configurations more likely to emit are suppressed when no photon is detected. The corresponding dynamics are given by
\begin{equation}\label{eq:tfim-no-jump-schrodinger-LCHS-derivation}
i\frac{d}{dt}\ket{\psi(t)} = H_{\mathrm{eff}}\ket{\psi(t)} .
\end{equation}
Let $L=\sum_{j\in V}\gamma_j n_j = \sum_{j\in V}\gamma_j \frac{I-Z_j}{2}$. Then \(L\succeq0\) and $H_{\operatorname{eff}} = H-iL$. Multiplying \cref{eq:tfim-no-jump-schrodinger-LCHS-derivation} by \(-i\), we obtain
\begin{equation}\label{eq:tfim-no-jump-LCHS-form}
\frac{d}{dt}\ket{\psi(t)} = -i(H-iL)\ket{\psi(t)}  = -(L+iH)\ket{\psi(t)}.
\end{equation}

In what follows, take \(h_j\equiv h\), \(J_{j\ell}\equiv J\), and \(\gamma_j\equiv\gamma\), and assume \(\mathcal G\) has maximum degree \(D\), independent of \(|V|\). We first verify locality and extensivity. The Hamiltonian \(H\) is \(2\)-local:. Indeed, \(hX_j\) is supported on on one site and \(JZ_jZ_\ell\) is supported on two sites. For each site \(j\), at most one onsite term and \(D\) interaction terms contain \(j\). Since \(\|X_j\|=\|Z_jZ_\ell\|=1\),
\begin{equation}
\sum_{\alpha:\,j\in\operatorname{supp}(H_\alpha)} \|H_\alpha\| \leq |h|+D|J|.
\label{eq:tfim-H-extensive-bound}
\end{equation}
Hence, \(H\) is \(g_H\)-extensive with \(g_H:=|h|+D|J|\). Similarly, \(L\) is \(1\)-local because each term \(\gamma n_j\) is supported only on site \(j\). Hence, $q_L=1$. Since \(n_j\) is a projection, \(\|n_j\|=1\), and therefore
\begin{equation}
\sum_{\alpha:\,j\in\operatorname{supp}(L_\alpha)} \|L_\alpha\| = \|\gamma n_j\| = \gamma .
\label{eq:tfim-L-extensive-bound}
\end{equation}
Thus \(L\) is \(g_L\)-extensive with $g_L:=\gamma$. Hence, \cref{spec-mizuta} implies that \(\{G_{k_i}\}_{i\in\mathcal I_Q}\) is \(q_Q\)-local and \(g_Q\)-extensive, with \(q_Q=\max\{q_H,q_L\}=2\) and \(g_Q\leq |h|+D|J|+R_Q\gamma\).

We now derive commutator bounds. Our approach uses the Pauli incompatibility graph, whose vertices are local Pauli summands and whose edges connect anti-commuting Pauli strings. Such commutativity-sensitive counting is in graph-based ordering methods in Trotter-error analyses~\cite{tranter2019ordering}. Since \(n_j=(I-Z_j)/2\), the identity part commutes with all terms and does not affect nested commutators. Thus, for commutator estimates, \(G_k=H+kL\) may be replaced by
\begin{equation}
\overline G_k = h\sum_{j\in V}X_j + J\sum_{\langle j,\ell\rangle\in E(V)}Z_jZ_\ell - \frac{k\gamma}{2}\sum_{j\in V}Z_j .
\label{eq:tfim-Gbar-k-def}
\end{equation}
This replacement leaves every nested commutator of length at least two unchanged. In what follows, we will use the following standard computation for commutators of Pauli strings \(P\) and \(Q\)
\begin{equation}
\|[xP,yQ]\| = \begin{cases} 0, & PQ=QP,\\ 2|xy|, & PQ=-QP, \end{cases}
\label{eq:pauli-commutator-norm}
\end{equation}
for $x,y \in \R$. 

\begin{prop}
\label{prop:tfim-graph-commutator-bound}
Assume that \(\mathcal G\) has maximum degree \(D\) and that \(\gamma\geq0\). For every \(p\geq2\) and every quadrature node \(|k_i|\leq R_Q\), we have
\begin{align}
\alpha_{\operatorname{comm},p}(G_{k_i}) &\leq  |V|\,4^{p-1}(p-1)! \,2|h| \left(D|J|+\frac{R_Q\gamma}{2}\right) \left(|h|+D|J|+\frac{R_Q\gamma}{2}\right)^{p-2} \\ &\leq |V|\,4^{p-1}(p-1)! \left(|h|+D|J|+\frac{R_Q\gamma}{2}\right)^p .
\label{eq:tfim-alpha-comm-graph-bound}
\end{align}
\end{prop}

\begin{proof}
We work with \(\overline G_{k_i}\) because \(\alpha_{\operatorname{comm},p}(G_{k_i})=\alpha_{\operatorname{comm},p}(\overline G_{k_i})\). Its local Pauli summands are \(hX_j\), \(JZ_jZ_\ell\), and \(-k_i\gamma Z_j/2\). We call \(hX_j\) transverse, or \(X\)-type, and the remaining summands diagonal, or \(Z\)-type. The proof uses graph counting. Each non-zero nested commutator grows a cluster in the Pauli incompatibility graph, and each new summand must overlap the current support. We first record the local counting estimates. If \(S\subseteq V\), the total transverse weight supported on \(S\) is $\sum_{j\in S}\|hX_j\|  = |S|\,|h|$. Similarly, the total transverse weight supported on \(S\) is
\begin{equation}
\sum_{\substack{\langle a,b\rangle\in E(V)\\ \{a,b\}\cap S\neq\emptyset}} \|JZ_aZ_b\| + \sum_{a\in S} \left\| \frac{k_i\gamma}{2}Z_a \right\| \leq |S|\left(D|J|+\frac{R_Q\gamma}{2}\right).
\label{eq:tfim-diagonal-weight-on-S}
\end{equation}
Each vertex belongs to at most \(D\) edges, \(|k_i|\leq R_Q\), and \(\gamma\geq0\). Fix a layer pattern \(\tau=(\tau_1,\ldots,\tau_p)\in\{X,Z\}^p\), with \(\tau_1\) innermost. Since the innermost commutator is nonzero only for opposite types, \(\tau_1\neq\tau_2\). If \(\tau\) has \(\nu\) transverse entries, the first summand has total weight at most \(|V||h|\) for \(X\)-type and \(|V|(D|J|+R_Q\gamma/2)\) for \(Z\)-type. After \(a\) insertions, the cluster support has size at most \(2a\), so the next summand has total weight at most \(2a|h|\) for \(X\)-type and \(2a(D|J|+R_Q\gamma/2)\) for \(Z\)-type. Each nonzero commutator contributes a factor \(2\). Thus, the total contribution is bounded by
\begin{equation}
|V|\,4^{p-1}(p-1)! |h|^\nu \left(D|J|+\frac{R_Q\gamma}{2}\right)^{p-\nu}.
\label{eq:tfim-fixed-pattern-bound}
\end{equation}
If the pattern has \(\nu\) transverse entries, then \(\tau_1\neq\tau_2\) forces the first two entries to contain one \(X\) and one \(Z\), giving \(2\) possible orders. The remaining \(p-2\) entries contain \(\nu-1\) transverse entries, giving \(2\binom{p-2}{\nu-1}\) admissible patterns. Summing over \(1\leq\nu\leq p-1\) gives
\begin{align}
\alpha_{\operatorname{comm},p}(G_{k_i}) &\leq |V|\,4^{p-1}(p-1)! \sum_{\nu=1}^{p-1} 2\binom{p-2}{\nu-1} |h|^\nu \left(D|J|+\frac{R_Q\gamma}{2}\right)^{p-\nu} \\ &= |V|\,4^{p-1}(p-1)! \,2|h| \left(D|J|+\frac{R_Q\gamma}{2}\right) \left(|h|+D|J|+\frac{R_Q\gamma}{2}\right)^{p-2}.
\end{align}
The final inequality follows from \(2|h|(D|J|+\frac{R_Q\gamma}{2}) \leq (|h|+D|J|+\frac{R_Q\gamma}{2})^2\).
\end{proof}

\cref{prop:tfim-graph-commutator-bound} gives estimates for each commutator length \(p\). To apply the estimate in \cref{spec-mizuta}, we need a finite-\(p_0\) growth scale controlling all \(2\leq p\leq p_0\). Define
\begin{equation}\label{eq:tfim-ThetaQp0-def}
\Theta_{Q,p_0} := \max_{2\leq p\leq p_0} \left( 4^{p-1}(p-1)! \,2|h| \left(D|J|+\frac{R_Q\gamma}{2}\right) \left( |h|+D|J|+\frac{R_Q\gamma}{2} \right)^{p-2} \right)^{1/p}.
\end{equation}
\cref{prop:tfim-graph-commutator-bound} gives \(\alpha_{\operatorname{comm},p}(G_{k_i})\leq |V|\Theta_{Q,p_0}^p\) for \(2\leq p\leq p_0\). \cref{prop:tfim-mizuta-profile} converts this finite-order commutator estimate into a bound on \(\overline\mu_{m,Q}(p_0)\) for the locality-based complexity analysis.

\begin{prop}
\label{prop:tfim-mizuta-profile}
Let \(p_0\geq3\). We have 
\begin{equation}\label{eq:tfim-mubar-BQ-bound}
\overline\mu_{m,Q}(p_0) \leq|V|^{1/3} (1+p_0)^{1/3}\Theta_{Q,p_0}, 
\end{equation}
where $\bar\mu_{m,Q}(p_0):=\max_{i\in\mathcal I_Q}\mu_{m,p_0}(G_{k_i})$, \(\mu_{m,p_0}\) is defined in \cref{mu-m-p0} and $\Theta_{Q,p_0}$ is defined in \cref{eq:tfim-ThetaQp0-def}. 
\end{prop}

\begin{proof}
The proof is analogous to that of \cref{loc-mu-growth-envelope}, so we omit the details.
\end{proof}

We now instantiate the locality-based complexity estimate. By \cref{prop:tfim-mizuta-profile} and the definitions of \(\Phi_{m,Q}^{\ast}(p_0)\) and \(\Lambda_{m,Q}(p_0)\) in \cref{loc-Lambda-Phi-def}, we have
\begin{equation}\label{eq:tfim-Phi-Lambda-bound}
\Phi_{m,Q}^{\ast}(p_0)\leq  (|V|^{1/3}(1+p_0)^{1/3}\Theta_{Q,p_0} )^{2m+1}, \quad \Lambda_{m,Q}(p_0)\leq \alpha_Q (|V|^{1/3}(1+p_0)^{1/3}\Theta_{Q,p_0} )^{2m+1}.
\end{equation}
Together with the locality and extensivity bounds, \(q_Q=2\) and \(g_Q\leq |h|+D|J|+R_Q\gamma\), so \cref{spec-mizuta} applies. Choose \(p_0\), \(r_m\), and \(\delta_m\) as in \cref{spec-mizuta}. By \cref{eq:tfim-ThetaQp0-def}, we have
\begin{equation}\label{eq:tfim-Theta-polynomial-bound}
\Theta_{Q,p_0}\leq Cp_0\left(|h|+D|J|+\frac{R_Q\gamma}{2}\right)
\end{equation}
for some \(C>0\). Hence, we have
\begin{equation}\label{eq:tfim-mubar-polynomial-bound}
\overline\mu_{m,Q}(p_0)\leq C|V|^{1/3}\left(|h|+D|J|+\frac{R_Q\gamma}{2}\right)(1+p_0)^{4/3}.
\end{equation}
Since \(\delta_{p_0}=3|V|e^{-p_0}\) and \(K_m=\mathcal O(m^2(\log m)^2)\), the admissibility condition in \cref{loc-p0-admissible} holds with \(p_0(m)=\mathcal O(\log(e+m))\). Thus, the commutator-profile contribution grows as \((\log(e+m))^{4/3}\). For the optimized-order condition, the relevant quantity is \(\mu_{m,Q}\), not only \(\overline\mu_{m,Q}(p_0)\).The rooted \(\Phi_{m,Q}^{\ast}\)- and \(\Lambda_{m,Q}\) terms in \cref{loc-mu-mQ-def} contribute powers of \(\overline\mu_{m,Q}(p_0)^{1+1/(2m)}\). Hence a safe growth function is \(F(m):=(\log(e+m))^2\). More explicitly, one may take
\begin{equation}\label{eq:tfim-muQ-choice}
\mu_Q:=C_Q\max\left\{1,\ q_Qg_Q,\ |V|^{1/3}\left(|h|+D|J|+\frac{R_Q\gamma}{2}\right),\ \alpha_Q^{1/2}\left[|V|^{1/3}\left(|h|+D|J|+\frac{R_Q\gamma}{2}\right)\right]^{3/2}\right\},
\end{equation}
where \(C_Q>0\) is independent of \(m\), \(T\), and \(\epsilon_Q\), and absorbs the implicit constants. With this choice, the optimized-order growth condition analogous to \cref{F-cond} holds as $\mu_{m,Q}\leq \mu_Q(\log(e+m))^2$.  Choose $m=\max\left\{m_0,\left\lceil\log\left(e+\frac{T}{\epsilon_Q}\right)\right\rceil\right\}$. Applying the optimized complexity estimate  with \(F(m)=(\log(e+m))^2\) gives the block encoding complexity to be
\begin{equation}\label{eq:tfim-optimized-block-complexity}
\mathcal Q_{\operatorname{block}}=\mathcal O\left(\left(1+\mu_Q\max\{1,T\}(\log(e+m))^2\right)\left(\log\left(e+\frac{T}{\epsilon_Q}\right)\right)^2\left(\log\log\left(e^e+\frac{T}{\epsilon_Q}\right)\right)^2\right).
\end{equation}
Similarly, the normalized output state complexity is
\begin{equation}\label{eq:tfim-optimized-state-complexity}
\mathcal O\left(\chi_Q(u(0))\left(1+\mu_Q\max\{1,T\}(\log(e+m))^2\right)\left(\log\left(e+\frac{T}{\epsilon_Q}\right)\right)^2\left(\log\log\left(e^e+\frac{T}{\epsilon_Q}\right)\right)^2\right).
\end{equation}

\printbibliography

@misc{Dalzell2024,
      title={A shortcut to an optimal quantum linear system solver}, 
      author={Alexander M. Dalzell},
      year={2024},
      eprint={2406.12086},
      archivePrefix={arXiv},
      primaryClass={quant-ph},
      url={https://arxiv.org/abs/2406.12086}, 
}

@article{LowSu2026,
   title={Quantum linear system algorithm with optimal queries to initial state preparation},
   volume={10},
   ISSN={2521-327X},
   url={http://dx.doi.org/10.22331/q-2026-03-23-2041},
   DOI={10.22331/q-2026-03-23-2041},
   journal={Quantum},
   publisher={Verein zur Forderung des Open Access Publizierens in den Quantenwissenschaften},
   author={Low, Guang Hao and Su, Yuan},
   year={2026},
   month=Mar, pages={2041} }

@article{Krovi2022,
   title={Improved quantum algorithms for linear and nonlinear differential equations},
   volume={7},
   DOI={10.22331/q-2023-02-02-913},
   journal={Quantum},
   publisher={Verein zur Forderung des Open Access Publizierens in den Quantenwissenschaften},
   author={Krovi, Hari},
   year={2023},
   month=2, pages={913} }

@article{BerryChildsOstranderEtAl2017,
	doi = {10.1007/s00220-017-3002-y},
	year = 2017,
	publisher = {Springer Science and Business Media {LLC}
},
	volume = {356},
	number = {3},
	pages = {1057--1081},
	author = {Dominic W. Berry and Andrew M. Childs and Aaron Ostrander and Guoming Wang},
	title = {Quantum algorithm for linear differential equations with exponentially improved dependence on precision},
	journal = {Communications in Mathematical Physics}
}

@article{Berry2014,
	doi = {10.1088/1751-8113/47/10/105301},
	year = 2014,
	publisher = {{IOP} Publishing},
	volume = {47},
	number = {10},
	pages = {105301},
	author = {Dominic W. Berry},
	title = {High-order quantum algorithm for solving linear differential equations},
	journal = {Journal of Physics A: Mathematical and Theoretical}
}

@article{Childs2021theorytrotter,
  author        = {Andrew M. Childs and Yuan Su and Minh C. Tran and Nathan Wiebe and Shuchen Zhu},
  title         = {Theory of Trotter Error with Commutator Scaling},
  journal = {Phys. Rev. X},
  volume = {11},
  issue = {1},
  pages = {011020},
  numpages = {49},
  year = {2021},
  publisher = {American Physical Society},
  doi = {10.1103/PhysRevX.11.011020},
}

@misc{low2019wellconditioned,
      title={Well-conditioned multiproduct Hamiltonian simulation}, 
      author={Guang Hao Low and Vadym Kliuchnikov and Nathan Wiebe},
      year={2019},
      eprint={1907.11679},
      archivePrefix={arXiv},
      primaryClass={quant-ph}
}

@misc{low2025optimallchs,
      title={Optimal quantum simulation of linear non-unitary dynamics}, 
      author={Low, Guang Hao  and Somma, Rolando D.},
      year={2025},
      eprint={2508.19238},
      archivePrefix={arXiv},
      primaryClass={quant-ph},
      url={https://arxiv.org/abs/2508.19238}, 
}

@article{harrow2009quantum,
  title={Quantum algorithm for linear systems of equations},
  author={Harrow, Aram W. and Hassidim, Avinatan and Lloyd, Seth},
  journal={Physical review letters},
  volume={103},
  number={15},
  pages={150502},
  year={2009},
  publisher={American Physical Society}
}

@article{berry2007efficient,
  title={Efficient quantum algorithms for simulating sparse Hamiltonians},
  author={Berry, Dominic W. and Ahokas, Graeme and Cleve, Richard and Sanders, Barry C.},
  journal={Communications in Mathematical Physics},
  volume={270},
  number={2},
  pages={359--371},
  year={2007},
  publisher={Springer}
}

@book{dollard1984product,
    author = {Dollard, John D. and Friedman, Charles N.},
    title = {Product integration with applications to differential equations},
    publisher = {Cambridge University Press},
    year = {1984}
}

@article{childs2019optimalproduct,
  title = {Nearly Optimal Lattice Simulation by Product Formulas},
  author = {Childs, Andrew M. and Su, Yuan},
  journal = {Phys. Rev. Lett.},
  volume = {123},
  issue = {5},
  pages = {050503},
  numpages = {6},
  year = {2019},
  month = {08},
  publisher = {American Physical Society},
  doi = {10.1103/PhysRevLett.123.050503},
  url = {https://link.aps.org/doi/10.1103/PhysRevLett.123.050503}
}

@article{childs2015lcutaylor,
  title = {Simulating Hamiltonian Dynamics with a Truncated Taylor Series},
  author = {Berry, Dominic W. and Childs, Andrew M. and Cleve, Richard and Kothari, Robin and Somma, Rolando D.},
  journal = {Physical Review Letters},
  volume = {114},
  issue = {9},
  pages = {090502},
  numpages = {5},
  year = {2015},
  month = {3},
  publisher = {American Physical Society},
  doi = {10.1103/PhysRevLett.114.090502},
  url = {https://link.aps.org/doi/10.1103/PhysRevLett.114.090502}
}

@article{
lloyd1996quantumsimulation,
author = {Seth Lloyd },
title = {Universal Quantum Simulators},
journal = {Science},
volume = {273},
number = {5278},
pages = {1073-1078},
year = {1996},
doi = {10.1126/science.273.5278.1073},
URL = {https://www.science.org/doi/abs/10.1126/science.273.5278.1073},
}

@article{Childs2012lcu,
    author  = {Andrew M. Childs and Nathan Wiebe},
    title   = {Hamiltonian Simulation Using Linear Combinations of Unitary Operations},
    journal = {Quantum Information and Computation},
	doi = {10.26421/qic12.11-12},
	year = 2012,
	publisher = {Rinton Press},
	volume = {12},
    pages   = {901--924},
}

@article{suzuki1985decomposition,
  title={Decomposition formulas of exponential operators and Lie exponentials with some applications to quantum mechanics and statistical physics},
  author={Suzuki, Masuo},
  journal={Journal of mathematical physics},
  volume={26},
  number={4},
  pages={601--612},
  year={1985},
  publisher={AIP Publishing}
}

@article{suzuki1991general,
  title={General theory of fractal path integrals with applications to many-body theories and statistical physics},
  author={Suzuki, Masuo},
  journal={Journal of mathematical physics},
  volume={32},
  number={2},
  pages={400--407},
  year={1991},
  publisher={American Institute of Physics}
}

@article{fang2023time,
  title={Time-marching based quantum solvers for time-dependent linear differential equations},
  author={Fang, Di and Lin, Lin and Tong, Yu},
  journal={Quantum},
  volume={7},
  pages={955},
  year={2023},
}

@article{berry2024quantum,
  title={Quantum algorithm for time-dependent differential equations using Dyson series},
  author={Berry, Dominic W. and Costa, Pedro C. S.},
  journal={Quantum},
  volume={8},
  pages={1369},
  year={2024},
}

@article{childs2020quantum,
  title={Quantum spectral methods for differential equations},
  author={Childs, Andrew M. and Liu, Jin-Peng},
  journal={Communications in Mathematical Physics},
  volume={375},
  number={2},
  pages={1427--1457},
  year={2020},
  publisher={Springer}
}

@misc{aftab2024mpf,
      title={Multi-product Hamiltonian simulation with explicit commutator scaling}, 
      author={Aftab, Junaid and An, Dong  and Trivisa, Konstantina},
      year={2024},
      eprint={2403.08922},
      archivePrefix={arXiv},
      primaryClass={quant-ph},
      note={Submitted to \textit{Communications in Mathematical Physics}},
}

@article{jin2023quantum,
  title={Quantum simulation of partial differential equations: applications and detailed analysis},
  author={Jin, Shi and Liu, Nana and Yu, Yue},
  journal={Physical Review A},
  volume={108},
  number={3},
  pages={032603},
  year={2023},
  publisher={American Physical Society}
}

@article{feynman1982,
    author = {Richard P. Feynman},
    title = {Simulating physics with computers},
    journal = {International Journal of Theoretical Physics},
    year = {1982},
    volume = {21},
    pages = {467-488}
}

@article{feynman1986quantum,
  title={Quantum mechanical computers},
  author={Feynman, Richard P.},
  journal={Foundations of Physics},
  volume={16},
  number={6},
  pages={507--531},
  year={1986}
}

@article{an2023LCHS,
  title={Linear combination of Hamiltonian simulation for nonunitary dynamics with optimal state preparation cost},
  author={An, Dong and Liu, Jin-Peng and Lin, Lin},
  journal={Physical Review Letters},
  volume={131},
  number={15},
  year={2023},
  publisher={American Physical Society}
}

@article{an2023LCHS-Optimal,
  title={Quantum Algorithm for Linear Non-unitary Dynamics with Near-Optimal Dependence on All Parameters},
  author={An, Dong and Childs, Andrew M. and Lin, Lin},
  journal={Communications in Mathematical Physics},
  volume={407},
  number={1},
  pages={19},
  year={2026},
  publisher={Springer}
}

@article{wecker2015solving,
  title={Solving strongly correlated electron models on a quantum computer},
  author={Wecker, Dave and Hastings, Matthew B. and Wiebe, Nathan and Clark, Bryan K. and Nayak, Chetan and Troyer, Matthias},
  journal={Physical Review A},
  volume={92},
  number={6},
  pages={062318},
  year={2015},
  publisher={APS}
}

@article{babbush2015chemical,
  title={Chemical basis of Trotter-Suzuki errors in quantum chemistry simulation},
  author={Babbush, Ryan and McClean, Jarrod and Wecker, Dave and Aspuru-Guzik, Al{\'a}n and Wiebe, Nathan},
  journal={Physical Review A},
  volume={91},
  number={2},
  pages={022311},
  year={2015},
  publisher={APS}
}

@article{Costa2023QLSA,
  title = {Optimal Scaling Quantum Linear-Systems Solver via Discrete Adiabatic Theorem},
  author = {Costa, Pedro C. S. and An, Dong and Sanders, Yuval R. and Su, Yuan and Babbush, Ryan and Berry, Dominic W.},
  journal = {PRX Quantum},
  volume = {3},
  issue = {4},
  numpages = {54},
  year = {2022},
  month = {10},
  publisher = {American Physical Society},
  doi = {10.1103/PRXQuantum.3.040303},
  url = {https://link.aps.org/doi/10.1103/PRXQuantum.3.040303}
}

@article{childs2017QLSA,
  title={Quantum algorithm for systems of linear equations with exponentially improved dependence on precision},
  author={Childs, Andrew M. and Kothari, Robin and Somma, Rolando D.},
  journal={SIAM Journal on Computing},
  volume={46},
  number={6},
  pages={1920--1950},
  year={2017},
  publisher={SIAM}
}

@article{childs2018toward,
  title={Toward the first quantum simulation with quantum speedup},
  author={Childs, Andrew M. and Maslov, Dmitri and Nam, Yunseong and Ross, Neil J. and Su, Yuan},
  journal={Proceedings of the National Academy of Sciences},
  volume={115},
  number={38},
  pages={9456--9461},
  year={2018},
  publisher={National Academy of Sciences}
}

@article{ikeda2023minimum,
  title={Minimum Trotterization formulas for a time-dependent Hamiltonian},
  author={Ikeda, Tatsuhiko N. and Abrar, Asir and Chuang, Isaac L. and Sugiura, Sho},
  journal={Quantum},
  volume={7},
  pages={1168},
  year={2023},
}

@misc{jin2025quantumalgorithmsstochasticdifferential,
      title={Quantum Algorithms for Stochastic Differential Equations: A Schr\"odingerisation Approach}, 
      author={Jin, Shi and Liu, Nana and Wei, Wei},
      year={2025},
      eprint={2412.14868},
      archivePrefix={arXiv},
      primaryClass={quant-ph},
      url={https://arxiv.org/abs/2412.14868}, 
}

@article{mizuta2025commutatorscalinghamiltoniansimulation,
  title={On the commutator scaling in Hamiltonian simulation with multi-product formulas},
  author={Mizuta, Kaoru},
  journal={Quantum},
  volume={10},
  pages={1974},
  year={2026},
}

@misc{mizuta2024MPFtimedep,
      title={Explicit error bounds with commutator scaling for time-dependent product and multi-product formulas}, 
      author={Mizuta, Kaoru  and Ikeda, Tatsuhiko N.  and Fujii, Keisuke},
      year={2024},
      eprint={2410.14243},
      archivePrefix={arXiv},
      primaryClass={quant-ph},
      url={https://arxiv.org/abs/2410.14243}, 
}

@article{jin2024quantum,
  title={Quantum simulation of partial differential equations via Schr{\"o}dingerization},
  author={Jin, Shi and Liu, Nana and Yu, Yue},
  journal={Physical Review Letters},
  volume={133},
  number={23},
  pages={230602},
  year={2024},
  publisher={APS}
}

@article{trefethen2014exponentially,
  title={The exponentially convergent trapezoidal rule},
  author={Trefethen, Lloyd N. and Weideman, J.A.C.},
  journal={SIAM review},
  volume={56},
  number={3},
  pages={385--458},
  year={2014},
  publisher={SIAM}
}

@article{an2021time,
  title={Time-dependent unbounded Hamiltonian simulation with vector norm scaling},
  author={An, Dong and Fang, Di and Lin, Lin},
  journal={Quantum},
  volume={5},
  pages={459},
  year={2021},
}

@article{berry2020time,
  title        = {Time-dependent Hamiltonian simulation with $L^1$-norm scaling},
  author       = {Berry, Dominic W. and Childs, Andrew M. and Su, Yuan and Wang, Xin and Wiebe, Nathan},
  journal      = {Quantum},
  volume       = {4},
  pages        = {254},
  year         = {2020},
}

@article{wiebe2010higher,
  title={Higher order decompositions of ordered operator exponentials},
  author={Wiebe, Nathan and Berry, Dominic and H{\o}yer, Peter and Sanders, Barry C.},
  journal={Journal of Physics A: Mathematical and Theoretical},
  volume={43},
  number={6},
  pages={065203},
  year={2010},
  publisher={IOP Publishing}
}

@misc{low2019hamiltoniansimulationinteractionpicture,
      title={Hamiltonian Simulation in the Interaction Picture}, 
      author={Low, Guang Hao and Wiebe, Nathan},
      year={2019},
      eprint={1805.00675},
      archivePrefix={arXiv},
      primaryClass={quant-ph},
      url={https://arxiv.org/abs/1805.00675}, 
}

@article{anFang2022time,
  title={Time-dependent Hamiltonian simulation of highly oscillatory dynamics and superconvergence for Schr{\"o}dinger equation},
  author={An, Dong and Fang, Di and Lin, Lin},
  journal={Quantum},
  volume={6},
  pages={690},
  year={2022},
}

@article{AnFangLin2021,
	doi = {10.22331/q-2021-05-26-459},
  
	year = 2021,
  
	publisher = {Verein zur Forderung des Open Access Publizierens in den Quantenwissenschaften},
  
	volume = {5},
  
	pages = {459},
  
	author = {Dong An and Di Fang and Lin Lin},
  
	title = {Time-dependent unbounded Hamiltonian simulation with vector norm scaling},
  
	journal = {Quantum}
}

@Article{BerryChilds2012,
  author  = {Berry, Dominic W. and Childs, Andrew M.},
  title   = {{Black-box Hamiltonian simulation and unitary implementation}},
  journal = {Quantum Information \& Computation},
  year    = {2012},
  volume  = {12},
  number  = {1-2},
  pages   = {29--62}, 
  doi = {10.26421/QIC12.1-2}
}

@Article{BerryChildsCleveEtAl2015,
  Title                    = {Simulating {Hamiltonian} dynamics with a truncated {Taylor} series},
  Author                   = {Berry, Dominic W. and Childs, Andrew M. and Cleve, Richard and Kothari, Robin and Somma, Rolando D.},
  Journal                  = {Phys. Rev. Lett.},
  Year                     = {2015},
  Pages                    = {090502},
  Volume                   = {114}, 
  doi = {10.1103/PhysRevLett.114.090502}
}

@inproceedings{BerryChildsKothari2015,
   title={Hamiltonian Simulation with Nearly Optimal Dependence on all Parameters},
   DOI={10.1109/focs.2015.54},
   booktitle={2015 IEEE 56th Annual Symposium on Foundations of Computer Science},
   publisher={IEEE},
   author={Berry, Dominic W. and Childs, Andrew M. and Kothari, Robin},
   year={2015},
   month={10}
}

@Article{LowChuang2017,
   title={Optimal Hamiltonian Simulation by Quantum Signal Processing},
   volume={118},
   DOI={10.1103/physrevlett.118.010501},
   number={1},
   journal={Physical Review Letters},
   publisher={American Physical Society (APS)},
   author={Low, Guang Hao and Chuang, Isaac L.},
   year={2017},
   month={1}, 
}

@InProceedings{BerryChildsCleveEtAl2014,
  author    = {Berry, Dominic W. and Childs, Andrew M. and Cleve, Richard and Kothari, Robin and Somma, Rolando D.},
  title     = {Exponential improvement in precision for simulating sparse Hamiltonians},
  booktitle = {Proceedings of the forty-sixth annual ACM symposium on Theory of computing},
  year      = {2014},
  pages     = {283--292}, 
  doi = {10.1145/2591796.2591854}, 
}

@article{jordan2012quantum,
  title={Quantum algorithms for quantum field theories},
  author={Jordan, Stephen P. and Lee, Keith S.M. and Preskill, John},
  journal={Science},
  volume={336},
  number={6085},
  pages={1130--1133},
  year={2012},
  publisher={American Association for the Advancement of Science}
}

@article{bauer2020quantum,
  title={Quantum algorithms for quantum chemistry and quantum materials science},
  author={Bauer, Bela and Bravyi, Sergey and Motta, Mario and Chan, Garnet Kin-Lic},
  journal={Chemical Reviews},
  volume={120},
  number={22},
  pages={12685--12717},
  year={2020},
  publisher={ACS Publications}
}

@article{cao2019quantum,
  title={Quantum chemistry in the age of quantum computing},
  author={Cao, Yudong and Romero, Jonathan and Olson, Jonathan P. and Degroote, Matthias and Johnson, Peter D. and Kieferov{\'a}, M{\'a}ria and Kivlichan, Ian D. and Menke, Tim and Peropadre, Borja and Sawaya, Nicolas P.D. et al.},
  journal={Chemical Reviews},
  volume={119},
  number={19},
  pages={10856--10915},
  year={2019},
  publisher={ACS Publications}
}

@article{babbush2018low,
  title={Low-depth quantum simulation of materials},
  author={Babbush, Ryan and Wiebe, Nathan and McClean, Jarrod and McClain, James and Neven, Hartmut and Chan, Garnet Kin-Lic},
  journal={Physical Review X},
  volume={8},
  number={1},
  pages={011044},
  year={2018},
  publisher={APS}
}

@misc{WangZhouWangZhengZhangLi2026Lindbladian,
  title        = {Lindbladian Simulation with Commutator Bounds},
  author       = {Wang, Xinzhao and Zhou, Shuo and Wang, Xiaoyang and Zheng, Yi-Cong and Zhang, Shengyu and Li, Tongyang},
  year         = {2026},
  eprint       = {2603.28602},
  archivePrefix= {arXiv},
  primaryClass = {quant-ph}
}

@misc{AnChildsLinYing2024LaplaceLCHS,
  title        = {Laplace Transform Based Quantum Eigenvalue Transformation via Linear Combination of Hamiltonian Simulation},
  author       = {An, Dong and Childs, Andrew M. and Lin, Lin and Ying, Lexing},
  year         = {2024},
  eprint       = {2411.04010},
  archivePrefix= {arXiv},
  primaryClass = {quant-ph}
}

@misc{LowSu2024QEP,
  title        = {Quantum Eigenvalue Processing},
  author       = {Low, Guang Hao and Su, Yuan},
  year         = {2024},
  eprint       = {2401.06240},
  archivePrefix= {arXiv},
  primaryClass = {quant-ph}
}

@misc{TakahiraOhashiSogabeUsuda2021Contour,
  title        = {Quantum Algorithms based on the Block-Encoding Framework for Matrix Functions by Contour Integrals},
  author       = {Takahira, Souichi and Ohashi, Asuka and Sogabe, Tomohiro and Usuda, Tsuyoshi Sasaki},
  year         = {2021},
  eprint       = {2106.08076},
  archivePrefix= {arXiv},
  primaryClass = {quant-ph}
}

@misc{JiangAn2026ContourQET,
  title        = {Contour-integral Based Quantum Eigenvalue Transformation: Analysis and Applications},
  author       = {Jiang, Shan and An, Dong},
  year         = {2026},
  eprint       = {2601.11959},
  archivePrefix= {arXiv},
  primaryClass = {quant-ph}
}

@misc{WangLiuXueZhuangDouChenGuo2025CBMD,
  title        = {Quantum Simulation of Non-unitary Dynamics via Contour-based Matrix Decomposition},
  author       = {Wang, Chao and Liu, Huan-Yu and Xue, Cheng and Zhuang, Xi-Ning and Dou, Menghan and Chen, Zhao-Yun and Guo, Guo-Ping},
  year         = {2025},
  eprint       = {2511.10267},
  archivePrefix= {arXiv},
  primaryClass = {quant-ph}
}

@misc{JinLiuMaYu2025Schrodingerization,
  title        = {On the Schrödingerization Method for Linear Non-unitary Dynamics with Optimal Dependence on Matrix Queries},
  author       = {Jin, Shi and Liu, Nana and Ma, Chuwen and Yu, Yue},
  year         = {2025},
  eprint       = {2505.00370},
  archivePrefix= {arXiv},
  primaryClass = {quant-ph}
}

@misc{JinMaZuazua2026Transmutation,
  title        = {Transmutation Based Quantum Simulation for Non-unitary Dynamics},
  author       = {Jin, Shi and Ma, Chuwen and Zuazua, Enrique},
  year         = {2026},
  eprint       = {2601.03616},
  archivePrefix= {arXiv},
  primaryClass = {quant-ph}
}

@misc{WangZhuangDouChenGuo2026PSF,
  title        = {A Unified Poisson Summation Framework for Generalized Quantum Matrix Transformations},
  author       = {Wang, Chao and Zhuang, Xi-Ning and Dou, Menghan and Chen, Zhao-Yun and Guo, Guo-Ping},
  year         = {2026},
  eprint       = {2604.02874},
  archivePrefix= {arXiv},
  primaryClass = {quant-ph}
}

@article{Torrey1956,
  author  = {Torrey, H. C.},
  title   = {Bloch Equations with Diffusion Terms},
  journal = {Physical Review},
  volume  = {104},
  number  = {3},
  pages   = {563--565},
  year    = {1956},
  doi     = {10.1103/PhysRev.104.563}
}

@article{Magin2008,
  author  = {Magin, Richard L. and Abdullah, Omar and Baleanu, Dumitru and Zhou, Xiaohong Joe},
  title   = {Anomalous diffusion expressed through fractional order differential operators in the {Bloch--Torrey} equation},
  journal = {Journal of Magnetic Resonance},
  volume  = {190},
  number  = {2},
  pages   = {255--270},
  year    = {2008},
  doi     = {10.1016/j.jmr.2007.11.005}
}

@article{Yu2013,
  author  = {Yu, Qiang and Liu, Fawang and Turner, Ian and Burrage, Kevin},
  title   = {Numerical investigation of three types of space and time fractional {Bloch--Torrey} equations in 2D},
  journal = {Central European Journal of Physics},
  volume  = {11},
  number  = {6},
  pages   = {646--665},
  year    = {2013},
  doi     = {10.2478/s11534-013-0220-6}
}

@article{BuenoOrovioBurrage2017,
  author  = {Bueno-Orovio, Alfonso and Burrage, Kevin},
  title   = {Exact solutions to the fractional time-space {Bloch--Torrey} equation for magnetic resonance imaging},
  journal = {Communications in Nonlinear Science and Numerical Simulation},
  volume  = {52},
  pages   = {91--109},
  year    = {2017},
  doi     = {10.1016/j.cnsns.2017.04.013}
}

@article{Moutal2020,
  author  = {Moutal, Nicolas and Moutal, Antoine and Grebenkov, Denis S.},
  title   = {Diffusion {NMR} in periodic media: efficient computation and spectral properties},
  journal = {Journal of Physics A: Mathematical and Theoretical},
  volume  = {53},
  number  = {32},
  pages   = {325201},
  year    = {2020},
  doi     = {10.1088/1751-8121/ab977e},
  eprint  = {2005.06975},
  archivePrefix = {arXiv}
}

@article{an2026fast,
  title={Fast-forwarding quantum algorithms for linear dissipative differential equations},
  author={An, Dong and Onwunta, Akwum and Yang, Gengzhi},
  journal={Quantum},
  volume={10},
  pages={1986},
  year={2026},
}

@article{ParhiziEtAlCDRS,
  author  = {Parhizi, Mohammad and Kilaz, Gozdem and Ostanek, Jason K. and Jain, Ankur},
  title   = {Analytical solution of the convection-diffusion-reaction-source ({CDRS}) equation using Green's function technique},
  journal = {International Communications in Heat and Mass Transfer},
  volume  = {131},
  pages   = {105869},
  year    = {2022},
  doi     = {10.1016/j.icheatmasstransfer.2021.105869}
}

@article{AgudAlbesaEtAlCDR,
  author  = {Agud Albesa, Luc{\'i}a and Boix Garc{\'i}a, Marta and Pla-Ferrando, Laura and Cardona, Salvador C.},
  title   = {A study about the solution of convection-diffusion-reaction equation with {Danckwerts} boundary conditions by analytical, method of lines and {Crank--Nicholson} techniques},
  journal = {Mathematical Methods in the Applied Sciences},
  year    = {2023},
  doi     = {10.1002/mma.8633}
}

@article{VillotaCadenaEtAlCDRSolv,
  author  = {Villota-Cadena, {\'A}ngel P. and Sandoval-Palis, Iv{\'a}n P. and Grijalva-Villegas, Gabriel F. and Herrera-Granda, Erick P.},
  title   = {{CDR-Solv}: Solving the convection-diffusion-reaction equation with algebraic sub-grid scale stabilization using {Python}},
  journal = {Applied Sciences},
  volume  = {15},
  number  = {18},
  pages   = {10256},
  year    = {2025},
  doi     = {10.3390/app151810256}
}

@article{DalibardCastinMolmer1992,
  author  = {Dalibard, Jean and Castin, Yvan and M{\o}lmer, Klaus},
  title   = {Wave-function approach to dissipative processes in quantum optics},
  journal = {Physical Review Letters},
  volume  = {68},
  number  = {5},
  pages   = {580--583},
  year    = {1992},
  doi     = {10.1103/PhysRevLett.68.580}
}

@article{PlenioKnight1998,
  author  = {Plenio, M. B. and Knight, P. L.},
  title   = {The quantum-jump approach to dissipative dynamics in quantum optics},
  journal = {Reviews of Modern Physics},
  volume  = {70},
  number  = {1},
  pages   = {101--144},
  year    = {1998},
  doi     = {10.1103/RevModPhys.70.101}
}

@article{LeeChan2013,
  author  = {Lee, Tony E. and Chan, Ching-Kit},
  title   = {Dissipative transverse-field {I}sing model: Steady-state correlations and spin squeezing},
  journal = {Physical Review A},
  volume  = {88},
  number  = {6},
  pages   = {063811},
  year    = {2013},
  doi     = {10.1103/PhysRevA.88.063811}
}

@article{RobertsClerk2023,
  author  = {Roberts, David and Clerk, Aashish A.},
  title   = {Exact Solution of the Infinite-Range Dissipative Transverse-Field {I}sing Model},
  journal = {Physical Review Letters},
  volume  = {131},
  number  = {19},
  pages   = {190403},
  year    = {2023},
  doi     = {10.1103/PhysRevLett.131.190403}
}

@article{tranter2019ordering,
  title={Ordering of trotterization: Impact on errors in quantum simulation of electronic structure},
  author={Tranter, Andrew and Love, Peter J. and Mintert, Florian and Wiebe, Nathan and Coveney, Peter V},
  journal={Entropy},
  volume={21},
  number={12},
  pages={1218},
  year={2019},
  publisher={MDPI}
}
\appendix 
\addtocontents{toc}{\protect\setcounter{tocdepth}{1}}
\section{Omitted Proofs in \cref{sec:alg-err-analysis}}\label{app-alg-err}

\subsection{Proof of \cref{approx-simplify-integ}}\label{proof-of-approx-simplify-integ}

\begin{proof}
Set \(z = k-i y_0\). Since \(1-iz = (1-y_0)-ik\) and \(b+iz = (b+y_0)+ik\), we have $|1-iz| = \sqrt{(y_0-1)^2+k^2}$ and $\sqrt{(b+y_0)^2+k^2}$. We also have \(|e^{d(1-iz)}| = e^{d(1-y_0)}\). Moreover, since \(z^2 = (k-i y_0)^2 = k^2-2iky_0-y_0^2\), we have $|e^{-(z^2+1)/(4c^2)}| = e^{-(k^2-y_0^2+1)/(4c^2)} = e^{(y_0^2-1)/(4c^2)}e^{-k^2/(4c^2)}$. 
Therefore, we have
\begin{equation}
|\hat f_{a,b}(k-i y_0;c,d)| = \frac{(b+1)^{a-1}}{\sqrt{2\pi}}e^{\,d(1-y_0)+\frac{y_0^2-1}{4c^2}}\frac{e^{-k^2/(4c^2)}}{\sqrt{(y_0-1)^2+k^2}\,((b+y_0)^2+k^2)^{(a-1)/2}}.
\end{equation}
Integrating over \(\mathbb{R}\) and multiplying by \(1/\sqrt{2 \pi}\) gives \cref{simplified-eq}.
\end{proof}

\subsection{Proof of \cref{bessel-lemma}}\label{proof-of-bessel-lemma}

\begin{proof}
Substitute \(k=s\sinh t\). Then \(dk=s\cosh t\,dt\) and \(\sqrt{k^2+s^2}=s\cosh t\). Hence,
\begin{equation}
\int_{\mathbb{R}}\frac{e^{-k^2/(4\sigma^2)}}{\sqrt{k^2+s^2}}\,dk = \int_{\mathbb{R}}e^{-s^2\sinh^2 t/(4\sigma^2)}\,dt.
\end{equation}
Using \(\sinh^2 t=(\cosh(2t)-1)/2\), we obtain
\begin{equation}
\int_{\mathbb{R}}e^{-s^2\sinh^2 t/(4\sigma^2)}\,dt = e^{s^2/(8\sigma^2)}\int_{\mathbb{R}}e^{-\frac{s^2}{8\sigma^2}\cosh(2t)}\,dt.
\end{equation}
Setting \(u=2t\) gives
\begin{equation}
\int_{\mathbb{R}}e^{-\frac{s^2}{8\sigma^2}\cosh(2t)}\,dt = \frac{1}{2}\int_{\mathbb{R}}e^{-\frac{s^2}{8\sigma^2}\cosh u}\,du = \int_0^\infty e^{-\frac{s^2}{8\sigma^2}\cosh u}\,du = K_0\left(\frac{s^2}{8\sigma^2}\right),
\end{equation}
where we have used the standard representation \(K_0(x)=\int_0^\infty e^{-x\cosh u}\,du\) for \(x>0\). This proves \cref{bessel-eq}.
\end{proof}

\subsection{Proof of \cref{minimize-surrogate-approx}}\label{proof-of-minimize-surrogate-approx}

\begin{proof}
By definition,
\begin{equation}
\log B(y_0) = (a-1)\log(b+1)-\log(y_0-1)-(a-1)\log(b+y_0)+d(1-y_0)+\frac{y_0^2-1}{4c^2}.
\end{equation}
Differentiating with respect to \(y_0\) gives
\begin{equation}
\frac{d}{dy_0}\log B(y_0) = -d+\frac{y_0}{2c^2}-\frac{1}{y_0-1}-\frac{a-1}{b+y_0}.
\end{equation}
Hence, any stationary point \(y_0^*\) of \(\log B\) satisfies \(\frac{d}{dy_0}\log B(y_0^*)=0\), equivalently \cref{surrogate-stat-eq}. Moreover,
\begin{equation}\label{inc-phi}
\frac{d^2}{dy_0^2}\log B(y_0) = \frac{1}{2c^2}+\frac{1}{(y_0-1)^2}+\frac{a-1}{(b+y_0)^2}>0.
\end{equation}
Thus, \(\log B\) is strictly convex for \(y_0>1\), so any stationary point is its unique global minimizer. Since \(B(y_0)>0\) and \(\frac{d}{dy_0}B(y_0)=B(y_0)\frac{d}{dy_0}\log B(y_0)\), \(B\) and \(\log B\) have the same stationary points. They also have the same minimizers because \(\exp\) is strictly increasing. Such a \(y_0\) exists because
\begin{equation}
\frac{d}{dy_0}\log B(y_0) \to -\infty \quad \text{as } y_0 \to 1, \quad \frac{d}{dy_0}\log B(y_0) \to +\infty \quad \text{as } y_0 \to \infty.
\end{equation}
Hence, by the intermediate value theorem, there is a unique \(y_0^\ast\in(1,\infty)\) such that \(\frac{d}{dy_0}\log B(y_0^\ast)=0\). Thus, \(y_0^\ast\) is the unique stationary point of both \(\log B(y_0)\) and \(B(y_0)\).
\end{proof}

\subsection{Proof of \cref{feasible-kernel-profile-explicit}}\label{proof-of-feasible-kernel-profile-explicit}

\begin{proof}
We prove the three statements separately.
\begin{enumerate}
    \item By \cref{feasible-kernel-approx}, it suffices to prove \(d\geq \Phi(s_\epsilon)\), where \(s_\epsilon\) is the unique solution of \(\Psi(s_\epsilon)=\epsilon_{\operatorname{approx}}\). For \(s>1\), define
    \begin{equation}
    Q(s) := \frac{c}{\sqrt{\pi}(s-1)}\exp\left(a-\frac{(s-1)^2}{4c^2}\right).
    \end{equation}
    Note that \(\Psi(s)\le Q(s)\) for \(s>1\). Moreover,
    \begin{align}
    Q(\widetilde x_0) &= \frac{c}{\sqrt{\pi}(\widetilde x_0-1)}\exp\left(a-\frac{(\widetilde x_0-1)^2}{4c^2}\right) = \frac{1}{\sqrt{2\pi\,\omega_\epsilon}}\exp\left(a-\frac{\omega_\epsilon}{2}\right) = \epsilon_{\operatorname{approx}},
    \end{align}
    since \(\omega_\epsilon e^{\omega_\epsilon}=e^{2a}/(2\pi\epsilon_{\operatorname{approx}}^2)\). Hence, \(\Psi(\widetilde x_0)\leq Q(\widetilde x_0)=\epsilon_{\operatorname{approx}}\). Since \(\Psi\) is strictly decreasing, \(s_\epsilon\leq \widetilde x_0\). Since \(\Phi\) is strictly increasing, \cref{eq:lambertW-d-threshold} implies \(d\geq \Phi(\widetilde x_0)\geq \Phi(s_\epsilon)\). The conclusion follows from \cref{feasible-kernel-approx}.

    \item We have \(e^{2a}/(2\pi\epsilon_{\operatorname{approx}}^2)\le e^{2a}/\epsilon_{\operatorname{approx}}^2\). Since \(e^{2a}/\epsilon_{\operatorname{approx}}^2\ge e^2>e\) because \(\epsilon_{\operatorname{approx}}\in(0,1)\), we have
    \begin{equation}
    W\left(\frac{e^{2a}}{2\pi\epsilon_{\operatorname{approx}}^2}\right) \leq W\left(\frac{e^{2a}}{\epsilon_{\operatorname{approx}}^2}\right) \le \log\left(\frac{e^{2a}}{\epsilon_{\operatorname{approx}}^2}\right) = 2a+2\log(1/\epsilon_{\operatorname{approx}}).
    \end{equation}
    It follows that
    \begin{equation}
    \widetilde{x}_{0} = 1+c\sqrt{2W\left(\frac{e^{2a}}{2\pi\epsilon_{\operatorname{approx}}^2}\right)} \le 1+\sqrt{2}c\sqrt{2a+2\log(1/\epsilon_{\operatorname{approx}})} = \overline x_0.
    \end{equation}
    Since \(\Phi\) is strictly increasing, \(d\ge \Phi(\overline x_0)\) implies \(d\ge \Phi(\widetilde x_0)\). The claim follows from part~(1).

    \item Substituting \(\overline x_0\) into \(\Phi\) gives
    \begin{equation}
    \Phi(\overline x_0) = \frac{1}{2c^2}+\frac{1}{c}\sqrt{a+\log(1/\epsilon_{\operatorname{approx}})}-\frac{1}{\overline x_0-1}-\frac{a-1}{b+\overline x_0}.
    \end{equation}
    This gives the stated asymptotic scaling for fixed \(b\) and \(c\).
\end{enumerate}
This completes the proof.
\end{proof}

\subsection{Proof of \cref{trunc-error}}\label{proof-of-trunc-error}

\begin{proof}
It is clear that \(|\hat f_{a,b}(k;c,d)| = \frac{(b+1)^{a-1}e^{d-\frac{1}{4c^2}}}{\sqrt{2\pi}}\frac{e^{-k^2/(4c^2)}}{\sqrt{1+k^2}(b^2+k^2)^{(a-1)/2}}\). Hence, we have
\begin{align}
E_{\mathrm{trunc}}(R) &\leq \frac{1}{\sqrt{2\pi}}\int_{\mathbb R\setminus[-R,R]}|\hat f_{a,b}(k;c,d)|\,dk \\
&= \frac{(b+1)^{a-1}e^{d-\frac{1}{4c^2}}}{2\pi}\int_{|k|>R}\frac{e^{-k^2/(4c^2)}}{\sqrt{1+k^2}(b^2+k^2)^{(a-1)/2}}\,dk \\
&= \frac{(b+1)^{a-1}e^{d-\frac{1}{4c^2}}}{\pi}\int_R^\infty\frac{e^{-k^2/(4c^2)}}{\sqrt{1+k^2}(b^2+k^2)^{(a-1)/2}}\,dk.
\end{align}
For \(k\geq 0\), the elementary inequalities \(\sqrt{1+k^2}\geq 1\) and \((b^2+k^2)^{(a-1)/2}\geq b^{a-1}\) imply that
\begin{equation}
E_{\mathrm{trunc}}(R) \leq \frac{(b+1)^{a-1}e^{d-\frac{1}{4c^2}}}{\pi b^{a-1}}\int_R^\infty e^{-k^2/(4c^2)}\,dk = \frac{(b+1)^{a-1}e^{d-\frac{1}{4c^2}}}{\pi b^{a-1}}c\sqrt{\pi}\operatorname{erfc}\left(\frac{R}{2c}\right).
\end{equation}
This gives \cref{trunc-comp}.
\end{proof}

\section{Omitted Proofs in \cref{complexity-analysis}}

\subsection{Proof of \cref{short-time-step}}\label{proof-of-short-time-step}

\begin{proof}
Applying \cref{abstract-time-indep-mpf-revised} and the triangle inequality, we have
\begin{align}
\left\|\frac{1}{\sqrt{2 \pi}}  \int_{-R}^{R} \hat f(k)\big(U^{(2m)}_k(\Delta)-U_k(\Delta)\big)dk\right\|
&\le \frac{1}{\sqrt{2 \pi}} \int_{-R}^{R} |\hat f(k)| \|U^{(2m)}_k(\Delta)-U_k(\Delta) \| dk   \\
&\le \frac{1}{\sqrt{2 \pi}} \int_{-R}^{R} |\hat f(k)| (C_m\Delta^{2m+1}\Phi_m(k,N)+R_m(\delta) ) dk.
\end{align}
The first term is $ C_m \Delta^{2m+1}\Lambda_m(R,N)$ by~\cref{eq:Lambda-cont} and the second term is $\alpha_{\hat f, R} R_m(\delta)$.
\end{proof}

\subsection{Proof of \cref{series-profile-convergence-lemma}}\label{proof-of-series-profile-convergence-lemma}

\begin{proof}
Let \(\overline R_Q:=\max\{1,R_Q\}\). For \(0\leq r\leq \overline R_Q\), define \(S_n(r):=\sum_{\ell=0}^{n}r^\ell\mathcal C_{n,\ell}(H,L)\) for \(n\geq 1\). By definition of \(\mathcal C_n(H,L)\),
\begin{equation}
S_n(r) \leq (n+1)\overline R_Q^n\mathcal C_n(H,L)
\end{equation}
for $0\leq r\leq \overline R_Q$. For \(n\geq J\), the definition of \(\chi_J(H,L)\) gives \(\mathcal C_n(H,L)\leq \chi_J(H,L)^{-n}\). Since \(n\mapsto (n+1)^{1/n}\) is decreasing for \(n\geq 1\), we also have \(n+1\leq (J+1)^{n/J}\) for \(n\geq J\). Combining these bounds gives
\begin{equation}
\label{series-Sn-kappa-bound-large}
S_n(r) \leq \left( \frac{(J+1)^{1/J}\overline R_Q}{\chi_J(H,L)} \right)^n = \kappa_Q^{-n}
\end{equation}
for all \(n\geq J\) and \(0\leq r\leq \overline R_Q\). The estimate in \cref{series-Sn-kappa-bound-large} applies directly only for \(n\geq J\). To absorb the finitely many lower orders, define
\begin{equation}
\label{series-BQ-def}
B_Q := \max\left\{ 1,\, \max_{1\leq n<J} \kappa_Q^n \sup_{0\leq r\leq \overline R_Q}S_n(r) \right\}.
\end{equation}
If \(J=1\), the second maximum is interpreted as absent. Then, $S_n(r)\leq B_Q\kappa_Q^{-n}$ for every \(n\geq 1\) and every \(0\leq r\leq \overline R_Q\). Now fix \(m\geq 1\). For \(0\leq r\leq \overline R_Q\), define the pointwise profile
\begin{align}
\label{series-pointwise-profile-def}
\Phi_m(r) := \sum_{\substack{j\in 2\mathbb Z_+\\ j\geq 2m}} \sum_{l=1}^{m} \frac{\rho_Q^{j+l-(2m+1)}}{l!} \sum_{\substack{j_1,\dots,j_l\in 2\mathbb Z_+\\ j_1+\cdots+j_l=j}} \prod_{\kappa=1}^{l} S_{j_\kappa+1}(r).
\end{align}
Setting \(n_\kappa:=j_\kappa+1\), we have $n_1+\cdots+n_l=j+l$. We obtain $\prod_{\kappa=1}^{l}S_{j_\kappa+1}(r) \leq B_Q^l\kappa_Q^{-(j+l)}$. The number of ordered decompositions of an even integer \(j\) into \(l\) positive even parts is $\binom{j/2-1}{l-1}$. Therefore, we have
\begin{align}
\label{series-pointwise-profile-bound}
\Phi_m(r) &\leq \rho_Q^{-(2m+1)} \sum_{l=1}^{m} \frac{B_Q^l}{l!} \sum_{\substack{j\in 2\mathbb Z_+\\ j\geq 2m}} \binom{j/2-1}{l-1} \left(\frac{\rho_Q}{\kappa_Q}\right)^{j+l}.
\end{align}
Since \(0<\rho_Q<\kappa_Q\), the final series converges. Indeed, the binomial factor in \cref{series-pointwise-profile-bound} grows polynomially in \(j\), while \((\rho_Q/\kappa_Q)^j\) decays geometrically. Hence, we have
$\sup_{0\leq r\leq \overline R_Q}\Phi_m(r)<\infty$. 
We next prove convergence of the coefficient series defining \(\Phi_{m,\nu}(H,L)\). Expanding the product in \cref{series-pointwise-profile-def} gives $\Phi_m(r)=\sum_{\nu\geq 0} \Phi_{m,\nu}(H,L)r^\nu$ for \(0\leq r\leq \overline R_Q\). All coefficients in this expansion are non-negative. Since \(\overline R_Q\geq 1\), evaluating at \(r=1\) and using $\sup_{0\leq r\leq \overline R_Q}\Phi_m(r)<\infty$ gives $\sum_{\nu\geq 0} \Phi_{m,\nu}(H,L) = \Phi_m(1) <\infty$. Therefore, each \(\Phi_{m,\nu}(H,L)\) is finite, and its defining series is absolutely convergent. In particular,
\begin{equation}
\label{series-Phi-star-finite}
\Phi_{m,Q}^* = \max_{i\in\mathcal I_Q}\Phi_m(|k_i|) <\infty ,
\end{equation}
because \(|k_i|\leq R_Q\leq \overline R_Q\). It remains to verify convergence of \(\Lambda_{m,Q}\). Since all terms in \cref{series-Dmnu-def} are nonnegative, Tonelli's theorem permits rearranging the sums, giving
\begin{align}
\label{series-Lambda-Tonelli}
\Lambda_{m,Q} &= \sum_{\nu\geq 0}\Phi_{m,\nu}(H,L)M_\nu^Q = \sum_{i\in\mathcal I_Q}|v_i|\, \Phi_m(|k_i|).
\end{align}
Consequently, we have $\Lambda_{m,Q} \leq \alpha_Q\Phi_{m,Q}^* <\infty$. This proves absolute convergence of the series.
\end{proof}

\subsection{Proof of \cref{sm-lemma}}\label{proof-of-sm-lemma}

\begin{proof}
Writing \(j=2m+n\) and dropping the even restriction on \(n\) only enlarges the sum, so
\begin{equation}\label{eq:simplified-profile-Sm-first-bound}
S_m(x) \leq \sum_{n=0}^{\infty} \sum_{l=1}^{m} \frac{1}{l!} \binom{2m+n-1}{l-1} x^{n+l-1}.
\end{equation}
Choose \(\tau\in(0,\log(1/x))\), so that \(xe^\tau<1\), and set $ K_\tau:=\max\{1,\tau^{-1}\}$. For \(0\leq l-1\leq m-1\), the elementary inequality \(y^r/r!\leq e^{\tau y}\tau^{-r}\) gives
\begin{equation}\label{eq:simplified-profile-binomial-bound}
\binom{2m+n-1}{l-1} \leq \frac{(2m+n)^{l-1}}{(l-1)!} \leq K_\tau^m e^{\tau(2m+n)} .
\end{equation}
Substituting \cref{eq:simplified-profile-binomial-bound} into \cref{eq:simplified-profile-Sm-first-bound}, we obtain
\begin{align}
S_m(x) \leq K_\tau^m e^{2\tau m} \left( \sum_{l=1}^{m}\frac{x^{l-1}}{l!} \right) \left( \sum_{n=0}^{\infty}(xe^\tau)^n \right) \leq \frac{e}{1-xe^\tau} (K_\tau e^{2\tau} )^m .
\label{eq:simplified-profile-Sm-exponential-bound}
\end{align}
Consequently, we have
\begin{equation}\label{eq:simplified-profile-Sm-root-bound}
\sup_{m\geq1}S_m(x)^{1/(2m)} \leq \left(\frac{e}{1-xe^\tau}\right)^{1/2} K_\tau^{1/2}e^\tau <\infty .
\end{equation}
This completes the proof.
\end{proof}

\subsection{Proof of \cref{loc-mu-growth-envelope}}\label{proof-of-loc-mu-growth-envelope}

\begin{proof}
By \cite[Equation 8]{mizuta2025commutatorscalinghamiltoniansimulation}, we have the crude bound
\begin{equation}
\alpha_{\operatorname{comm},p}(H) \leq (p-1)!(2qg)^{p-1}Ng
\label{eq:mizuta-commutator-growth-input}
\end{equation}
for $p\geq 2$. Let \(j_1,\ldots,j_l\) be an admissible tuple in the definition of \(\mu_{m,p_0}(H)\). Thus, we have $2\leq j_\nu\leq p_0-1$, $j_1+\cdots+j_l=j$ and $1\leq l\leq \lfloor j/2\rfloor$. Since \(j_\nu+1\leq p_0\), the bound in \cref{eq:mizuta-commutator-growth-input} gives
\begin{align}
\alpha_{\operatorname{comm},j_\nu+1}(H) \leq j_\nu!(2qg)^{j_\nu}Ng \leq p_0^{j_\nu}(2qg)^{j_\nu}Ng = Ng(2qg p_0)^{j_\nu}.
\end{align}
Therefore, we have
\begin{align}
\prod_{\nu=1}^{l} \alpha_{\operatorname{comm},j_\nu+1}(H) \leq (Ng)^l(2qg p_0)^{j_1+\cdots+j_l} = (Ng)^l(2qg p_0)^j.
\end{align}
The number of admissible \(l\)-tuples is bounded by \((1+p_0)^l\). Taking the \(1/(j+l)\)-th root gives
\begin{align}
\left( \sum_{\substack{ 2\leq j_1,\ldots,j_l\leq p_0-1\\ j_1+\cdots+j_l=j}} \prod_{\nu=1}^{l} \alpha_{\operatorname{comm},j_\nu+1}(H) \right)^{1/(j+l)}  \leq (2qg p_0)^{j/(j+l)} \left((1+p_0)Ng\right)^{l/(j+l)} . \label{eq:root-alpha-local-growth-bound}
\end{align}
Since \(l\leq \lfloor j/2\rfloor\), we have $\frac{l}{j+l}\leq \frac13$ and $\frac{j}{j+l}\leq 1$. Using \(p_0\leq1+p_0\), we obtain
\begin{align}
(2qg p_0)^{j/(j+l)} \left((1+p_0)Ng\right)^{l/(j+l)}  \leq \max\{1,2qg\}\max\{1,Ng\}^{1/3}(1+p_0)^{4/3}.
\end{align}
Taking the supremum over all admissible \(j,l\) in the definition of \(\mu_{m,p_0}(H)\) proves \cref{eq:local-growth-envelope-from-mizuta}. If \(\{G_{k_i}:i\in \mathcal I_Q\}\) is uniformly \(q_Q\)-local and \(g_Q\)-extensive, then the same argument applies uniformly to each \(G_{k_i}\). Taking the maximum over \(i\in \mathcal I_Q\) gives \cref{eq:mubar-local-growth-envelope-from-mizuta}.
\end{proof}

\subsection{Proof of \cref{general-kernel-strip-bound}}\label{proof-of-general-kernel-strip-bound}

\begin{proof}
Let \(U_z(T):=e^{-i(H+zL)T}\). By~\cite[Lemma~3]{low2025optimallchs}, the map \(z\mapsto U_z(T)\) is entire. The poles of \(\hat f_{a,b}(z;c,d)\) occur at \(z=-i\) and, when \(a>1\), at \(z=ib\). Hence, if \(\rho<\rho_\ast\), then \(F_T\) is analytic on an open set containing \(\overline S_\rho\). Let \(z=x+i\beta\), with \(|\beta|\le \rho\). Since
\begin{equation}
1-iz=(1+\beta)-ix, \qquad b+iz=(b-\beta)+ix,
\end{equation}
we have
\begin{equation}
|1-iz|=\sqrt{(1+\beta)^2+x^2}, \qquad |b+iz|=\sqrt{(b-\beta)^2+x^2}.
\end{equation}
Moreover, \(|e^{d(1-iz)}|=e^{d(1+\beta)}\) and \(|e^{-(z^2+1)/(4c^2)}|=e^{(\beta^2-1)/(4c^2)}e^{-x^2/(4c^2)}\). By~\cite[Lemma~4]{low2025optimallchs}, applied in the time-independent case and specialized to \(L\succeq0\), we have \(\|U_{x+i\beta}(T)\|\le e^{\max\{\beta,0\}T\|L\|}\). Therefore,
\begin{align}
\|F_T(x+i\beta)\| &\le \frac{(b+1)^{a-1}}{2\pi}\frac{\exp\left(d(1+\beta)+\frac{\beta^2-1}{4c^2}+\max\{\beta,0\}T\|L\|\right)e^{-x^2/(4c^2)}}{\sqrt{(1+\beta)^2+x^2}\left((b-\beta)^2+x^2\right)^{(a-1)/2}}.
\end{align}
Since \(|\beta|\le \rho<1\), we have \(\sqrt{(1+\beta)^2+x^2}\ge 1-\rho\). If \(a>1\), then \(\rho<b\), and therefore \(\left((b-\beta)^2+x^2\right)^{(a-1)/2}\ge (b-\rho)^{a-1}\). Thus,
\begin{equation}
\|F_T(x+i\beta)\| \le \frac{(b+1)^{a-1}}{2\pi(1-\rho)D_{a,b}(\rho)}\exp\left(d(1+\beta)+\frac{\beta^2-1}{4c^2}+\max\{\beta,0\}T\|L\|\right)e^{-x^2/(4c^2)}.
\end{equation}
For \(\beta\in[-\rho,\rho]\), the exponent satisfies
\begin{align}
d(1+\beta)+\frac{\beta^2-1}{4c^2}+\max\{\beta,0\}T\|L\| &\le d+\frac{\rho^2-1}{4c^2}+\rho\max\{-d,d+T\|L\|\}.
\end{align}
Indeed, the left-hand side is convex in \(\beta\), so its maximum on \([-\rho,\rho]\) occurs at \(\beta=\pm\rho\). Integrating over \(\mathbb R\) and using \(\int_{\mathbb R} e^{-x^2/(4c^2)}\,dx = 2c\sqrt{\pi}\) gives \cref{eq:strip-L1-bound} with \(M_{a,b}(\rho)\). The same Gaussian upper bound also shows that \(F_T(x+i\beta)\) decays uniformly to zero on \(\overline S_\rho\) as \(|x|\to\infty\).
\end{proof}

\subsection{Proof of \cref{sinh-budget}}\label{proof-of-sinh-budget}

\begin{proof}
The choice of \(h\) controls the mesh error. Indeed, \cref{sinh-strip-bound} gives \(E_{\operatorname{sinh},\eta}(h,\beta)\le \epsilon_{\operatorname{mesh}}\). It remains to control the tail contribution. Since \(Y_h\ge Y\), the choice of \(Y\) gives
\begin{equation}
e^{-\lambda_\eta^2\sinh^2Y_h}\le e^{-\lambda_\eta^2\sinh^2Y}\le \frac{\epsilon_{\operatorname{tail}}}{2c\sqrt{\pi}C_{\operatorname{ker}}}.
\end{equation}
Moreover, the choice of \(Y\) also gives \(\lambda_\eta\sinh Y_h\ge1\), so \cref{sinh-lattice-quadrature} implies \(E_{\operatorname{tail},h}^{\operatorname{sinh},\eta}(Y_h)\le \epsilon_{\operatorname{tail}}\). Combining the mesh and tail estimates, we obtain
\begin{equation}
E_{\operatorname{quad}}(Q_{Y,h}^{\operatorname{sinh},\eta})\le E_{\operatorname{sinh},\eta}(h,\beta)+E_{\operatorname{tail},h}^{\operatorname{sinh},\eta}(Y_h)\le \epsilon_{\operatorname{mesh}}+\epsilon_{\operatorname{tail}}.
\end{equation}
Since \(n_h=\lceil Y/h\rceil\), we have \(|\mathcal I_{Q_{Y,h}^{\operatorname{sinh},\eta}}|=2n_h+1\le 2Y/h+3\). Hence,
\begin{equation}
\frac1h=\max\left\{1,\frac{\ell_{\operatorname{mesh}}^{\operatorname{sinh},\eta}}{2\pi\beta}\right\}=\mathcal O\left(1+\frac{\ell_{\operatorname{mesh}}^{\operatorname{sinh},\eta}}{\beta}\right),
\end{equation}
while \(Y=\mathcal O\left(\log\left(1+\frac{\sqrt{\ell_{\operatorname{tail}}^{\operatorname{sinh},\eta}}}{\lambda_\eta}\right)\right)\). Substituting these estimates gives the bound on \(|\mathcal I_{Q_{Y,h}^{\operatorname{sinh},\eta}}|\) in \cref{eq:sinh-cardinality-bound}. Moreover, \(Y_h=n_hh\le Y+h\le Y+1\), and therefore
\begin{equation}
\sinh(Y_h)\le \sinh(Y+1)=\mathcal O(\sinh Y+\cosh Y)=\mathcal O\left(1+\frac{\sqrt{\ell_{\operatorname{tail}}^{\operatorname{sinh},\eta}}}{\lambda_\eta}\right).
\end{equation}
This gives
\begin{equation}
R_{Q_{Y,h}^{\operatorname{sinh},\eta}}=\eta\sinh(Y_h)=\mathcal O\left(\eta+\frac{\eta}{\lambda_\eta}\sqrt{\ell_{\operatorname{tail}}^{\operatorname{sinh},\eta}}\right)=\mathcal O\left(\eta+2c\sqrt{\ell_{\operatorname{tail}}^{\operatorname{sinh},\eta}}\right),
\end{equation}
which proves the radius bound in \cref{eq:sinh-radius-bound}. The balanced-budget estimates follow from \(\ell_{\operatorname{mesh}}^{\operatorname{sinh},\eta}=\ell_{\operatorname{tail}}^{\operatorname{sinh},\eta}=\mathcal O(\log(1/\epsilon_q))\). It remains to bound the LCU normalization and the discrete moments. For real \(k\), \cref{eq:real-kernel-modulus} implies
\begin{equation}\label{eq:sinh-kernel-sharp-majorant}
\frac{1}{\sqrt{2\pi}}|\hat f_{a,b}(k;c,d)|=\widetilde C_{\operatorname{ker}}\frac{e^{-k^2/(4c^2)}}{\sqrt{1+k^2}(b^2+k^2)^{(a-1)/2}}.
\end{equation}
For the free-scale sinh rule, \(k_i=\eta\sinh(qh)\), and hence
\begin{align}
M_{\nu,Y,h}^{\operatorname{sinh},\eta} &\le \widetilde C_{\operatorname{ker}}h\sum_{q\in\mathbb Z}\eta\cosh(qh)\frac{e^{-\lambda_\eta^2\sinh^2(qh)}|\eta\sinh(qh)|^\nu}{\sqrt{1+\eta^2\sinh^2(qh)}\left(b^2+\eta^2\sinh^2(qh)\right)^{(a-1)/2}}.
\label{eq:sinh-moment-sharp-start}
\end{align}
For \(k>0\), we have
\begin{equation}\label{eq:sinh-denom-min-bound}
\frac{1}{\sqrt{1+k^2}(b^2+k^2)^{(a-1)/2}}\le \min\left\{\frac{1}{b^{a-1}},\frac{1}{k^a}\right\}.
\end{equation}
Indeed, the first bound follows from \(\sqrt{1+k^2}\ge1\) and \((b^2+k^2)^{(a-1)/2}\ge b^{a-1}\), while the second follows from \(\sqrt{1+k^2}\ge k\) and \((b^2+k^2)^{(a-1)/2}\ge k^{a-1}\). Let
\begin{equation}\label{eq:sinh-kappa-b-def}
\kappa_b:=\begin{cases}1, & a=1,\\ b^{(a-1)/a}, & a>1.\end{cases}
\end{equation}
We split the full-line sum in \cref{eq:sinh-moment-sharp-start} into the regions \(|\eta\sinh(qh)|\le\kappa_b\) and \(|\eta\sinh(qh)|>\kappa_b\). Since \(h\le1\), the contribution from \(|\eta\sinh(qh)|\le\kappa_b\) is bounded by a constant depending only on \(\nu,a,b\). More precisely,
\begin{align}
h\sum_{|\eta\sinh(qh)|\le\kappa_b}\eta\cosh(qh)\frac{e^{-\lambda_\eta^2\sinh^2(qh)}|\eta\sinh(qh)|^\nu}{\sqrt{1+\eta^2\sinh^2(qh)}\left(b^2+\eta^2\sinh^2(qh)\right)^{(a-1)/2}} \le C_{\nu,a,b}.
\end{align}
On the complementary region, \cref{eq:sinh-denom-min-bound} gives
\begin{align}
&h\sum_{|\eta\sinh(qh)|>\kappa_b}\eta\cosh(qh)\frac{e^{-\lambda_\eta^2\sinh^2(qh)}|\eta\sinh(qh)|^\nu}{\sqrt{1+\eta^2\sinh^2(qh)}\left(b^2+\eta^2\sinh^2(qh)\right)^{(a-1)/2}} \\
&\quad\le h\sum_{|\eta\sinh(qh)|>\kappa_b}\eta\cosh(qh)e^{-\lambda_\eta^2\sinh^2(qh)}|\eta\sinh(qh)|^{\nu-a}.
\end{align}
This transformed Riemann sum has the same three regimes as the corresponding continuous kernel moment. Indeed, using the change of variables \(k=\eta\sinh x\), \(dk=\eta\cosh x\,dx\), the continuous analogue is \(\int_{\kappa_b}^{\infty}e^{-k^2/(4c^2)}k^{\nu-a}\,dk\). The full-line Riemann sums with \(0<h\le1\) are controlled by the corresponding integrals plus a constant multiple of the maximal summand, and hence have the same scaling. Therefore, if \(\nu<a-1\), then
\begin{equation}
h\sum_{|\eta\sinh(qh)|>\kappa_b}\eta\cosh(qh)e^{-\lambda_\eta^2\sinh^2(qh)}|\eta\sinh(qh)|^{\nu-a}\le C_{\nu,a,b}.
\end{equation}
If \(\nu=a-1\), then
\begin{equation}
h\sum_{|\eta\sinh(qh)|>\kappa_b}\eta\cosh(qh)e^{-\lambda_\eta^2\sinh^2(qh)}|\eta\sinh(qh)|^{-1}\le C_{a,b}\bigl(1+\log(1+c)\bigr).
\end{equation}
If \(\nu>a-1\), then comparison with the corresponding Gaussian integral gives
\begin{align}
h\sum_{|\eta\sinh(qh)|>\kappa_b}\eta\cosh(qh)e^{-\lambda_\eta^2\sinh^2(qh)}|\eta\sinh(qh)|^{\nu-a} &\le C_{\nu,a,b}+C_\nu\int_0^\infty e^{-k^2/(4c^2)}k^{\nu-a}\,dk \le C_{\nu,a,b}+C_\nu c^{\nu-a+1}.
\end{align}
Combining these three cases with \cref{eq:sinh-moment-sharp-start} proves \cref{eq:sinh-sharp-moment-bound}. The normalization bound is the case \(\nu=0\), and the asymptotic statement follows directly from the definition of \(\mathfrak M_{\nu,a,b}(c)\).
\end{proof}

\section{Omitted Proofs in \cref{applications}}

\subsection{Proof of \cref{lem:fractional-symbol-difference-simple}}\label{proof-of-fractional-symbol-difference-simple}

\begin{proof}
Let $b(x):=4\sum_{\alpha=1}^d \sin^2(x_\alpha/2)$ and $a_s(x):=b(x)^s$ for  $x\in\mathbb T^d$. We claim that \(a_s\) satisfies the \(\beta_s\)-H\"older condition. Define $F(x):=(2|\sin(x_\alpha/2)|)_{\alpha=1}^d$. Then $b(x)^{1/2}=|F(x)|$.  Since each map \(x_\alpha\mapsto 2|\sin(x_\alpha/2)|\) is Lipschitz on \(\mathbb T\), we have
\begin{equation}
|F(x)-F(y)|\le C_d |x-y|_{\operatorname{per}}.
\end{equation}
for some $C_d > 0$. Using the reverse triangle inequality, we have $|b(x)^{1/2} - b(y)^{1/2}| \le C_d |x-y|_{\operatorname{per}}$. We consider two cases:
\begin{enumerate}
\item If \(0<s\le 1/2\), then \(0<2s\le 1\), and the map \(r\mapsto r^{2s}\) is \(2s\)-Hölder on \([0,\infty)\). Hence, we have
\begin{equation}
|a_s(x)-a_s(y)| = |(b(x)^{1/2})^{2s} - (b(y)^{1/2})^{2s}| \le |b(x)^{1/2}-b(y)^{1/2}|^{2s} \le C'_{s,d}|x-y|_{\operatorname{per}}^{2s},
\end{equation}
where $C'_{s,d} = C^{2s}_d$.

\item If \(1/2\le s<1\), then \(2s\ge 1\). Note that $b(x)^{1/2}\le 2\sqrt d$ and the map \(r\mapsto r^{2s}\) is Lipschitz on \([0,2\sqrt d]\). Therefore, we have
\begin{equation}
|a_s(x)-a_s(y)| = |(b(x)^{1/2})^{2s} - (b(y)^{1/2})^{2s}| \le C''_{s,d}|b(x)^{1/2}-b(y)^{1/2}| \le C'_{s,d}|x-y|_{\operatorname{per}},
\end{equation}
where $C'_{s,d} = C''_{s,d} C_d$
\end{enumerate}
Combining the two cases gives $|a_s(x)-a_s(y)| \le C'_{s,d}|x-y|_{\operatorname{per}}^{\beta_s}$, where $\beta_s=\min\{2s,1\}$ and $ C'_{s,d} > 0$ is independent of $N$. Taking \(x=2\pi p/M\) and \(y=2\pi q/M\) gives
\begin{align}
|\lambda_p^s-\lambda_q^s| &= M^{2s}\left|a_s\!\left(\frac{2\pi p}{M}\right)-a_s\!\left(\frac{2\pi q}{M}\right)\right|  \\ & \le (2\pi)^{\beta_s} C'_{s,d}M^{2s-\beta_s}|p-q|_{\mathrm{per}}^{\beta_s} \\ & := C_{s,d}M^{\theta_s}|p-q|_{\mathrm{per}}^{\beta_s} = C_{s,d}N^{\theta_s/d}|p-q|_{\mathrm{per}}^{\beta_s},
\end{align}
where $C_{s,d} = (2\pi)^{\beta_s} C'_{s,d}$. This completes the proof.
\end{proof}

\subsection{Proof of \cref{fractional-general-word-simple}}\label{proof-of-fractional-general-word-simple}

\begin{proof}
For $a\geq 0$, define the following weighted row-column norm in the Fourier basis:
\begin{equation}\label{eq:fourier-weighted-norm-current}
\|A\|_{a,\mathcal F} := \max\left\{ \max_{p\in\mathbb Z_M^d} \sum_{q\in\mathbb Z_M^d} \left(1+|p-q|_{\operatorname{per}}\right)^a |(F_NAF_N^\dagger)_{p,q}|, \max_{q\in\mathbb Z_M^d} \sum_{p\in\mathbb Z_M^d} \left(1+|p-q|_{\operatorname{per}}\right)^a |(F_NAF_N^\dagger)_{p,q}| \right\}.
\end{equation}
It can be checked that \cref{eq:fourier-weighted-norm-current} indeed defines a norm. For $a=0$, this norm controls the spectral norm. Indeed, since $F_N$ is unitary and $\|A\|\leq \sqrt{\|A\|_1\|A\|_\infty}$ for any matrix $A$, we have
\begin{equation}\label{eq:fourier-weighted-controls-spectral-current}
\|A\| = \|F_N A F_N^\dagger\| \leq \sqrt{\|F_N A F_N^\dagger\|_1\|F_N A F_N^\dagger\|_\infty} \leq \|A\|_{0,\mathcal F}.
\end{equation}
We first record the two commutator-type operator estimates used below.
\begin{enumerate}
\item Recall from \cref{eq:comm-with-Ls-fourier} that, for every matrix $A$, we have
\begin{equation}\label{eq:Ls-commutator-weighted-current}
(F_N\operatorname{ad}_{L_s}(A)F_N^\dagger )_{p,q} = (\lambda_p^s-\lambda_q^s)(F_NAF_N^\dagger)_{p,q}.
\end{equation}
Applying the weaker consequence of  \cref{lem:fractional-symbol-difference-simple} from \cref{eq:fractional-symbol-difference-weak}, we obtain
\begin{equation}\label{eq:Ls-weighted-bound-current}
\|\operatorname{ad}_{L_s}(A)\|_{a,\mathcal F} \leq C_{s,d}N^{\theta_s/d} \|A\|_{a+\beta_s,\mathcal F}.
\end{equation}
Therefore, $\operatorname{ad}_{L_s}$ maps the Fourier-weighted norm with exponent $a+\beta_s$ into the corresponding norm with exponent $a$, at the cost of the factor $C_{s,d}N^{\theta_s/d}$.

\item Next, since $F_NHF_N^\dagger$ has entries $\widehat V_{p-q}$, multiplication by $H$ becomes convolution in the Fourier basis. We use the submultiplicative weight estimate
\begin{equation}\label{eq:periodic-submultiplicative-weight}
1+|p-q|_{\operatorname{per}} \leq \bigl(1+|p-r|_{\operatorname{per}}\bigr) \bigl(1+|r-q|_{\operatorname{per}}\bigr),
\end{equation}
which follows from the triangle inequality for the periodic distance. We claim that
\begin{equation}\label{eq:H-product-weighted-current}
\|HA\|_{a,\mathcal F} \leq \mathcal V_a^{(N)}\|A\|_{a,\mathcal F}, \qquad \|AH\|_{a,\mathcal F} \leq \mathcal V_a^{(N)}\|A\|_{a,\mathcal F}.
\end{equation}
Indeed, we have
\begin{equation}
(F_NHAF_N^\dagger)_{p,q} = ((F_NHF_N^\dagger) (F_N AF_N^\dagger))_{p,q} = \sum_{r\in\mathbb Z_M^d}  \widehat V_{p-r}(F_NAF_N^\dagger)_{r,q}.
\end{equation}
Hence, by \cref{eq:periodic-submultiplicative-weight} we have
\begin{align}
&\sum_{q\in\mathbb Z_M^d} (1+|p-q|_{\operatorname{per}})^a |(F_NHAF_N^\dagger)_{p,q} | \\ &\leq \sum_{r\in\mathbb Z_M^d} (1+|p-r|_{\operatorname{per}})^a|\widehat V_{p-r}| \sum_{q\in\mathbb Z_M^d} (1+|r-q|_{\operatorname{per}})^a |(F_NAF_N^\dagger)_{r,q} | \leq \mathcal V_a^{(N)}\|A\|_{a,\mathcal F}.
\end{align}
Taking the maximum over $p$ gives the row estimate for $HA$. The column estimate is identical. The bound for $AH$ follows in the same way. Thus \cref{eq:H-product-weighted-current} holds. Consequently, we have
\begin{equation}\label{eq:H-weighted-bound-current}
\|\operatorname{ad}_{H}(A)\|_{a,\mathcal F} = \|[H,A]\|_{a,\mathcal F} \leq \|HA\|_{a,\mathcal F}+\|AH\|_{a,\mathcal F} \leq 2\mathcal V_a^{(N)}\|A\|_{a,\mathcal F}.
\end{equation}
Therefore, $\operatorname{ad}_H$ is bounded on each weighted Fourier norm and contributes no explicit $N^{\theta_s/d}$ factor.
\end{enumerate}
We now apply these two estimates to the nested word. Write $W_{j,\ell}(H,L_s) = [A_j,[A_{j-1},\dots,[A_2,A_1]\cdots]]$ for $A_i\in\{H,L_s\}$. For \(m=1,\dots,j-1\), define the partial commutator
\begin{equation}\label{eq:partial-word-def}
W^{(m)} := [A_{m+1},[A_m,\dots,[A_2,A_1]\cdots]].
\end{equation}
Equivalently, \(W^{(m)}=\operatorname{ad}_{A_{m+1}}(W^{(m-1)})\) for \(m=2,\dots,j-1\). Let \(n_m:=m+1\), \(\ell_m\) be the number of occurrences of \(L_s\) among \(A_1,\dots,A_{m+1}\), and let \(h_m:=n_m-\ell_m\) be the number of occurrences of \(H\) among the same operators. We prove by induction on \(m\) that, for every \(a\geq0\),
\begin{equation}\label{eq:induction-claim}
\|W^{(m)}\|_{a,\mathcal F} \leq 2^{h_m} C_{s,d}^{\ell_m} (\mathcal V_{a+n_m\beta_s}^{(N)} )^{h_m} N^{\ell_m\theta_s/d}.
\end{equation}
\begin{enumerate}
\item Let \(m=1\). We have $W^{(1)}=[A_2,A_1]$. If \(A_2=A_1\), then \(W^{(1)}=0\), so the estimate is immediate. Otherwise, \(W^{(1)} = \pm [L_s,H] = \pm \operatorname{ad}_{L_s}(H)\). Using \cref{eq:Ls-weighted-bound-current}, we obtain
\begin{equation}
\|W^{(1)}\|_{a,\mathcal F} \leq C_{s,d}N^{\theta_s/d} \|H\|_{a+\beta_s,\mathcal F}.
\end{equation}
Since \(F_NHF_N^\dagger\) has entries \(\widehat V_{p-q}\), we have
\begin{equation}
\|H\|_{a+\beta_s,\mathcal F} \leq \mathcal V_{a+\beta_s}^{(N)} \leq \mathcal V_{a+2\beta_s}^{(N)}.
\end{equation}
In the non-zero base case, \(n_1=2\), \(\ell_1=1\), and \(h_1=1\). Hence, we obtain \cref{eq:induction-claim}. This proves the base case.

\item Assume \cref{eq:induction-claim} holds for \(W^{(m-1)}\), where \(2\leq m\leq j-1\). Since $W^{(m)} = \operatorname{ad}_{A_{m+1}}(W^{(m-1)}),$ there are two cases.
\begin{enumerate}
\item If \(A_{m+1}=L_s\), then \cref{eq:Ls-weighted-bound-current} gives
\begin{align}
\|W^{(m)}\|_{a,\mathcal F} = \|\operatorname{ad}_{L_s}(W^{(m-1)})\|_{a,\mathcal F} \leq C_{s,d}N^{\theta_s/d} \|W^{(m-1)}\|_{a+\beta_s,\mathcal F}.
\end{align}
Applying the induction hypothesis to \(W^{(m-1)}\) with weight \(a+\beta_s\), we obtain
\begin{align}
\|W^{(m)}\|_{a,\mathcal F} &\leq C_{s,d}N^{\theta_s/d} 2^{h_{m-1}} C_{s,d}^{\ell_{m-1}} ( \mathcal V_{a+\beta_s+n_{m-1}\beta_s}^{(N)} )^{h_{m-1}} N^{\ell_{m-1}\theta_s/d}.
\end{align}
Because \(A_{m+1}=L_s\), we have $\ell_m=\ell_{m-1}+1, h_m=h_{m-1}$ and $n_m=n_{m-1}+1$. Thus, $a+\beta_s+n_{m-1}\beta_s=a+n_m\beta_s$. Hence, we have
\begin{equation}
\|W^{(m)}\|_{a,\mathcal F} \leq 2^{h_m} C_{s,d}^{\ell_m} ( \mathcal V_{a+n_m\beta_s}^{(N)})^{h_m} N^{\ell_m\theta_s/d}.
\end{equation}

\item If \(A_{m+1}=H\), then \cref{eq:H-weighted-bound-current} gives
\begin{align}
\|W^{(m)}\|_{a,\mathcal F} = \|\operatorname{ad}_{H}(W^{(m-1)})\|_{a,\mathcal F} \leq 2\mathcal V_a^{(N)} \|W^{(m-1)}\|_{a,\mathcal F}.
\end{align}
Applying the induction hypothesis gives
\begin{align}
\|W^{(m)}\|_{a,\mathcal F} &\leq 2\mathcal V_a^{(N)} 2^{h_{m-1}} C_{s,d}^{\ell_{m-1}} ( \mathcal V_{a+n_{m-1}\beta_s}^{(N)} )^{h_{m-1}} N^{\ell_{m-1}\theta_s/d}.
\end{align}
Because \(A_{m+1}=H\), we have $\ell_m=\ell_{m-1}$, $h_m=h_{m-1}+1$ and $n_m=n_{m-1}+1$. Moreover, \(\mathcal V_b^{(N)}\) is non-decreasing in \(b\). Therefore, we have $\mathcal V_a^{(N)} \leq \mathcal V_{a+n_m\beta_s}^{(N)}$ and $\mathcal V_{a+n_{m-1}\beta_s}^{(N)} \leq \mathcal V_{a+n_m\beta_s}^{(N)}$. Hence, it follows that
\begin{equation}
\|W^{(m)}\|_{a,\mathcal F} \leq 2^{h_m} C_{s,d}^{\ell_m} ( \mathcal V_{a+n_m\beta_s}^{(N)} )^{h_m} N^{\ell_m\theta_s/d}.
\end{equation}
This completes the induction.
\end{enumerate}
\end{enumerate}
Taking \(m=j-1\) and \(a=0\), we have $n_{j-1}=j$, $\ell_{j-1}=\ell$ and $h_{j-1}=j-\ell$. Thus \cref{eq:induction-claim} gives
\begin{equation}
\|W_{j,\ell}(H,L_s)\|_{0,\mathcal F} \leq 2^{j-\ell} C_{s,d}^{\ell} (\mathcal V_{j\beta_s}^{(N)} )^{j-\ell} N^{\ell\theta_s/d}.
\end{equation}
The claim follows since $\|W_{j,\ell}(H,L_s)\|\leq \|W_{j,\ell}(H,L_s)\|_{0,\mathcal F}$.
\end{proof}

\subsection{Proof of \cref{dirichlet-nested-commutator-bound-1d}}\label{proof-of-dirichlet-nested-commutator-bound-1d}

\begin{proof}
For \(0\leq p,q\leq r\), we have $(S^T)^pP_0S^q=e_pe_q^T$.
Indeed, since \(S=\sum_{j=0}^{M-1}e_je_{j+1}^T\), we have $Se_{j+1}=e_j$ and $S^T e_j=e_{j+1}$. Thus, $(S^T)^pe_0=e_p$. Also, we have $e_0^TS^q=e_q^T$. Since \(P_0=e_0e_0^T\), it follows that
\begin{equation}
(S^T)^pP_0S^q = (S^T)^p e_0 e_0^T S^q = e_pe_q^T.
\end{equation}
Similarly, we have $S^pP_M(S^T)^q=e_{M-p}e_{M-q}^T$. Each non-zero matrix of the form $e_pe_q^T$ has operator norm one. Therefore, for every representation \(X=\sum_\eta c_\eta E_\eta\) the triangle inequality gives $\|X\|\leq \sum_\eta |c_\eta|$. Taking the infimum over all such representations proves $\|X\|\leq \|X\|_{\mathcal B}$. We now show that for every \(X\in\mathcal B_r\), we have $[A_1,X]\in\mathcal B_{r+1}$, $[B_1,X]\in\mathcal B_{r+1}$. It suffices to check it on the generators of \(\mathcal B_r\). First consider $(S^T)^pP_0S^q=e_pe_q^T$ for some $0\leq p,q\leq r$. Since \(S e_k=e_{k-1}\) for \(1\leq k\leq M\) and \(S e_0=0\), while \(S^T e_k=e_{k+1}\) for \(0\leq k\leq M-1\) and \(S^T e_M=0\), we have
\begin{equation}
S e_p e_q^T = \begin{cases} e_{p-1}e_q^T, & p\geq1,\\ 0, & p=0, \end{cases} \quad \quad S^T e_p e_q^T = \begin{cases} e_{p+1}e_q^T, & p\leq M-1,\\ 0, & p=M. \end{cases}
\end{equation}
Similarly, using \(e_q^T S = e_{q+1}^T\) for \(q \leq M-1\), \(e_M^T S = 0\), \(e_q^T S^T = e_{q-1}^T\) for \(q \geq 1\), and \(e_0^T S^T = 0\), we obtain
\begin{equation}
e_p e_q^TS = \begin{cases} e_pe_{q+1}^T, & q\leq M-1,\\ 0, & q=M, \end{cases} \quad \quad e_p e_q^T S^T = \begin{cases} e_pe_{q-1}^T, & q\geq1,\\ 0, & q=0. \end{cases}
\end{equation}
Therefore each of $S e_pe_q^T, e_pe_q^TS,  S^Te_pe_q^T, e_pe_q^TS^T$ is either zero or a matrix unit \(e_{p'}e_{q'}^T\) with \(0 \leq p',q' \leq r+1\), and hence belongs to \(\mathcal B_{r+1}\). The case \(S^p P_M (S^T)^q = e_{M-p}e_{M-q}^T\) is similar. Since every \(X\in\mathcal B_r\) is a linear combination of the generators \((S^T)^pP_0S^q\) and \(S^p P_M (S^T)^q\), we have that $SX,\;XS,\;S^T X,\;XS^T\in\mathcal B_{r+1}$. Using the definitions of \(A_1\) and \(B_1\), we have
\begin{align}
[A_1,X] &= -SX - S^T X + XS + XS^T, \\ [B_1,X] &= SX - S^T X - XS + XS^T.
\end{align}
Hence, \([A_1,X]\) and \([B_1,X]\) belong to \(\mathcal B_{r+1}\) for each \(X \in \mathcal B_r\). If \(X=\sum_\eta c_\eta E_\eta\) is any representation of \(X\), then
\begin{equation}
SX=\sum_\eta c_\eta SE_\eta, \quad S^TX=\sum_\eta c_\eta S^TE_\eta, \quad XS=\sum_\eta c_\eta E_\eta S, \quad XS^T=\sum_\eta c_\eta E_\eta S^T.
\end{equation}
By the argument above, each \(SE_\eta,S^TE_\eta,E_\eta S,E_\eta S^T\) is either zero or in \(\mathcal B_{r+1}\). Hence each of the four shifted matrices admits a representation in $\mathcal{B}_{r+1}$ with coefficient cost at most \(\sum_\eta |c_\eta|\). Taking the infimum over all representations of \(X\) gives
\begin{equation}
\|SX\|_{\mathcal B},\ \|S^TX\|_{\mathcal B},\ \|XS\|_{\mathcal B},\ \|XS^T\|_{\mathcal B} \leq \|X\|_{\mathcal B}.
\end{equation}
Using the sub-additivity of \(\|\cdot\|_{\mathcal B}\) gives $\|[A_1,X]\|_{\mathcal B}\leq4\|X\|_{\mathcal B}$ and  $\|[B_1,X]\|_{\mathcal B}\leq4\|X\|_{\mathcal B}$. By \cref{lem:dirichlet-first-commutator-1d}, we have $[A_1,B_1]=2(P_0-P_M)$. Since \(P_0,P_M\in\mathcal B_0\), this implies $\|[A_1,B_1]\|_{\mathcal B}\leq4$ and  $\|[B_1,A_1]\|_{\mathcal B}\leq4$. Every non-zero right-nested mixed commutator of length \(j\geq2\) is obtained from \([A_1,B_1]\) or \([B_1,A_1]\) by \(j-2\) additional commutations with \(A_1\) or \(B_1\). Applying the argument above  inductively gives
\begin{equation}
\|W_j(A_1,B_1)\|_{\mathcal B}\leq4\cdot4^{j-2}=4^{j-1}.
\end{equation}
Hence, we have $\|W_j(A_1,B_1)\|\leq\|W_j(A_1,B_1)\|_{\mathcal B}\leq4^{j-1}$. This completes the proof.
\end{proof}

\section{Extension to Non-Homogeneous and Time-Dependent Cases}\label{sec:extensions}
In \cref{sec:alg-err-analysis} and \cref{complexity-analysis}, we analyzed the time-independent homogeneous case in detail. In this section, we explain how the same framework extends to more general cases, namely the non-homogeneous case in \cref{sec:non-homogeneous-extension} and the time-dependent case in \cref{sec:time-dependent-extension}.

\subsection{Non-Homogeneous Case}\label{sec:non-homogeneous-extension}
We first consider the time-independent non-homogeneous problem
\begin{equation}\label{eq:non-homogeneous-time-independent}
\frac{d}{dt}u(t)=-Au(t)+b(t).
\end{equation}
As before, write \(A=L+iH\), with \(L,H\) Hermitian and \(L\succeq 0\). By Duhamel's principle, the solution of \cref{eq:non-homogeneous-time-independent} is
\begin{equation}\label{eq:duhamel-time-independent}
u(T)=e^{-AT}u(0)+\int_0^T e^{-A(T-s)}b(s)\,ds.
\end{equation}
\cref{eq:duhamel-time-independent} shows that the non-homogeneous term is a superposition of homogeneous evolutions over times \(T-s\). Applying the LCHS representation to each homogeneous propagator in the Duhamel term gives the two-variable representation
\begin{equation}\label{eq:duhamel-LCHS-double-integral}
\int_0^T e^{-A(T-s)}b(s)\,ds \approx \frac{1}{\sqrt{2\pi}}\int_0^T\int_{\mathbb R}\hat f_{a,b}(k;c,d)U_k(T-s)b(s)\,dk\,ds.
\end{equation}
Thus, the non-homogeneous contribution involves two discretizations: the LCHS quadrature \(Q\) in the Fourier variable \(k\), and the time quadrature \(P\) in the Duhamel variable \(s\). This viewpoint is used below. For \(\tau\in[0,T]\), define the ideal and MPF-implemented post-quadrature operators by
\begin{equation}\label{eq:WQ-tau-ideal-mpf}
W_Q^{\operatorname{ideal}}(\tau):=\sum_{i\in\mathcal I_Q}v_iU_{k_i}(\tau), \qquad W_Q^{\operatorname{MPF}}(\tau):=\sum_{i\in\mathcal I_Q}v_iU^{(2m)}_{k_i}(\tau),
\end{equation}
The quantities \(\alpha_Q\), \(\Lambda_{m,Q}\), and \(\Phi_{m,Q}^{\ast}\) are defined as before. Let \(P=\{(s_p,\omega_p):p\in\mathcal I_P\}\) be a quadrature rule on \([0,T]\) for the integral in \cref{eq:duhamel-LCHS-double-integral}. We write
\begin{equation}
\tau_p:=T-s_p, \qquad B_P:=\sum_{p\in\mathcal I_P}|\omega_p|\,\|b(s_p)\|.
\end{equation}

We assume \(B_P<\infty\). The assumption \(b\in L^1([0,T])\) suffices for \cref{eq:duhamel-time-independent}, but does not define point values \(b(s_p)\). The quadrature implementation is therefore stated under the additional source-access assumption that the selected values \(b(s_p)\) are well-defined and accessible via a state-preparation oracles.   Discretizing \cref{eq:duhamel-LCHS-double-integral} by \(Q\) in \(k\) and \(P\) in \(s\), and then replacing each \(U_{k_i}(\tau_p)\) by its MPF approximation, gives
\begin{equation}\label{eq:duhamel-LCHS-MPF-double-sum}
\frac{1}{\sqrt{2\pi}}\int_0^T\int_{\mathbb R}\hat f_{a,b}(k;c,d)U_k(T-s)b(s)\,dk\,ds \approx \sum_{p\in\mathcal I_P}\sum_{i\in\mathcal I_Q}\omega_p v_iU^{(2m)}_{k_i}(\tau_p)b(s_p).
\end{equation}
Equivalently, the non-homogeneous LCHS--MPF approximation is
\begin{equation}\label{eq:non-homogeneous-LCHS-MPF-approx}
u_{Q,P}^{\operatorname{MPF}}(T):=W_Q^{\operatorname{MPF}}(T)u(0)+\sum_{p\in\mathcal I_P}\omega_pW_Q^{\operatorname{MPF}}(\tau_p)b(s_p).
\end{equation}
We state the error estimate in terms of two quadrature errors. First, define the uniform-in-time LCHS quadrature error by
\begin{equation}\label{eq:uniform-quadrature-error}
E_{\operatorname{quad}}^{[0,T]}(Q):=\sup_{0\leq \tau\leq T}\left\|\frac{1}{\sqrt{2\pi}}\int_{\mathbb R}\hat f_{a,b}(k;c,d)U_k(\tau)\,dk-W_Q^{\operatorname{ideal}}(\tau)\right\|.
\end{equation}
The quadrature estimates in \cref{quad-rules} also bound \(E_{\operatorname{quad}}^{[0,T]}(Q)\). Indeed, the strip estimates used there are monotone in the final time. Replacing \(T\) by any \(\tau\in[0,T]\) only decreases the factor involving \(T\|L\|\). Hence the same quadrature rule \(Q\), chosen using the endpoint time \(T\), satisfies
\begin{equation}
E_{\operatorname{quad}}^{[0,T]}(Q)\le E_{\operatorname{mesh}}(Q;T)+E_{\operatorname{tail}}(Q;T),
\end{equation}
with the right-hand side given by the corresponding estimates in \cref{trap-budget} or \cref{sinh-budget}. Second, define the time-quadrature error by
\begin{equation}\label{eq:time-quadrature-error}
E_{\operatorname{time}}(P):=\left\|\int_0^T e^{-A(T-s)}b(s)\,ds-\sum_{p\in\mathcal I_P}\omega_p e^{-A(T-s_p)}b(s_p)\right\|.
\end{equation}
The quantity \(E_{\operatorname{time}}(P)\) is kept abstract below. Once \(P\) and the regularity of the representative \(b\) are specified, it can be bounded by the corresponding quadrature estimate.

\begin{lem}\label{lem:uniform-propagator-error-nonhom}
Fix an admissible kernel profile \(\vec{\theta}=(a,b,c,d)\) and a quadrature rule \(Q\). Let \(T,\delta>0\), \(m,r,N\in\mathbb N\), and set \(\Delta=T/r\). Assume that \(0<\Delta\leq\Delta(R_Q,N,\delta)\). Then
\begin{align}\label{eq:uniform-propagator-error-nonhom}
\sup_{0\leq\tau\leq T}\left\|e^{-A\tau}-W_Q^{\operatorname{MPF}}(\tau)\right\| &\le E_{\operatorname{approx}}(y_0)+E_{\operatorname{quad}}^{[0,T]}(Q) \\ &\quad +(1+\eta_Q(T/r))^{r-1}\left(C_m\frac{T^{2m+1}}{r^{2m}}\Lambda_{m,Q}+r\alpha_QR_m(\delta)\right).
\label{eq:uniform-propagator-error-nonhom-second-eq}
\end{align}
\end{lem}

\begin{proof}
The proof is the same as that of \cref{post-quadrature-error-decomposition-general}, applied uniformly for \(\tau\in[0,T]\). The approximation and quadrature errors are bounded by \(E_{\operatorname{approx}}(y_0)\) and \(E_{\operatorname{quad}}^{[0,T]}(Q)\), respectively. For the MPF term, the step size over time \(\tau\) is \(\tau/r\), so
\begin{equation}
\tau/r\leq T/r\leq \Delta(R_Q,N,\delta).
\end{equation}
Since \(\tau\leq T\), the principal MPF error term is bounded by the corresponding expression with \(T\) in place of \(\tau\). This proves \cref{eq:uniform-propagator-error-nonhom,eq:uniform-propagator-error-nonhom-second-eq}.
\end{proof}

We now combine the three error sources in the non-homogeneous construction. The homogeneous propagator error is needed uniformly for all \(\tau\in[0,T]\), since the source term contains propagators \(e^{-A(T-s)}\). \cref{prop:non-homogeneous-error-bound} separates this uniform LCHS--MPF error from the additional time-quadrature error used to discretize the source term.

\begin{prop}\label{prop:non-homogeneous-error-bound}
Let \(u(T)\) be the solution of \cref{eq:non-homogeneous-time-independent}, and let \(u_{Q,P}^{\operatorname{MPF}}(T)\) be defined by \cref{eq:non-homogeneous-LCHS-MPF-approx}. Then
\begin{equation}\label{eq:non-homogeneous-error-bound}
\left\|u(T)-u_{Q,P}^{\operatorname{MPF}}(T)\right\| \leq E_{\operatorname{time}}(P)+\mathcal E_Q^{[0,T]}(r,\delta)\left(\|u(0)\|+B_P\right),
\end{equation}
where
\begin{align}\label{eq:uniform-EQ-def}
\mathcal E_Q^{[0,T]}(r,\delta) := E_{\operatorname{approx}}(y_0)+E_{\operatorname{quad}}^{[0,T]}(Q)+(1+\eta_Q(T/r))^{r-1}\left(C_m\frac{T^{2m+1}}{r^{2m}}\Lambda_{m,Q}+r\alpha_QR_m(\delta)\right).
\end{align}
\end{prop}

\begin{proof}
Using Duhamel's formula, \cref{eq:duhamel-time-independent}, and the definition of \(u_{Q,P}^{\operatorname{MPF}}(T)\), we write
\begin{align}
u(T)-u_{Q,P}^{\operatorname{MPF}}(T) &= \left(e^{-AT}-W_Q^{\operatorname{MPF}}(T)\right)u(0) \\ &\quad+\left(\int_0^T e^{-A(T-s)}b(s)\,ds-\sum_{p\in\mathcal I_P}\omega_p e^{-A(T-s_p)}b(s_p)\right) \\ &\quad+\sum_{p\in\mathcal I_P}\omega_p\left(e^{-A\tau_p}-W_Q^{\operatorname{MPF}}(\tau_p)\right)b(s_p).
\end{align}
Taking norms, using \cref{eq:time-quadrature-error}, and applying \cref{lem:uniform-propagator-error-nonhom} to \(T\) and to each \(\tau_p\in[0,T]\), we obtain
\begin{align}
\left\|u(T)-u_{Q,P}^{\operatorname{MPF}}(T)\right\| &\leq \mathcal E_Q^{[0,T]}(r,\delta)\|u(0)\|+E_{\operatorname{time}}(P)+\mathcal E_Q^{[0,T]}(r,\delta)\sum_{p\in\mathcal I_P}|\omega_p|\,\|b(s_p)\|.
\end{align}
This proves \cref{eq:non-homogeneous-error-bound}.
\end{proof}

\cref{prop:non-homogeneous-error-bound} gives an abstract error decomposition. We next choose the time-quadrature, approximation, and post-quadrature error budgets so that their total contribution is at most the target accuracy \(\epsilon\). This gives the fixed-order non-homogeneous analogue of \cref{abstract-fixed-order-arbitrary-quadrature}.

\begin{cor}\label{cor:non-homogeneous-fixed-order}
Let \(\epsilon>0\). Choose \(\epsilon_{\operatorname{time}},\epsilon_{\operatorname{approx}}\), and \(\epsilon_{\operatorname{comb}}>0\) such that
\begin{equation}\label{eq:non-homogeneous-budget}
\epsilon_{\operatorname{time}}+\left(\|u(0)\|+B_P\right)\left(\epsilon_{\operatorname{approx}}+\epsilon_{\operatorname{comb}}\right)\leq \epsilon.
\end{equation}
Fix \(P\) and assume that \(E_{\operatorname{time}}(P)\leq \epsilon_{\operatorname{time}}\). Choose an admissible kernel satisfying \(E_{\operatorname{approx}}(y_0)\leq\epsilon_{\operatorname{approx}}\). Suppose that \(Q\) satisfies \(E_{\operatorname{quad}}^{[0,T]}(Q)<\epsilon_{\operatorname{comb}}\), and define \(\epsilon_Q:=\epsilon_{\operatorname{comb}}-E_{\operatorname{quad}}^{[0,T]}(Q)>0\). If \(r=r_Q(\delta)\) is chosen as in \cref{abstract-rQ-delta-def}, with this value of \(\epsilon_Q\), and if the admissibility conditions
\begin{equation}\label{eq:nonhom-admissibility-conditions}
r_Q(\delta)R_m(\delta)\leq\frac12, \qquad r_Q(\delta)\alpha_QR_m(\delta)\leq\frac{\epsilon_Q}{2e}
\end{equation}
hold, then \(\|u(T)-u_{Q,P}^{\operatorname{MPF}}(T)\|\leq\epsilon\).
\end{cor}

\begin{proof}
For \(r=r_Q(\delta)\), the same argument used in \cref{abstract-fixed-order-arbitrary-quadrature} gives \(\mathcal E_Q^{[0,T]}(r_Q(\delta),\delta)\leq \epsilon_{\operatorname{approx}}+\epsilon_{\operatorname{comb}}\). Combining this bound with \cref{prop:non-homogeneous-error-bound} and \cref{eq:non-homogeneous-budget} proves the claim.
\end{proof}

We now describe the implementation cost. Assume state-preparation oracles for \(\ket{\widetilde u(0)}:=u(0)/\|u(0)\|\), when \(u(0)\neq0\), and for \(\ket{\widetilde b(s_p)}:=b(s_p)/\|b(s_p)\|\), for each \(p\in\mathcal I_P\) with \(b(s_p)\neq0\). Terms with zero vector coefficients are omitted from the outer LCU, and the phases of \(\omega_p\) are included in the outer preparation. Let \(\beta_P:=\|u(0)\|+B_P\). Define the source-preparation oracle
\begin{equation}
\operatorname{PREP}^{P}\ket 0 = \frac{1}{\sqrt{\beta_P}}\left(\sqrt{\|u(0)\|}\ket{\star}+\sum_{p\in\mathcal I_P}\sqrt{|\omega_p|\,\|b(s_p)\|}\ket p\right),
\end{equation}
with the phases of \(\omega_p\) included in the corresponding left-preparation state. 
The vector in \cref{eq:non-homogeneous-LCHS-MPF-approx} is then implemented by an outer LCU over \(\{\star\}\cup\mathcal I_P\), with normalization \(\beta_P\). Since each term applies a block-encoding of \(W_Q^{\operatorname{MPF}}(\tau)\) with normalization \(\alpha_Q\), the total normalization is \(\alpha_Q\beta_P\). Thus, the amplitude amplification factor is $\chi_{Q,P}(u(0),b):=\frac{\alpha_Q\beta_P}{\|u_{Q,P}^{\operatorname{MPF}}(T)\|}$. \cref{cor:non-homogeneous-complexity} gives the complexity estimates.

\begin{cor}\label{cor:non-homogeneous-complexity}
Under the assumptions of \cref{cor:non-homogeneous-fixed-order}, the normalized state proportional to \(u_{Q,P}^{\operatorname{MPF}}(T)\) can be prepared using
\begin{equation}
\mathcal Q_{\operatorname{state}}=\mathcal O\!\left(\chi_{Q,P}(u(0),b)\,r_Q(\delta)K_m\right)
\end{equation}
controlled second-order product-formula queries, together with \(\mathcal O(\chi_{Q,P}(u(0),b))\) calls to the data-preparation oracles for \(\ket{\widetilde u(0)}\), \(\ket{\widetilde b(s_p)}\), and their inverses, as well as to \(\operatorname{PREP}^{P}\), \(\operatorname{PREP}_{\operatorname L}^{Q}\), and \(\operatorname{PREP}_{\operatorname R}^{Q}\). Moreover, under the optimized-order assumptions of \cref{abstract-optimized-order-arbitrary-quadrature}, we have
\begin{align}
\mathcal Q_{\operatorname{state}} = \mathcal O\left(\chi_{Q,P}(u(0),b)\left(1+\mu_QF(m)\max\{1,T\}\right)\left(\log\left(e+\frac{T}{\epsilon_Q}\right)\right)^2\left(\log\log\left(e^e+\frac{T}{\epsilon_Q}\right)\right)^2\right).
\label{eq:non-homogeneous-optimized-complexity}
\end{align}
\end{cor}

\begin{proof}
The implementation uses the nested LCU construction from \cref{implem}, with an additional LCU combining the initial-condition and discretized Duhamel terms. One block-encoding costs \(\mathcal O(r_Q(\delta)K_m)\) controlled second-order product-formula queries, and amplitude amplification uses \(\mathcal O(\chi_{Q,P}(u(0),b))\) such applications. The optimized estimates follow by substituting the optimized bound for \(r_Q(\delta_m)K_m\) from \cref{abstract-optimized-order-arbitrary-quadrature}. The factor \(F(m)\) is inherited from \cref{F-cond}.
\end{proof}

\begin{remark}\label{rem:nonhom-normalized-complexity-L1}
Let \(u(T)\) denote the exact solution of \cref{eq:non-homogeneous-time-independent}. If $E_{\operatorname{time}}(P)+\mathcal E_Q^{[0,T]}(r,\delta)\beta_P\leq \frac{\|u(T)\|}{2}$, then
\begin{equation}
\|u_{Q,P}^{\operatorname{MPF}}(T)\|\geq \|u(T)\|-E_{\operatorname{time}}(P)-\mathcal E_Q^{[0,T]}(r,\delta)\beta_P\geq \frac{\|u(T)\|}{2}.
\end{equation}
Consequently, \(\chi_{Q,P}(u(0),b)\leq 2\alpha_Q\beta_P/\|u(T)\|\), and the state query complexity is $\mathcal O\left(\frac{\alpha_Q\beta_P}{\|u(T)\|}\,r_Q(\delta)K_m\right)$.  If the time-quadrature rule satisfies the source-normalization estimate
\begin{equation}\label{eq:source-normalization-stability}
B_P=\sum_{p\in\mathcal I_P}|\omega_p|\,\|b(s_p)\|\leq C_{\operatorname{src}}\|b\|_{L^1[0,T]},
\end{equation}
with \(C_{\operatorname{src}} > 0\) independent of the target precision, then \(\beta_P\leq \|u(0)\|+C_{\operatorname{src}}\|b\|_{L^1[0,T]}\). In particular, when \(C_{\operatorname{src}}=\mathcal O(1)\), the normalized state-preparation cost has the usual non-homogeneous normalization factor \((\|u(0)\|+\|b\|_{L^1[0,T]})/\|u(T)\|\), up to the LCHS normalization \(\alpha_Q\) and the MPF simulation cost. This is the standard form in which normalized-state query complexity is presented in the literature.
\end{remark}

We see that no new commutator quantities arise. The source \(b(t)\) affects only source-state preparation, the time-quadrature error \(E_{\operatorname{time}}(P)\), and the normalization \(B_P\). The post-quadrature quantities \(R_Q\), \(\alpha_Q\), \(\Lambda_{m,Q}\), \(\Phi_{m,Q}^{\ast}\), \(r_Q\), \(\mu_Q\) are inherited from the homogeneous analysis.

\subsection{Time-Dependent Case}\label{sec:time-dependent-extension}
We now extend the analysis to the time-dependent case for the special case local and extensive Hamiltonians. For simplicity, we consider the homogeneous problem.
\begin{equation}\label{eq:td-homogeneous-problem}
\frac{d}{dt}u(t) = -A(t)u(t), \qquad A(t) = L(t) + iH(t),
\end{equation}
where, for every \(t \in [0,T]\), the matrices \(L(t)\) and \(H(t)\) are Hermitian and \(L(t) \succeq 0\). The extension to the non-homogeneous case follows as in \cref{sec:non-homogeneous-extension}. Let $U_A(s,t) := \mathcal T \exp (-\int_s^t A(\tau) d\tau )$ denote the non-unitary propagator from \(s\) to \(t\). As before, for \(k \in \mathbb R\), define $G_k(t) := H(t) + kL(t)$ and let
\begin{equation}
U_k(s,t) := \mathcal T \exp \left(-i\int_s^t G_k(\tau)d\tau \right).
\end{equation}
The approximation and truncation estimates from \cref{approx-error-time-indept,trunc-error} continue to apply after replacing \(e^{-AT}\) by \(U_A(0,T)\) and \(U_k(T)\) by \(U_k(0,T)\). In particular, for the kernel profile \(\vec{\theta}=(a,b,c,d)\), the same admissible choices of \(d\) and \(y_0\) from \cref{feasible-kernel-approx,feasible-kernel-profile-explicit} continue to hold. Motivated by \cref{prop:mizuta-time-dep-short}, define
\begin{equation}
C_m:=4^{2m+1}, \quad R_m(\delta):=\|\vec a\|_1\|\vec b\|_1\delta, \quad \Phi_m(k,N,\delta) := \mu_{k,p_0}^{2m+1}, \quad \mu_{R,p_0}:=\sup_{|k|\le R}\mu_{k,p_0},
\end{equation}
where \(\mu_{k,p_0}\) denotes the quantity \(\mu_{p_0}\) from \cref{prop:mizuta-time-dep-short} applied to the Hamiltonian \(G_k(t)\), \(p_0\) is defined as in \cref{prop:mizuta-time-dep-short}, and \(\|\vec a\|_1\) and \(\|\vec b\|_1\) are the $1$-norms of the MPF coefficient vectors discussed above. Throughout this subsection, we assume that the family \(\{G_k(t): |k|\le R\}\) admits a decomposition into \(q_R\)-local, \(g_R\)-extensive terms, and that the derivative bound in \cref{prop:mizuta-time-dep-short} holds with a constant \(f_R\), uniformly in \(k\) and in the initial time of the sub-interval.  Define
\begin{equation}\label{eq:td-admissible-step-size}
\Delta(R,N,\delta):=\min\left\{\frac{1}{8e^3p_0(2q_Rg_R+2\Gamma f_R)}, \frac{1}{8\mu_{R,p_0}} \right\}.
\end{equation}
We first record the pointwise short-time estimate obtained from \cref{prop:mizuta-time-dep-short}.

\begin{lem}\label{lem:td-pointwise-short-time-step}
Fix \(R>0\), \(m,N\ge1\), and \(\delta\in(0,1)\).  Let
\(0<\Delta\le \Delta(R,N,\delta)\).  For \(|k|\le R\) and
\(s\in[0,T-\Delta]\), let \(\widetilde U_{k,s}^{(2m)}\) denote the order-\(2m\)
time-dependent MPF approximation to \(U_k(s,s+\Delta)\).  Then
\begin{equation}\label{eq:td-pointwise-short-time-step}
\| \widetilde U_{k,s}^{(2m)}-U_k(s,s+\Delta) \| \le \|\vec a\|_1 C_m\Delta^{2m+1}\Phi_m(k,N,\delta) + R_m(\delta).
\end{equation}
\end{lem}

\begin{proof}
By assumption, the constants \(q_R\), \(g_R\), \(f_R\), and \(\mu_{R,p_0}\) control the estimate uniformly for every \(|k|\le R\) and every initial time \(s\).  Now apply \cref{prop:mizuta-time-dep-short} to the shifted Hamiltonian $\tau\mapsto G_k(s+\tau)$ for $\tau\in[0,\Delta]$.  
\end{proof}

The passage from short-time to long-time simulation is not identical to the time-independent case. This is because the exact propagator on \([0,T]\) is an ordered product of generally different interval propagators \(U_k(t_j,t_{j+1})\), rather than a power of a single fixed short-time propagator. Let \(t_j=j\Delta\), where \(\Delta=T/r\).  Then
\begin{equation}\label{eq:td-exact-ordered-product}
U_k(0,T)=\overset{\longleftarrow}{\prod_{0 \leq j \leq r-1}} U_k(t_j,t_{j+1}).
\end{equation}
The factors in \cref{eq:td-exact-ordered-product} are generally distinct, so one cannot use the
identity \(U_k(0,T)=U_k(0,\Delta)^r\).  We instead use the following telescoping estimate for
non-stationary products.

\begin{lem}\label{lem:td-nonstationary-telescoping}
Let \(r \ge 1\), and set \(t_j = j\Delta\) for \(j=0,\ldots,r\), where \(\Delta = T/r\). For fixed \(k\), define $U_{k,j}:=U_k(t_j,t_{j+1})$ and $\widetilde U_{k,j}:=\widetilde U_{k,t_j}^{(2m)}$. Set $U_k^{(2m)}(0,T) := \overset{\longleftarrow}{\prod_{0 \leq j \leq r-1}} \widetilde U_{k,j}$. If $\|\widetilde U_{k,j}-U_{k,j}\|\le \eta_k(\Delta)$ for some $\eta_k(\Delta) > 0$
for every \(j=0,\ldots,r-1\), then
\begin{equation}\label{eq:td-nonstationary-telescoping-bound}
\| U_k^{(2m)}(0,T)-U_k(0,T) \| \le r\eta_k(\Delta)\left(1+\eta_k(\Delta)\right)^{r-1}.
\end{equation}
\end{lem}

\begin{proof}
Using the ordered-product convention from \cref{sec:notation}, we have
\begin{align}
\overset{\longleftarrow}{\prod_{0 \leq j \leq r-1}}\widetilde U_{k,j} - \overset{\longleftarrow}{\prod_{0 \leq j \leq r-1}}U_{k,j} &= \sum_{\ell=0}^{r-1} \left( \overset{\longleftarrow}{\prod_{\ell+1 \leq j \leq r-1}}\widetilde U_{k,j} \right) (\widetilde U_{k,\ell}-U_{k,\ell}) \left( \overset{\longleftarrow}{\prod_{0 \leq j \leq \ell -1}}U_{k,j} \right).
\end{align}
Empty products are interpreted as the identity.  Since \(U_{k,j}\) is unitary, \(\|U_{k,j}\|=1\).  Moreover, we have
\begin{equation}
\|\widetilde U_{k,j}\| \le \|U_{k,j}\| + \|\widetilde U_{k,j}-U_{k,j}\| \le 1+\eta_k(\Delta).
\end{equation}
Taking norms gives
\begin{equation}
\| U_k^{(2m)}(0,T)-U_k(0,T) \| \le \sum_{\ell=0}^{r-1} \eta_k(\Delta)\left(1+\eta_k(\Delta)\right)^{r-1-\ell}.
\end{equation}
The right-hand side is bounded by \(r\eta_k(\Delta)(1+\eta_k(\Delta))^{r-1}\). This proves \cref{eq:td-nonstationary-telescoping-bound}.
\end{proof}

We now pass to the post-quadrature setting. The construction remains the same as before, but we recall the relevant details for completeness. Let \(Q=\{(k_i,w_i)\}_{i\in\mathcal I_Q}\) be a quadrature rule, and define \(v_i\), \(\alpha_Q\), and \(R_Q\) as in \cref{eq:generic-quadrature-operator}. Similarly, define \(\Lambda_{m,Q}\) as in \cref{eq:Lambda-Q}, and set \(\Phi_{m,Q}^{\ast} := \max_{i\in\mathcal I_Q}\Phi_m(k_i,N,\delta)\). The ideal and MPF-implemented quadrature operators, \(W_Q^{\operatorname{ideal}}\) and \(W_Q^{\operatorname{MPF}}\), respectively, are defined as above.  As in \cref{discrete-mpf-error}, we have
\begin{align}\label{eq:td-postquad-inner-error}
\| W_Q^{\operatorname{ideal}}-W_Q^{\operatorname{MPF}} \| &\le \left(1+\eta_Q(T/r)\right)^{r-1} \left( \|\vec a\|_1C_m \frac{T^{2m+1}}{r^{2m}} \Lambda_{m,Q} + r\alpha_QR_m(\delta) \right),
\end{align}
where \(\eta_Q(\Delta) := \max_{i \in \mathcal I_Q}\left(\|\vec a\|_1 C_m\Delta^{2m+1}\Phi_m(k_i,N,\delta) + R_m(\delta)\right)\). Define the quadrature error $ E_{\operatorname{quad}}(Q)$ as before. Combining the approximation, quadrature, and inner simulation errors gives
\begin{align}\label{eq:td-total-error-bound}
\left\| U_A(0,T)-W_Q^{\operatorname{MPF}} \right\| &\le E_{\operatorname{approx}}(y_0) + E_{\operatorname{quad}}(Q) \\ &\quad+ \left(1+\eta_Q(T/r)\right)^{r-1} \left( \|\vec a\|_1C_m \frac{T^{2m+1}}{r^{2m}} \Lambda_{m,Q} + r\alpha_QR_m(\delta) \right).
\end{align}
The implementation is the same as in \cref{implem}, except that the controlled operation for a fixed quadrature node \(k_i\) implements the ordered product $ U_{k_i}^{(2m)}(0,T) = \overset{\longleftarrow}{\prod_{0\leq j \leq r-1}} \widetilde U_{k_i,j}^{(2m)}$ Thus one controlled implementation of \(U_{k_i}^{(2m)}(0,T)\) still uses \(\mathcal O(rK_m)\) controlled second-order product-formula queries, provided each time-dependent second-order product-formula segment is available with the same oracle cost model as in \cref{prop:mizuta-time-dep-short}. The complexity of the algorithm follows from the preceding analysis in \cref{abs-rel} and \cref{spec-mizuta}. We do not repeat those details here.

\end{document}